\documentclass[a4paper,UKenglish,cleveref, autoref, thm-restate, algorithmic, algorithm]{lipics-v2021}
\usepackage{booktabs}
\usepackage{algorithm}
\usepackage{todonotes}
\usepackage{algorithm}
\usepackage{algpseudocode}
\usepackage{changepage}

\hideLIPIcs \nolinenumbers 
\graphicspath{{./figures/}}

\title{A Dynamic Piecewise-linear Geometric Index with Worst-case Guarantees} %

\author{Emil Toftegaard Gæde}{Technical University of Denmark, Denmark}{etoga@dtu.dk}{https://orcid.org/0009-0001-9462-6359}{}
\author{Ivor van der Hoog}{IT University of Copenhagen, Denmark}{ivva@itu.dk}{https://orcid.org/0009-0006-2624-0231}{}
\author{Eva Rotenberg}{IT University of Copenhagen, Denmark}{erot@itu.dk}{https://orcid.org/0000-0001-5853-7909}{}
\author{Tord Stordalen}{Technical University of Denmark}{}{}{}

\authorrunning{E. T. Gæde, I. van der Hoog, E. Rotenberg, and T. Stordalen} 

\Copyright{Emil T. Gæde, Ivor van der Hoog, Eva Rotenberg, and Tord Stordalen} 

\ccsdesc[100]{Theory of computation~Computational Geometry} 

\keywords{Algorithms Engineering, Data Structures, Indexing, Convex Hulls}  

\category{} 

\relatedversiondetails{Full version}{https://arxiv.org/abs/2503.05007} 

\supplementdetails[subcategory={Source Code},cite=impl_index]{Software}{https://github.com/Sgelet/DynamicLearnedIndex}
\supplementdetails[subcategory={Test Bed},cite=impl_bench]{Software}{https://github.com/Sgelet/LearnedIndexBench}

\usepackage{xspace}

\funding{This work was supported by the Carlsberg Foundation Fellowship CF21-0302 ``Graph Algorithms with Geometric Applications'', the VILLUM Foundation grant (VIL37507) ``Efficient Recomputations for Changeful Problems'', and the European Union's Horizon 2020 research and innovation programme under the Marie Sk\l{}odowska-Curie grant agreement No 899987. }

\graphicspath{ {figures/} }

\newcommand{\CH}{\textnormal{CH}}
\newcommand{\slope}{\textnormal{slope}}
\newcommand{\eps
}{\varepsilon}

\EventEditors{Anne Benoit, Haim Kaplan, Sebastian Wild, and Grzegorz Herman}
\EventNoEds{4}
\EventLongTitle{33rd Annual European Symposium on Algorithms (ESA 2025)}
\EventShortTitle{ESA 2025}
\EventAcronym{ESA}
\EventYear{2025}
\EventDate{September 15--17, 2025}
\EventLocation{Warsaw, Poland}
\EventLogo{}
\SeriesVolume{351}
\ArticleNo{62}

\begin{document}

\maketitle

\begin{abstract}
Indexing data is a fundamental problem in computer science. 
The input is a set $S$ of $n$ distinct integers from a universe $\mathcal{U}$.
Indexing queries take a value $q \in \mathcal{U}$ and return the \texttt{membership}, \texttt{predecessor} or \texttt{rank} of $q$ in $S$. 
A \texttt{range} query takes two values $q, r \in \mathcal{U}$ and returns the set $S \cap [q, r]$.

Recently, various papers study a special case where the the input data behaves in an approximately piece-wise linear way. 
Given the sorted (rank,value) pairs, 
and given some constant $\varepsilon$, one wants to maintain a small number of axis-disjoint line-segments such that, for each rank, the value is within $\pm \varepsilon$ of the corresponding line-segment. 
Ferragina and Vinciguerra (VLDB 2020) observe that this geometric problem is useful for solving indexing problems, particularly when the number of line-segments is small compared to the size of the  dataset. 

We study the dynamic version of this geometric problem. In the dynamic setting, inserting or deleting just one data point may cause up to three line-segments to be merged, or one line-segment to be split at most three-way. To determine and compute this, we use techniques from dynamic maintenance of convex hulls, and provide new algorithms with worst-case guarantees, including an $O(\log n)$ algorithm to compute a separating line between two non-intersecting convex hulls -- an operation previously missing from the literature.

We then use our fully-dynamic geometry-based subroutine in an indexing data structure, combining it with a natural hashing technique. 
The resulting indexing data structure has theoretically efficient worst-case guarantees in expectation.
We compare its practical performance to the solution of Ferragina and Vinciguerra, which was shown to perform better in certain structured settings [Sun, Zhou, Li VLDB 2023]. 
%
Our empirical analysis shows that our solution supports more efficient range queries in the special case where the update sequence contains many deletions.

\end{abstract}

\newpage

\section{Introduction}
We investigate the use of \emph{learned indices} for the design of \emph{dynamic indexing data structures}.

\subparagraph{Indexing data structures.}
An indexing data structure maintains a set $S$ of $n$ distinct integers from a universe $\mathcal{U}$.
Let $\textsc{RANK} : S \rightarrow [n]$ be the function mapping each $s \in S$ to its index in the sorted order of $S$.
The objective is to support the following indexing queries:

\begin{itemize}
    \item \texttt{member($q$)} returns \texttt{true} if $q \in S$.
    \item \texttt{predecessor($q$)} returns $\max \{ t \in S \mid t < q \}$. \hfill \textnormal{(We allow $q \notin S$.)}
    \item \texttt{rank($q$)} returns $\textsc{RANK}(\texttt{predecessor}(q)) + 1$. \hfill \textnormal{(We allow $q \notin S$.)}
\end{itemize}

\noindent
Additionally, we consider \emph{range queries}, where $k$ denotes the output size:
\begin{itemize}
    \item \texttt{range($q$, $t$)} returns $S \cap [q, t]$. \hfill \textnormal{(We allow $q, t \notin S$.)}
\end{itemize}

\noindent
Static indexing data structures fall into three broad categories: \emph{Tree-based solutions} store $S$ in a sorted array $A$, requiring no additional space but incurring logarithmic query costs. Tree traversals enable \texttt{predecessor}, \texttt{rank}, and \texttt{member} queries in $O(\log n)$ time and range queries in $O(\log n + k)$ time~\cite{athanassoulis2014bf, bender2000cache, puatracscu2006time, wang2018building}.
\emph{Map-based solutions} store $S$ in sorted order and maintain a hash map $H : S \rightarrow [n]$ mapping each element to its rank~\cite{chan1998bitmap, koudas2000space, pagh2004cuckoo}. This enables constant-time support for \texttt{member}, \texttt{predecessor}, and \texttt{rank} queries if queries are restricted to elements of $S$, and $O(k)$ time for \texttt{range} queries when both endpoints lie in $S$. The additional space is $O(n)$.
A third category use what are called \emph{learned indices}, a recently introduced term~\cite{kraska2018case, galakatos2019fiting, ferragina2020pgm, kipf2020radixspline, ding2020alex}. Given an integer parameter $\varepsilon$, a learned index is a function $h_\varepsilon : \mathcal{U} \rightarrow [0, n]$ such that
\[
    h_\varepsilon(q) \in [\texttt{rank}(q) - \varepsilon, \texttt{rank}(q) + \varepsilon].
\]
The function $h_\varepsilon$ is learned from $S$ and used to guide search in a sorted array $A$ storing $S$. Ferragina and Vinciguerra~\cite{ferragina2020pgm} interpret $h_\varepsilon$ geometrically: each $s \in S$ maps to a point $(\textsc{RANK}(s), s)$ in the plane, and $h_\varepsilon$ is learned as a piecewise-linear approximation to~$F_S$.

A notable instance of learned indices is the \emph{PGM index}~\cite{ferragina2020pgm}, where $h_\varepsilon$ is a $y$-monotone piecewise-linear function made of segments, with the property that each point in $F_S$ lies within an $\varepsilon$-wide horizontal strip around some segment. Let $|h_\varepsilon|$ denote the number of segments. The data structure supports indexing queries in $O(\varepsilon + \log |h_\varepsilon|)$ time and range queries in $O(\varepsilon + k + \log |h_\varepsilon|)$ time. They also show how to construct a PGM index in linear time, such that there exists no PGM index $h_\varepsilon'$ with $|h_\varepsilon|> \frac{3}{2}|h_\varepsilon'|$.
Ferragina and Vinciguerra argue that learned indices are the `best of both worlds' since:

\begin{itemize}
    \item the supported queries are as general as those supported by tree-based solutions, 
    \item the solution uses only $O(|h_\varepsilon|)$ additional space, and
    \item  $O(\eps + \log |h_\varepsilon|)$ is, for an appropriate choice of $\eps$, efficient in practice. 
\end{itemize}

\noindent
Their performance has been empirically benchmarked in several studies~\cite{ferragina2020learned, kraska2018case, wongkham2022updatable, wang2018building}.

\begin{figure}[h]
    \centering
    \includegraphics[width = \linewidth]{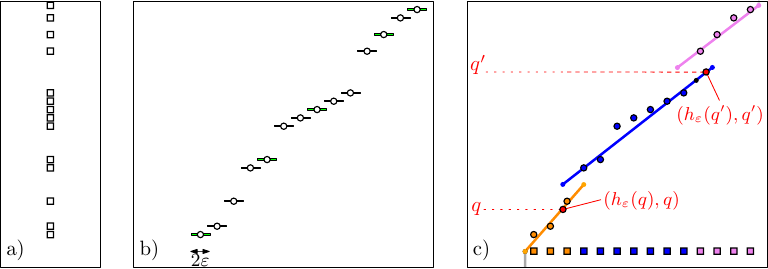}
    \caption{(a) a set of $n$ values $S$. (b)  $S$ corresponds to an $xy$-monotone point set $F_S$. (c) The PGM index computes a  $y$-monotone set of segments that starts and ends with a vertical halfline.      }
    \label{fig:PGM_example}
\end{figure}

\subparagraph{Dynamic indexing data structures.} 
Due to their fundamental role, dynamic indexing structures have received extensive theoretical and practical attention. 
When $S$ is dynamic, maintaining a sorted array becomes inefficient. Tree-based structures can be updated in $O(\log n)$ time with tree rotations. Map-based approaches allow constant-time \texttt{member} updates but are typically not extended to support other indexing queries.
Learned indices offer a promising direction by exploiting structural properties of $S$, akin to parametrised algorithms. However, just as parametrised algorithms, data structures based on learned indices are not always efficient: if $S$ lacks exploitable structure or access patterns are skewed, traditional indexing data structures are preferred~\cite{Sun23survey, marcus2020benchmark}. Experimental studies have examined properties of learned indices~\cite{ferragina2021performance,ferragina2020effective,marcus2020benchmark, Sun23survey}, in an effort to classify when they are appropriate to use. As a brief summary, learned indices pay a price in updates~\cite{marcus2020benchmark}, and traditional indices are preferred if $S$ or the access pattern is complex or skewed, or if concurrency is possible~\cite{Sun23survey}. If there is sufficient structure, both space usage and access times can benefit from learned indices~\cite{ferragina2020effective,ferragina2021performance}, subject to the strategy employed by the learned index. 

\subparagraph{Dynamic learned indices through the logarithmic method.}
The logarithmic method of Overmars~\cite{overmars1983design} provides an amortised way to maintain learned indices. $S$ is partitioned into $\lceil \log n \rceil$ buckets $B_i$, each of size $2^i$. Each bucket is either full or empty, and stores its contents in an array $A_i$ in sorted order and maintains a learned index $h_{\varepsilon}^i$ over $A_i$. 

Let the learned index $h_\varepsilon$ have a construction time of  $T(n)$. 
This data structure can be maintained insertion-only in amortised $O(T(n) \log n)$ time. An insertion inserts a new value into $B_0$. 
Let $j$ be the maximum integer such that all $B_i$ for $i \in [0, j - 1]$ are full.
This approach empties these buckets, fills $B_{j}$ in sorted order, and constructs $(A_j, h_{\varepsilon}^j)$ in $O(T(2^{j}))$ time. 
Whenever we delete some $s \in S$, this approach instead inserts a \emph{tombstone} $s^*$, which is a special copy of $s$. 
If an insertion fills a new bucket $B_{j}$, it first iterates over all elements. 
If $B_{j}$ contains both $s$ and $s^*$, it removes both elements. 
It then constructs $(A_{j}, h_{\varepsilon}^j)$ twice. Once on all `normal' values, and once on all tombstones in $B_{j}$.  
This way, deletions take the same time as insertions do. 
In this paper, we consider the following open question:

\begin{quote}
    ``Can a learned index be dynamically be maintained with worst-case guarantees?''
\end{quote}

\subparagraph{Intermezzo: computing a line cover.}
The interpretation of learned indices by Ferragina and Vinciguerra~\cite{ferragina2020pgm} translates to a geometric problem where the goal is to (approximately) cover a monotone set of two-dimensional points by a set of line segments. Using the logarithmic method, and the static  algorithm of O'Rourke~\cite{o1981line} to approximately points, they dynamically maintain an \emph{$\varepsilon$-cover}: a set of lines that are guaranteed to be within an $\varepsilon$ horizontal distance from each point.
We consider the problem of maintaining dynamic $\varepsilon$-covers to be an interesting geometric problem in its own right. 

\subparagraph{From learned indices to indexing data structures.}
Under the logarithmic method, indexing queries decompose naturally across the buckets:

\begin{itemize}
    \item For $\textnormal{\texttt{member}}(q)$, $q \in S$ if and only if there exists an $i \in [ \lceil \log n \rceil]$ with $q \in B_i$.
    \item For $\textnormal{\texttt{predecessor}}(q)$, the output is the maximum predecessor across $B_i$ for $i \in [ \lceil \log n \rceil]$.
    \item For $\textnormal{\texttt{rank}}(q)$, the rank is the sum of all ranks of $q$ in $B_i$ for $i \in [ \lceil \log n \rceil]$.
    \item For $\textnormal{\texttt{range}}(q, t)$, the reported range is the union of all ranges in $B_i$ for $i \in [ \lceil \log n \rceil]$.
\end{itemize}

\noindent
This way, indexing queries require only an additional factor $O(\log n)$ time. 
Indexing queries can be answered by combining queries to both the normal and tombstone structures. E.g., $\textnormal{\texttt{rank}}(q)$ is the rank of $q$ in the `normal' data structure minus the rank of $q$ in the tombstone structure.  This approach has two downsides:

\begin{itemize}
    \item  First, approach has an amortised update time.
    \item Second, this approach does not support output-sensitive range queries -- as there may be $O(n)$ values $s$ that (together with their tombstones $s^*$) lie in between a query pair $(q, t)$. 
\end{itemize}

\noindent
This leads to the following open question:

\begin{quote}
    ``Can a dynamic learned index be converted into an output-sensitive dynamic indexing data structure with worst-case guarantees?''
\end{quote}

\subparagraph{Contribution and organization.}
We propose maintaining a dynamic $\eps$-cover (and thereby a dynamic learned index) via dynamic convex hull techniques. 
Section~\ref{sec:covered} shows that deciding whether $S$ admits an $\varepsilon$-cover of complexity 1 is equivalent to convex hull intersection testing (Figure~\ref{fig:intersection_test}).
Section~\ref{sec:robust} shows a robust algorithm to compute the intersection between two convex hulls in $O(\log n)$ time. We adapt our algorithm to output a separating line (which is a learned index of complexity $1$) in the negative case. 
Section~\ref{sec:learned} combines these with dynamic convex hull data structures to yield a dynamic learned index worst-case $O(\log^2 n)$ update time.
We empirically compare our learned index to the PGM index from~\cite{ferragina2020pgm}.

Section~\ref{sec:indexing} introduces a novel hashing-based approach to convert learned indices into dynamic indexing data structures, using $O(\varepsilon^{-1})$ additional expected overhead.
We compare our dynamic indexing structure to the amortised PGM index of~\cite{ferragina2020pgm} in terms of update time and index complexity. We do not benchmark against traditional indexing data structures 
-- since the relation between learned and traditional indices is previously studied~\cite{Sun23survey}. Instead, our goal is to push the theoretical limits and worst-case guarantees of learned indices.

\begin{figure}[h]
    \centering
    \includegraphics[width = \linewidth]{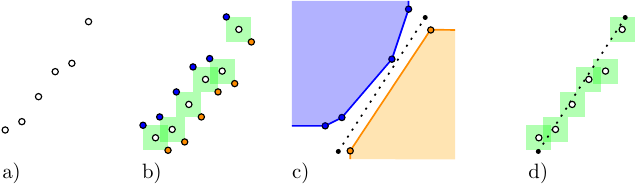}
    \caption{For any set $F_S$, we construct two convex hulls. We prove that there exists a segment $\ell$ within $L_\infty$-distance $\eps$ of all points in $F_S$ if and only if these hulls do not intersect. We adapt the convex hull intersection testing algorithm to find $\ell$ whenever these hulls are disjoint. }
    \label{fig:intersection_test}
\end{figure}

\newpage
\section{Preliminaries}

The input is a dynamic set $S$ of $n$ distinct positive integers from some universe $\mathcal{U}$.
For $a, b \in \mathbb{Z}$ with $a \leq b$, we define $S[a, b]$ as the set $S \cap [a, b]$.
We denote by $F_S$ the two-dimensional point set obtained by mapping each $s \in S$
to $(\textsc{RANK}(S), s)$. Throughout this paper, we distinguish between positions and strict positions. E.g., lying above or strictly above a line.

\begin{definition}[\cite{ferragina2020pgm}]
Let $\eps$ be a positive integer.
A \emph{PGM index} $h_\varepsilon$ of $S$ is defined as a $y$-monotone set of segments that together cover the $y$-axis.
We regard $h_\varepsilon$ as a map from $y$-coordinates to $x$-coordinates and require that for all $q \in \mathcal{U}$, $h_\varepsilon(q) \in [\textnormal{\texttt{rank}}(q) - \eps, \textnormal{\texttt{rank}}(q) + \eps]$.
\end{definition}

\noindent
Ferragina and Vinciguerra~\cite[Lemma~1]{ferragina2020pgm} wrongfully claim an $O(n)$-time algorithm to compute a minimum complexity PGM index $h_\varepsilon$. They invoke a streaming algorithm by O'Rourke~\cite{o1981line} for fitting straight lines through data ranges.  We show that this algorithm outputs a PGM index $h_\varepsilon$ such that there exists no PGM index $h_\varepsilon'$ with $|h_\varepsilon| > \frac{3}{2}h_\varepsilon'$ (see Appendix~\ref{app:orourke}). 
Their algorithm restricts $S$ to contain no duplicates. 
We assume the same setting and compute something slightly different as we define an $\eps$-cover instead:

 \begin{definition}
     Let $\eps$ be a positive integer. We define an $\eps$-cover $f$ of $S$ as a set of vertically separated segments with slope at least $1$ where all $(r, s) \in F_{S}$ are within $L_\infty$-distance $\eps$ of $f$. 
 \end{definition}

\noindent
 An $\eps$-cover has a functionality and complexity similar to a learned index:

 \begin{observation}
    \label{obs:cover_to_index}
     Let $f$ be an $\eps$-cover and $Q$ be a horizontal line with height $q \in [ \min S, \max S]$. Let $(s, t)$ be the segment in $f$ closest to $q$. Then $(line(s, t) \cap Q).x \in [\textnormal{\texttt{rank}}(q) - 2\eps, \textnormal{\texttt{rank}}(q) + 2\eps]$.
 \end{observation}

 \begin{observation}
\label{obs:count}
For fixed $\eps$, let $k$ denote the minimum complexity of any PGM index of $S$. 
If $f$ is an $\eps$-cover of $S$ of minimum complexity, then $f$ contains at most $k - 2$ edges.
\end{observation}
 
\begin{definition}
    For any fixed $\eps$-cover $f$ of $S$, we define $\Lambda(f)$ as the set of pairwise interior-disjoint one-dimensional intervals that correspond to the $y$-coordinates of segments in $f$.
\end{definition}

\subparagraph{Dynamic convex hulls.}
We dynamically maintain an $\eps$-cover $f$ of $S$ of approximately minimum complexity. To this end, we use a result by Overmars and van Leeuwen~\cite{overmars1981maintenance} to dynamically maintain for all $[a, b] \in \Lambda(f)$ the convex hull of $F_{S[a, b]}$. 
For any point set $F$, denote by $CH(F)$ their convex hull.
The data structure in~\cite{overmars1981maintenance} is a balanced binary tree over $F$, which at its root maintains a balanced binary tree over the edges $CH(F)$ in their cyclical ordering.  It uses $O(n)$ space and has worst-case $O(\log^2 n)$ update time.

\subparagraph{Rank-based convex hulls.}
For any update in $S$, up to $n$ values in $F_S$ may change their $x$-coordinate. This complicates the maintenance of a dynamic data structure over $F$. 
Gæde, Gørtz, van Der Hoog, Krogh, and Rotenberg~\cite{Gaede2024simple} observe that all algorithmic logic in~\cite{overmars1981maintenance} requires only the \emph{relative} $x$-coordinates between points. 
They adapt~\cite{overmars1981maintenance} to give an efficient and robust implementation of what they call \emph{a rank-based convex hull} data structure $T(S)$ with $O(\log^2 n)$ update time. For ease of exposition, we overly simplify their functionality: 

For each $[a, b] \in \Lambda(f)$, we store $S[a, b]$ in $T(S[a, b])$. $T(S[a, b])$ maintains a balanced binary tree $\gamma(S[a, b])$ storing the edges of $\CH(F_{S[a, b]})$ in their cyclical ordering. 
We use this data structure as a black box, using the following  functions that take at most $O(\log^2 n)$ time: 

\begin{itemize}
    \item $T(S[a, b])$.\texttt{get\_hull()} returns the tree $\gamma(S[a, b])$.
    \item $T(S[a, b])$.\texttt{split($v$)} returns,  for $v \in [a, b]$, $T(S[a, v])$ and $T(S[v, b])$.
    \item $T(S[a, b])$.\texttt{split($T([S[b, c])$)}  returns $T(S[a, c])$. 
    \item $T(S[a, b])$.\texttt{update($v$)}  updates, for $v \in [a, b]$, the set $S$ (deleting or inserting $v$). 
\end{itemize}

\newpage
\section{Testing whether a set can be $\eps$-covered by a single segment}
\label{sec:covered}

We consider the following subproblem: given a parameter $\varepsilon$, a set $S$ of $n$ distinct integers, and the edges of $\CH(F_S)$ stored in a balanced binary tree, can we compute in $O(\log n)$ time whether there exists an $\varepsilon$-cover $f$ of complexity~1? 
Formally, we seek a line $\ell$ of slope at least 1 such that all points in $F_S$ lie within $L_\infty$-distance $\varepsilon$ of $\ell$.
Let $L$ (resp.\ $U$) denote the set obtained by shifting each $p \in F_S$ downwards and rightwards (resp.\ upwards and leftwards) by $\varepsilon$, and adding the point $(\infty, -\infty)$ (resp.\ $(-\infty, \infty)$).

\begin{lemma}
    Let $\ell$ be a line of slope at least 1. Then all points in $F_S$ lie within $L_\infty$-distance $\varepsilon$ of $\ell$ if and only if $\ell$ lies below all points in $U$ and above all points in $L$.
\end{lemma}

\begin{proof} 
    Any line with positive slope lies above $(\infty, - \infty)$ and below $(-\infty, \infty)$.
    Consider a point $p \in F_S$ and the two corresponding points $l \in L$ and $u \in U$ and denote by $C$ an axis-aligned square of radius $\eps$ centred at $p$.
    If $\ell$ lies below $l$ then all points on $\ell$ left of $l$ lie below $C$. 
    If $\ell$ lies above $u$ then all points on $\ell$ right of $u$ lie above $C$.     If $\ell$ lies above $l$ and below $u$ then because $\ell$ has positive slope, it must intersect $C$. The statement follows. 
\end{proof}

\begin{corollary}
  $\ell$ is an $\eps$-cover of $S$ iff it has a slope $\geq 1$ and separates $CH(L)$ from $CH(U)$. 
\end{corollary}

\noindent
Given $\CH(F_S)$, we can extract $\CH(L)$ and $\CH(U)$ in $O(\log n)$ time. Chazelle and Dobkin~\cite[Section 4.2]{chazelle1980detection} remark that, in the negative case, convex hull intersection testing can be modified to produce a separating line. In our setting, the hulls consist of segments with slope at least 1, and any such separator corresponds to an $\varepsilon$-cover. Thus, our problem reduces to the classical convex hull intersection problem and we are seemingly done.

However, the history of convex hull intersection testing is long and intricate. Both Chazelle and Dobkin~\cite{chazelle1980detection} and Dobkin and Kirkpatrick~\cite{DOBKIN1983} independently proposed the first $O(\log n)$-time algorithms. In 1987, Chazelle and Dobkin~\cite{Chazelle1987intersection} presented a more detailed description of their method. Dobkin and Kirkpatrick revisited their own work in 1990~\cite{Dobkin1990}, proposing a unified $O(\log^2 n)$-time algorithm for polyhedron intersection, which O'Rourke later identified as incorrect~\cite{ORourke1998Computational}. He corrected the argument and provided a $C$-implementation.
Further work by Dobkin and Souvaine~\cite{DOBKIN1991} noted that earlier implementations lacked robustness. More recently, Barba and Langerman~\cite{Barba2015optimal} observed that the community still lacked a complete, robust algorithm for polyhedral intersection. They proposed an alternative $O(\log n)$ algorithm based on polar transformations. Walther's master's thesis~\cite{WaltherThesis}, supervised by Afshani and Brodal, implemented both this and earlier methods, but the source code is no longer available.

This 35-year history highlights the complexity and subtlety of convex hull intersection testing. Despite its history, no robust and modern $O(\log n)$-time implementation is available. Moreover, no published algorithm explicitly computes a separating line in the negative case.

\subparagraph{Contribution.}
In Appendix~\ref{app:intersection_testing}, we present a robust $O(\log n)$-time algorithm for convex hull intersection testing. Our algorithm is specialised to convex hulls composed of positively sloped segments and including the points $(\infty, -\infty)$ and $(-\infty, \infty)$. We formally prove its correctness and adapt it to compute a separating line in the negative case, thereby constructing an $\varepsilon$-cover of complexity 1 when it exists.
The later adaption and its analysis are nontrivial, and arguably (partly) fill a gap in the existing literature on convex hull intersection testing.

\begin{restatable}{theorem}{separation}
    \label{thm:separation}
Let $A$ and $B$ be convex chains of edges with slope at least $1$, stored in a balanced binary tree on their left-to-right order.
There exists an $O(\log n)$ time to decide whether there exists a line that separates $A$ and $B$. This algorithm requires only orientation-testing for ordered triangles and can output a separating line whenever it exists. 
\end{restatable}

\newpage

\section{Robustness}
\label{sec:robust}

A geometric predicate is a function that takes geometric objects and outputs a Boolean.
Our algorithms compute geometric predicates and use their output to branch along a decision tree. 
In $F_S$, consecutive points differ in $x$-coordinate by exactly $1$ whilst their $y$-coordinate may wildly vary. 
Consequently, any segment that $\eps$-covers a subsequence of $F_S$ is quite steep. 
This quickly leads to rounding errors when computing geometric predicates, which in turn creates robustness errors. To illustrate our point, we discuss one of our main algorithms:

\texttt{intersection\_test}  (Algorithm~\ref{alg:intersection_test}) which determines whether an upper quarter convex hull $CH(A)$ and a lower quarter convex hull $CH(B)$ intersect.
We receive these hulls as two trees. Our algorithm computes a few geometric predicates given the edges $\alpha$ and $\beta$ stored at their respective roots.  
Given $(\alpha, \beta)$, we either conclude that $CH(A)$ and $CH(B)$ intersect, or, that all edges succeeding (or preceding) $\alpha$ (or $\beta$) cannot intersect the other convex hull.
Based on the Boolean output, our algorithm then branches into a subtree of $\alpha$ (or $\beta$). This way, we verify whether $CH(A)$ and $CH(B)$ intersect in logarithmic time. 
Rounding causes these predicates to output a wrong conclusion, and our algorithm may branch into a subtree containing edges of $CH(A)$ that are guaranteed to not intersect $CH(B)$. 
Our algorithm then wrongfully concludes that there exists a line $\ell$ separating $CH(A)$ and $CH(B)$. Subsequent algorithms then exhibit undefined behaviour when they attempt to compute this line.

\subparagraph{Geometric predicates.} Our algorithms use on three predicates for their decision making:
\begin{itemize}
    \item \texttt{slope}. Given positive segments $(\alpha, \beta)$, output whether $\slope(\alpha) < \slope(\beta)$.
    \item \texttt{lies\_right}. Given two positive segments $\alpha$ and $\beta$ with different slopes, output whether the first vertex of $\beta$ lies right of $line(\alpha) \cap line(\beta)$.
    \item \texttt{wedge}. Consider a pair of positive segments $(\alpha, \gamma)$ that share a vertex and define $W$ as the cone formed by their supporting halflines containing $(\infty, -\infty)$. 
    Given a positive segment $\beta$ outside of $W$, output whether $line(\beta)$ intersects $W$. 
\end{itemize}

The segments are given by points with integer coordinates. 
The slopes of these segments (and thereby any representation of their supporting line) are often not integer. 
A naive way to compute these predicates is to represent slopes using \emph{doubles}. However, this is both computationally slow and prone to rounding errors (and thus, robustness errors). 

If we insist on correct output, one can use an \emph{algebraic type} instead. This type represents values using algebraic expressions. E.g., the slope of a positive segment $(a, b)$ is the quotient: $\frac{b.y - a.y}{b.x - a.x}$ and so, in our case, it can be represented as a pair of integers. 
Algebraic types can subsequently be accurately compared to each other. Indeed, if we want to verify whether $\frac{s}{t} < \frac{q}{r}$ we may robustly verify whether $sr < qt$ using only integers.
Exact (algebraic type) comparisons are frequently implemented, and present in the CGAL CORE library~\cite{fabri2000design}. 

However, exact comparisons are expensive.
Our implementation of \texttt{slope} requires two integer multiplications, which is still relatively efficient. 
Evaluating more complex expressions requires too much time. As a rule of thumb, we want to avoid compounding algebraic types to maintain efficiency. Na\"{i}vely,  \texttt{lies\_right} compounds two quotients and \texttt{wedge} compounds three. 
We give robust implementations of these functions by invoking three subfunctions. These compare slopes, or whether a point lies above or below a supporting halfplane: 

\begin{align*}
&\texttt{slope}((a, b), (c, d)) &:=&  (b.y - a.y) \cdot (d.x - c.x) < (d.y - c.y) \cdot (b.x - a.x) \\
&\texttt{above\_line}( (a, b), c) &:=&  (b.x - a.x)(c.y - b.y) - (c.x - b.x)(b.y - a.y) \geq 0 \\
&\texttt{below\_line}( (a, b), c) &:=&  (b.x - a.x)(c.y - b.y) - (c.x - b.x)(b.y - a.y) \leq 0
\end{align*}

\noindent
We can create $\texttt{lies\_right}$ from our robust predicates (see Figure~\ref{fig:robust} (a)):  
\begin{lemma}
    Let $\alpha = (a, b)$ and $\beta = (c, d)$ be two positive segments of different slope. Then:
    \begin{align*}
    \textnormal{\texttt{lies\_right}}(\alpha, \beta) &=  \bigl( \textnormal{\texttt{slope}}((a, b), (c, d)) ==  \textnormal{\texttt{above\_line}}( (a, b), c)   \bigr) \\ 
    &\vee  \bigl( \textnormal{\texttt{slope}}((c, d), (b, c) ) == \textnormal{\texttt{below\_line}}( (a, b), c)  \bigr)
    \end{align*}
\end{lemma}

\begin{proof}
    Suppose that $\slope(\alpha) < \slope(\beta)$. 
    Then $c$ lies right of $line(\alpha) \cap line(\beta)$ if and only if $c$ lies above the halfplane bounded from above by $line((a, b))$.
    That happens if and only if $(a, b, c)$ are collinear or make a counter-clockwise turn.
    This in turn occurs if and only if the determinant if the matrix    
    $\begin{vmatrix}
 (b.x - a.x) & (c.x - b.x) \\
 (b.y - a.y) & (c.y - b.y)
\end{vmatrix}$ is zero or more.   If $slope(\alpha) > slope(\beta)$ the determinant must be negative instead.   
\end{proof}

 Similarly, we can create \texttt{wedge} from our robust predicates. We note for the reader that explain our equations in words in the proof of the lemma:

\begin{lemma}[Figure~\ref{fig:robust} (b)]
If $\alpha = (a, b)$, $\gamma = (b, c)$ and $\beta = (d, e)$ be three segments of positive slope where $W = \overleftarrow{\alpha} \cup \overrightarrow{\gamma}$ bounds a convex area containing $(\infty, -\infty)$.  Then 

\begin{align*}
   &\textnormal{\texttt{wedge}}(\alpha, \gamma, \beta) :=  \\
   & \bigl( \textnormal{\texttt{below\_line}}( (b, c), d) \wedge  \bigl( \textnormal{\texttt{above\_line}}( (d, e), b) \vee \textnormal{\texttt{slope}}( (a, b), (d, e)  \bigr) \bigr) \vee   \\
   & \bigl(\textnormal{\texttt{below\_line}}( (a, b), e) \wedge  \bigl( \textnormal{\texttt{above\_line}}( (d, e), b)  \vee  \textnormal{\texttt{slope}}( (d, e), (b, c) )  \bigr) \bigr) \vee  \\
& \bigl( \neg \textnormal{\texttt{below\_line}}( (a, b), e) \wedge  \neg \textnormal{\texttt{below\_line}}( (b, c), d)  \wedge \bigl( \textnormal{\texttt{slope}}( (a,b), (d, e) ) \vee  \textnormal{\texttt{slope}}( (b,c), (d, e) ) \bigr) \bigr)
\end{align*} 
\end{lemma}

\begin{proof}
    The predicate is a case distinction of three mutually exclusive cases. 
    
    If the first vertex of $\beta$ lies below the supporting line of $\gamma$ then $line(\beta)$ intersects $W$ if and only if it intersects $\overleftarrow{\alpha}$. This happens if and only if one of two conditions hold: either $b$ lies below the supporting line of $b$, or,  $\slope(\alpha) < \slope(\beta)$.

    If the second vertex of $\beta$ lies below $line(\alpha)$ then the argument is symmetric. 

    If neither of those cases apply then both endpoints of $\beta$ must lie in the open green area. In this case, whenever $\slope(\alpha) < \slope(\beta)$, the supporting line of $\beta$ always intersects $W$. 
    Whenever $\slope(\beta) < \slope(\gamma)$, the supporting line of $\beta$ always intersects $W$.
    Whenever $\slope(\alpha) \geq \slope(\beta) \geq \slope(\gamma)$, the supporting line of $\beta$ cannot intersect $W$. 
\end{proof}

\begin{figure}[h]
    \centering
    \includegraphics[]{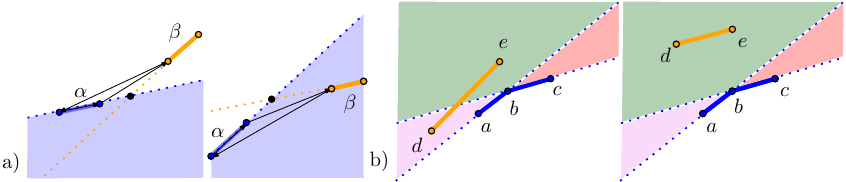}
    \caption{
    (a) We reduce testing whether the first vertex of $\beta$ lies right of the intersection point to comparing slopes and the orientation of a triangle.
    (b) If $d$ lies below the halfplane of $line(b, c)$ then $line((d, e))$ intersects the wedge if and only if $b$ lies below $line((d, e))$. 
    }
    \label{fig:robust}
\end{figure}

\section{Dynamically maintaining a learned index}
\label{sec:learned}

We dynamically maintain a learned index $h_\varepsilon$ of $S$ by maintaining an $\eps$-cover $f$ of $S$. We guarantee that there exists no $\eps$-cover $f'$ of $S$ with $|f| > \frac{3}{2} |f'|$. By Observation~\ref{obs:cover_to_index}, we obtain a learned index $h_\varepsilon$. 
By Observation~\ref{obs:count}, there exists no PGM index $h_\varepsilon$ where $|f| > \frac{3}{2}|h_\varepsilon|$.

To maintain $f$, we maintain a  balanced binary tree $B(f)$ over $\Lambda(f)$.
Additionally, for each $[a, b] \in \Lambda(f)$, we maintain  a \emph{rank-based convex hull} $T(S[a, b])$ of $S[a, b]$ as described in~\cite{Gaede2024simple}. 
We note that we store all segments in $f$ using \emph{relative} $x$-coordinates.
That is, we assume for all $[a, b] \in \Lambda(f)$ that the rank of the first element in $S[a, b]$ is zero. We may then use $B(f)$ to `offset' each line to compute the actual coordinates in rank-space.

\begin{theorem}
    \label{eps:eps_cover}
   We can dynamically maintain an $\eps$-cover $f$ of $S$ in $O(\log^2 n)$ worst-case time.
    We guarantee that there exists no $\eps$-cover $f'$ of $S$ where $|f| > \frac{3}{2}|f'|$.
\end{theorem}

\begin{proof}
    The proof is illustrated by Figure~\ref{fig:merge}.
    For any $s,t \in \mathbb{Z}$ with $s \leq t$, we say that $S[s, t]$ is \emph{blocked} if there exists no $\eps$-cover of $S[s, t]$ of size $1$. 
    We maintain an $\eps$-cover $f$ where for all consecutive intervals $[a, b], [c, d] \in \Lambda(f)$, $S[a, d]$ is blocked. Thereby, $|f| \leq \frac{3}{2}|f'|$ for any $\eps$-cover $f'$ of $S$ (we give a proof of this fact in Appendix~\ref{app:orourke}). 

        We consider inserting a value $s$ into $S$; deletions are handled analogously.  
        We query $B(f)$ in $O(\log n)$ time for an interval $[a, b]$ that contains $s$. If no such interval exists, set $[a, b] = [s, s]$. 
        We search $T(S[a, b])$ and test whether $s \in S$. If so, we reject the update. 

        Otherwise, we remove $[a, b]$ from $\Lambda(f)$ and insert the intervals $([a, s], [s, s], [s, b]$). We obtain $T(S[a, s])$, $T(S[s, s])$ and $T(S[s, b])$ through the split operation. 

        Let $([w, x], [y,z], [a, s], [s, s], [s, b], [c, d], [e, f])$ be consecutive intervals in $\Lambda(f)$ and denote $I  = ([y, z], [a, s], [s, s], [s, b], [c, d])$ (see Figure~\ref{fig:merge} (c) ). 
        For each $(s, t) \in I$, we have access to $T(S[s, t])$.  For any consecutive pair $([s, t], [q, r] )$ in $I$, we may join the trees $T(S[s, t])$ and $T(S[q, r])$ in $O(\log^2 n)$ time to obtain $T([s, r])$.
        We invoke $T([s, r])$.\texttt{get\_hull()} and apply Theorem~\ref{thm:separation} to test in $O(\log^2 n)$ total time whether $S[s, r]$ is \emph{blocked}. If it is not,  we replace $[s, t]$ and $[q, r]$ by $[s, r]$. Otherwise, we keep $T(S[s, r])$ and a complexity-1 $\eps$-cover of $S[s, r]$.

        By recursively merging pairs in $I$, we obtain in $O(\log^2 n)$ time a sequence $I'$ of intervals $( [y, \beta], \ldots, [\gamma, d])$ where consecutive intervals are blocked. 
        Since $[y, z] \subseteq [y, \beta]$, $( [w, x], [y, \beta])$ is blocked. Similarly, $([\gamma, d], [e, f])$ must be blocked.
        We remove the line segments corresponding to $I$ from $f$ and replace them with line segments derived from $I'$ in constant time. As a result, we maintain our $\eps$-cover $f$ and our data structure in $O(\log^2 n)$ total time. 
  \end{proof}

\begin{figure}[h]
    \centering
    \includegraphics[]{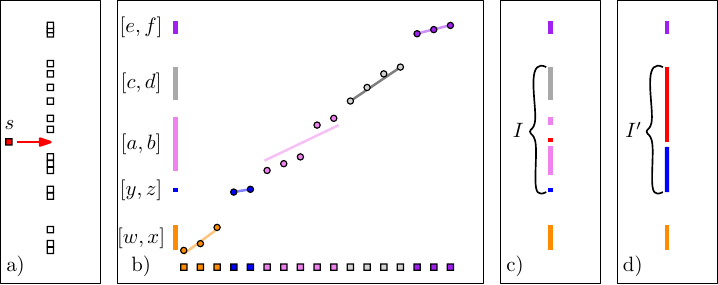}
    \caption{(a) Let $S$ be a set of values and let us insert $s$. (b) We consider our $\eps$-cover $f$ and five consecutive intervals in $\Lambda(f)$.
    (c) We create seven intervals by splitting $[a, b]$ on $s$. 
    (d) By recursively merging intervals in $I$, we obtain a set of intervals $I'$ where consecutive intervals are blocked.     
    }
    \label{fig:merge}
\end{figure}

\newpage

\section{From an $\eps$-cover to an indexing structure}
\label{sec:indexing}

A  learned index $h_\varepsilon$ does not immediately support indexing and range queries.
We obtain an indexing structure by combining $h_\varepsilon$ with a hash map $H$. 
Combining learned models with hash maps is not new~\cite{ lin2023learning, torralba200880, wang2015learning}  and this technique has even been applied to learned indexing~\cite{kraska2018case}. 
The core idea is to store $S$ in an unordered vector $A$ and maintain a Hash map $H : \mathbb{Z} \mapsto [n]$. 
Given some $q \in \mathcal{U}$, the learned function then produces a value $v$  such that $A[h(v)]$ is `close' to $q$. 
It is compelling to create $H$ such that $A[ H(h(q))]$ is (approximately) the predecessor of $q$. However, dynamically, this approach fails for the same reason that storing each $s \in S$ at $A[\textnormal{\texttt{rank}}(s)]$ fails. Since the ranks of elements in $S$ are constantly changing, we build a hash map using the parts of $S$ that remain constant: the values.

\subparagraph{Our data structure.}
In Appendix~\ref{app:indexing}, we define a data structure independent of the learned index $h_\varepsilon$ (we illustrate our approach in Figure~\ref{fig:approach} (a) + (b)). A \emph{page} $p$ is an integer with a vector that stores all $s \in S$ where  $\lfloor \frac{s}{\eps} \rfloor = p$, in order. 
We store all non-empty pages $P$ in an unordered vector $A$. We maintain a hash map $H : P \rightarrow [|A|]$, where $A[H(p)]$ contains the page $p$.  
We additionally maintain a doubly linked list over all pages in $P$, arranged in sorted order. 

\subparagraph{Our queries.}
We restrict our learned index $h_\eps$ to a \emph{vertical $\eps$-cover}. I.e., $h_\eps$ is a $y$-monotone collection of line segments such that for all points $p \in F_S$, a vertical line segment of height $2\eps$ centred at $p$ intersects a segment in $h_\eps$.
We compute $h_\eps$ oblivious of our paging structure. 

Given $q \in \mathcal{U}$, we project $q$ onto $h_\eps$ (Figure~\ref{fig:approach} (d)).  We project to the $x$-axis, floor the value, and project back to $h_\eps$. 
We prove that the resulting $y$-value corresponds to the page $p$ containing \texttt{predecessor}($q$). 
This way, we answer \texttt{predecessor} using $O(\eps + \log |h_\eps|)$ time.

As a result, we dynamically maintain a learned index $h_\eps$ and a data structure that updates in $O(\eps + \log^2 n)$ and supports indexing queries in $O(\eps + \log |h_\eps|)$ expected time (Theorem~\ref{thm:main}). 

\begin{figure}[h]
    \centering
    \includegraphics[width = 0.95\linewidth]{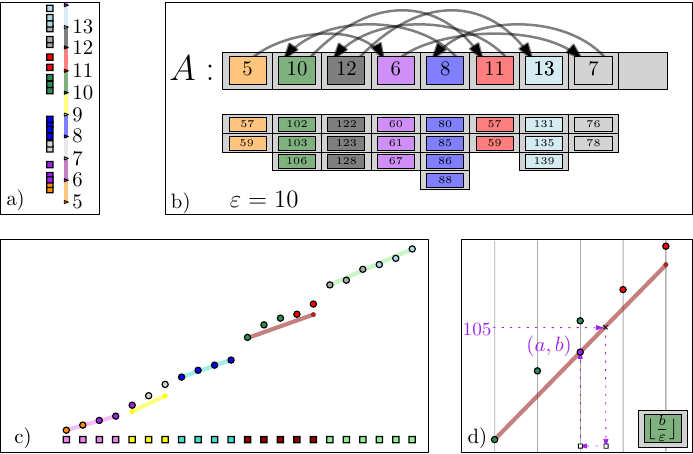}
    \caption{An illustration of our approach in Appendix~\ref{app:indexing}.}
    \label{fig:approach}
\end{figure}

\begin{theorem}\label{thm:main}
    For any $\eps$, there exists a data structure to dynamically maintain a vertical $\eps$-cover $F$ of a dynamic set of distinct integers in $O(\eps + \log^2 n)$ time. 
    We guarantee that there exists no vertical $\eps$-cover $F'$ with $|F| > \frac{3}{2}|F'|$. The  
    data structure supports indexing queries in $O(\eps + \log |F|)$ expected time and range queries in additional $O(k)$ time where $k$ is the output size. 
\end{theorem}

\newpage
\section{Experiments}
\label{sec:experiments}

Our implementation is written in \texttt{C++} and made publicly available~\cite{impl_index}. We compare to the \texttt{C++} implementation in~\cite{ferragina2020pgm}, which uses a PGM index under the logarithmic method. 
The experiments were conducted on a machine with a 4.2GHz AMD Ryzen 7 7800X3D and 128GB memory. Our test bench is available~\cite{impl_bench}, and can replicate experiments, generate synthetic data, and produce plots.
As input we consider two synthetic data sets and two real world data sets. Three contain data of geometric nature, with one of random nature to align with precedent. Each set consists of unique 8 byte integers in randomly shuffled order. In Appendix~\ref{app:experiments}, we showcase additional experiments on other datasets.
\begin{itemize}
    \item \textbf{LINES} is a synthetic data set of 5M integers that, in rank space, produces 5 lines of exponentially increasing slope. This set models the ideal scenario for a PGM index. 
    \item \textbf{LONGITUDE} is a real world data set that contains the longitudes of roughly 246M points of interest from OpenStreetMap, over the region of Italy. 
    This data is thereby inherently of geometric nature.
    This data set was used in both~\cite{ferragina2020pgm}  and \cite{kipf2019sosdbenchmarklearnedindexes}. We follow~\cite{ferragina2020pgm} and convert the data to integers by removing the decimal point from the raw longitudes.
    \item \textbf{UNIF} originates from~\cite{ferragina2020pgm}. It is a synthetic data set, containing a uniform random sample of 50M integers from $(0,10^{11})$. We adapt this data set to our dynamic setting.
    \item \textbf{DRIFTER} is a real world data set, containing roughly 1.7M steps of accumulated distance travelled by ocean drifters tracked through GPS~\cite{conradi2023}.
\end{itemize}
\subparagraph{Measurements.}
We compare the quality of the learned indices based the complexity of $h_\varepsilon$, in a dynamic setting. We use the same choice of $\eps = 64$ as in~\cite{ferragina2020pgm} across our experiments.
For performance of the indexing structures, we measure their time per operation in a dynamic scenarios with a range of query to update ratios. 
We note that logarithmic PGM is by default equipped with an optimisation that avoids building a PGM for data below a certain size. In this case, it instead only uses an underlying sorted array without additional search structure. 
In order to properly compare the performances, this optimisation has been disabled. 

\subsection{The learned index complexity}

Figure~\ref{fig:res_linecount} presents the complexity of the learned indices, measured by the number of line segments maintained during random-order insertions. The behaviour differs between datasets. For the geometric \textsc{Lines} and \textsc{Longitude} datasets, the dynamic and logarithmic PGMs initially perform similarly. As the data grows, the performance of the logarithmic PGM degrades and its jagged progression reflects its logarithmic partitioning.

On the highly structured synthetic \textsc{Lines} data, the logarithmic method consistently retains more segments than necessary. It misses the optimal line count by a wide margin due to its fragmentation across $O(\log n)$ buckets. A similar pattern appears in the \textsc{Longitude} dataset, where the logarithmic PGM maintains roughly 50 percent more segments than our solution. We note that precisely on these structured data sets, the complexity of the learned indices is $o(n)$. I.e., precisely here one also expects improvements in query time.  

\begin{figure}[hbt]
    \centering
    \includegraphics[width=0.49\linewidth]{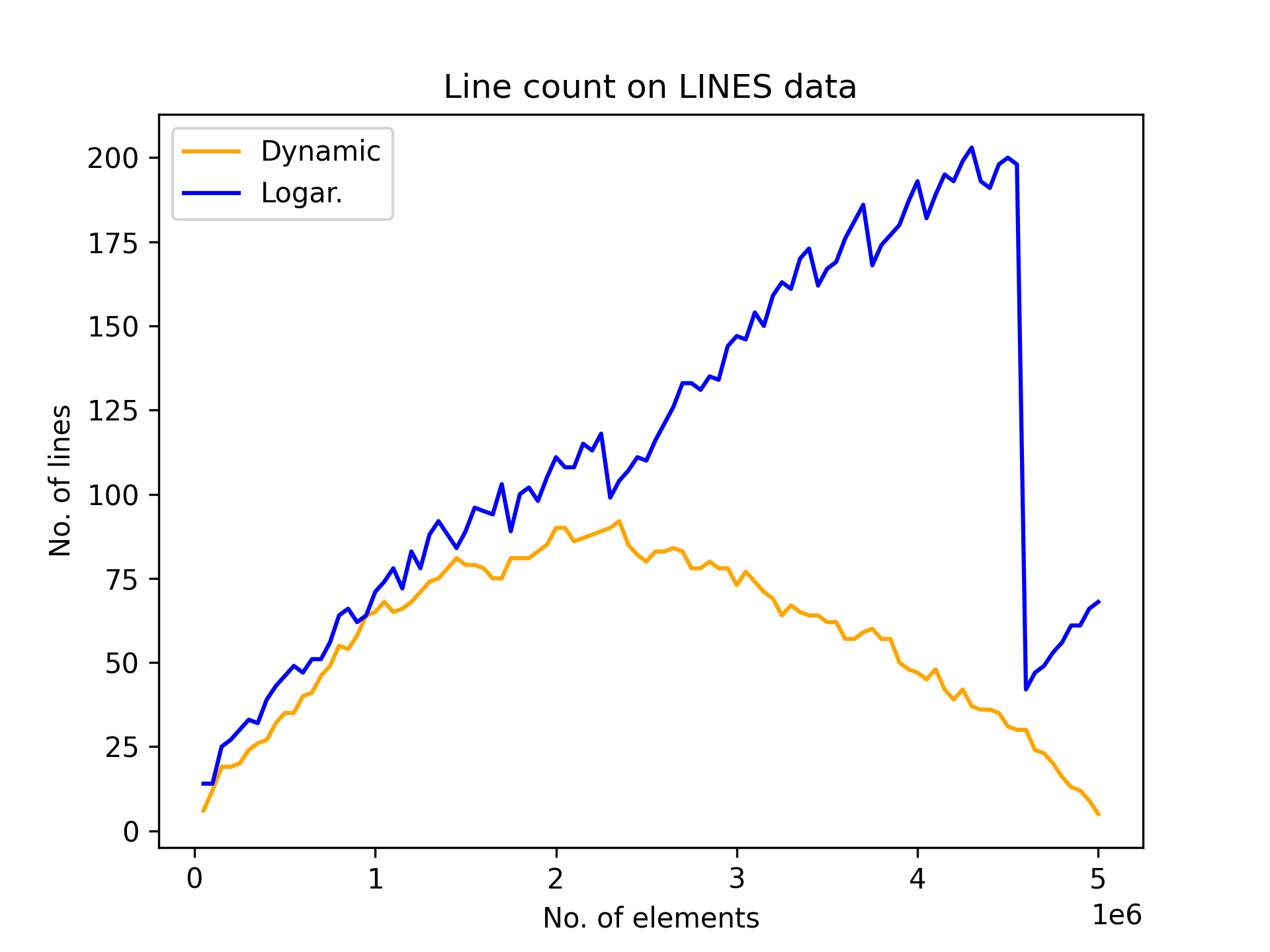}
    \includegraphics[width=0.49\linewidth]{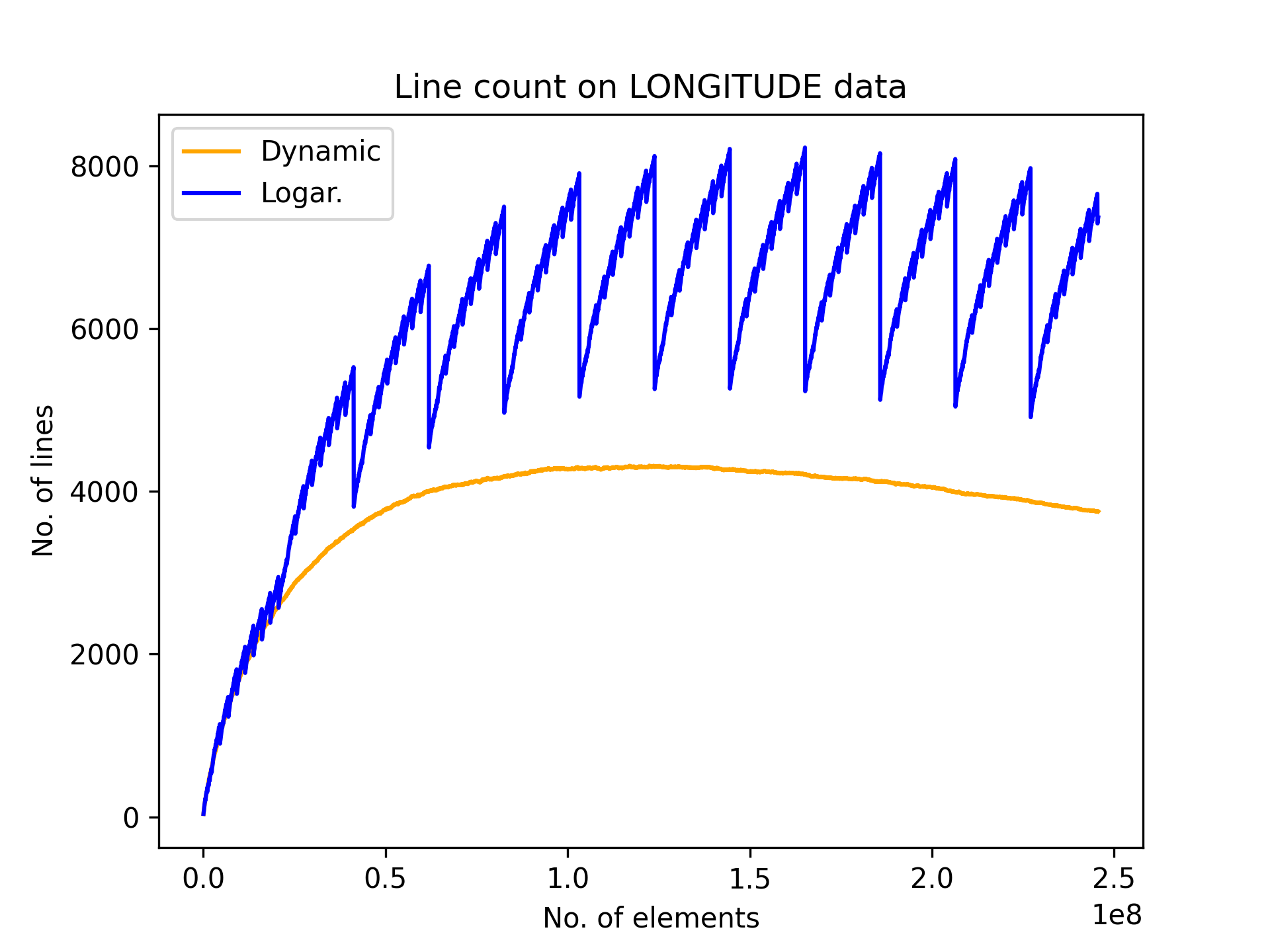}
    \includegraphics[width=0.49\linewidth]{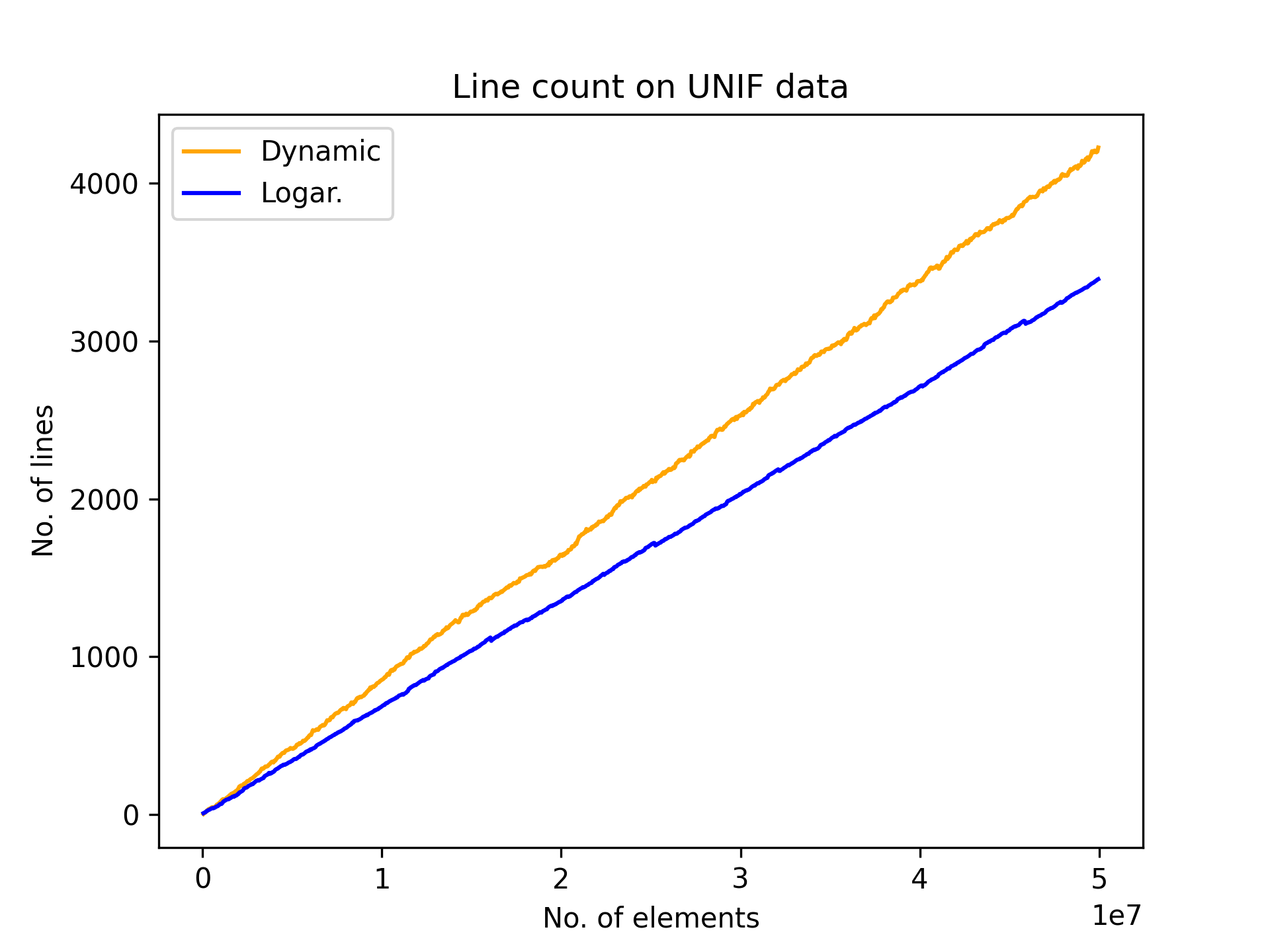}
    \includegraphics[width=0.49\linewidth]{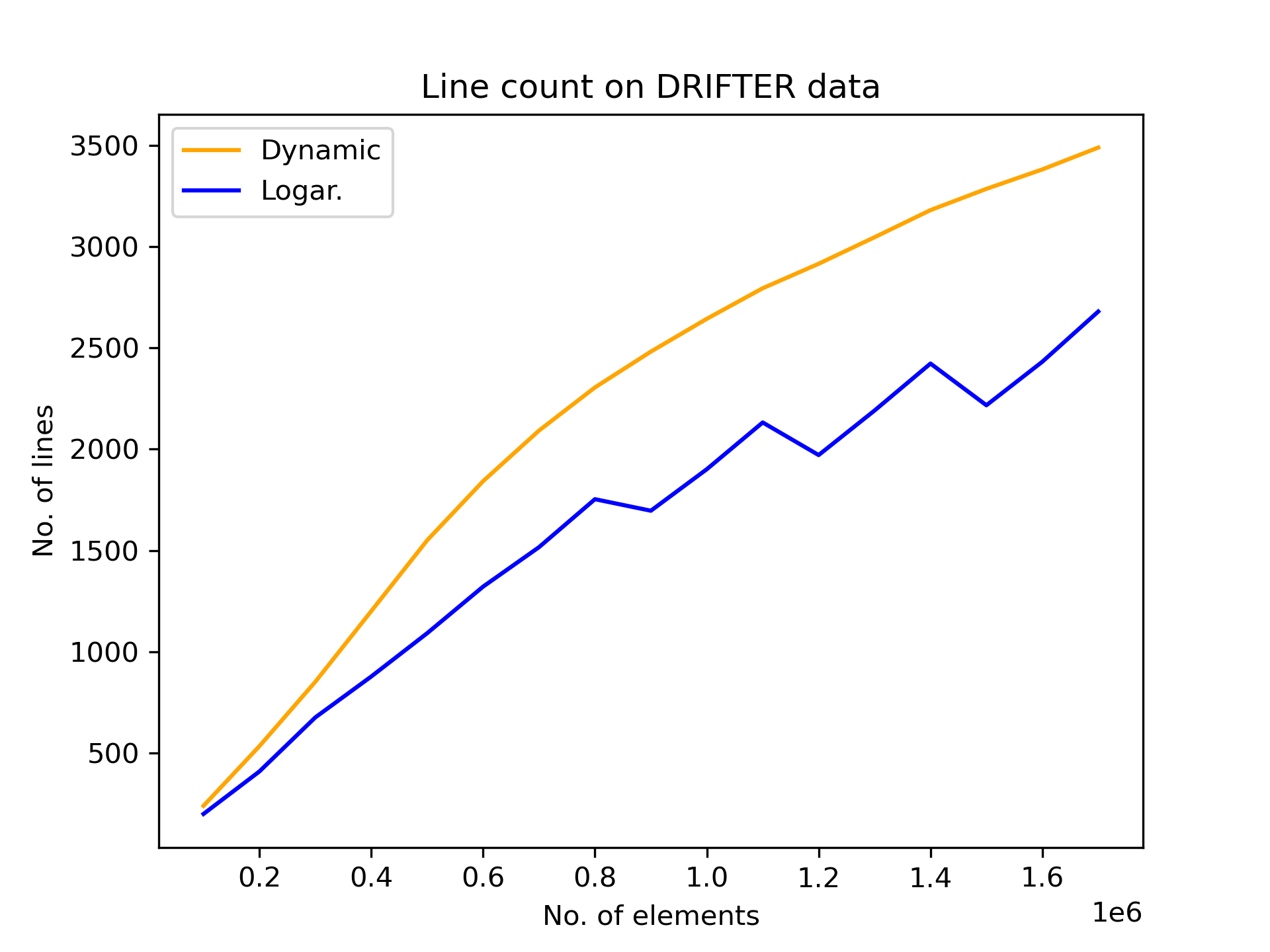}
    \caption{The complexity of the learned indices throughout insertion-only construction. The top graphs represent geometrically structured data. The left graphs represent synthetic data.}
    \label{fig:res_linecount}
\end{figure}

For unstructured data such as \textsc{Unif}, both learned indices display similar asymptotic trends, with the logarithmic PGM using approximately 30 percent fewer segments. Here, the complexity appears to scale as $\Omega(n)$ for both methods, suggesting that learned indices offer few to no improvements in the absence of exploitable structure.
Surprisingly, the \textsc{Drifter} data -- despite its geometric origin -- shows similar results to that of \textsc{Unif}. Again, complexity scales linearly, and the logarithmic method outperforms the dynamic one by around 30 percent. This implies that whatever latent geometric structure exists is insufficiently captured by the learned index under either strategy.

An interesting observation is that, on unstructured data, the imposed bucketing of the logarithmic PGM index can introduce a form of regularity that the model benefits from, essentially imposing artificial structure where none exists.

\subsection{Running time comparisons}
Recall that our indexing structure is composed of a learned index over a vertical $\varepsilon$-cover and a paging structure. We first examine the performance of both the learned index in its own, and then the performance of the full indexing structure.

For the dynamic scenarios, we first follow the precedence set by prior papers~\cite{ferragina2020pgm,kipf2019sosdbenchmarklearnedindexes}. These first construct, insertion-only, the indexing structure. They then perform a batch of $10$M operations
These batches consist of insertions, deletions, and range queries over ranges such that the output contains approximately $\frac{\sqrt{n}}{10}$ elements.
We deviate from the precedent by also deleting from the index before performing a batch of operations, to simulate a scenario in which the structure has existed and transformed prior to processing.

\subsubsection{Maintaining a learned index}
Figure~\ref{fig:maintenance} shows the cost of maintaining each learned index under dynamic operations. Our update procedure, which ensures worst-case $O(\log^2 n)$ bounds, performs significantly worse than the amortised $O(\log n)$ updates offered by the logarithmic method -- especially on larger datasets such as \textsc{Longitude} and \textsc{Unif}. On smaller datasets like \textsc{Drifter}, the gap narrows, but the logarithmic method generally remains preferable in these update-only scenarios.

\begin{figure}[h]
    \centering
    \includegraphics[width=0.32\linewidth]{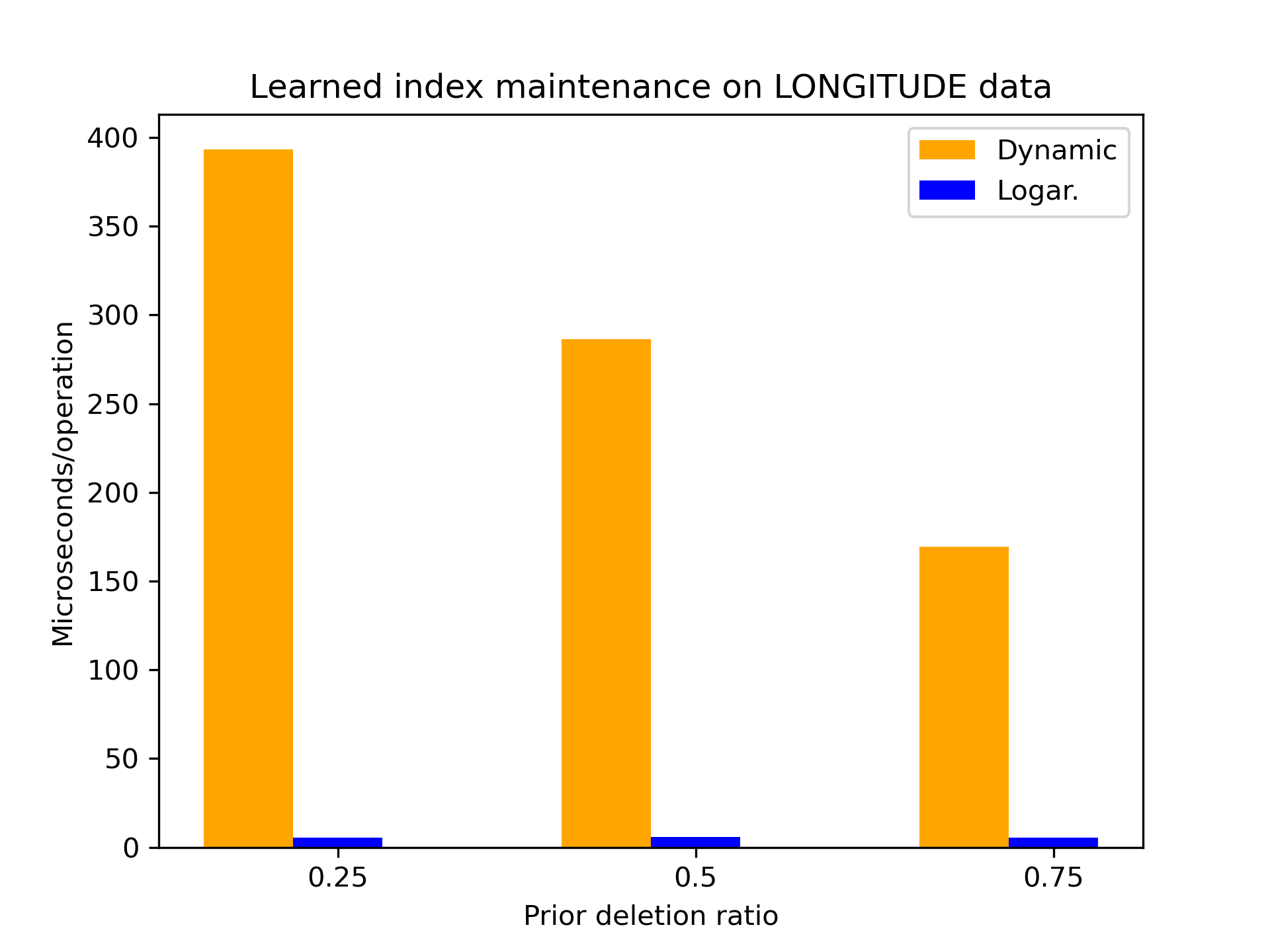}
    \includegraphics[width=0.32\linewidth]{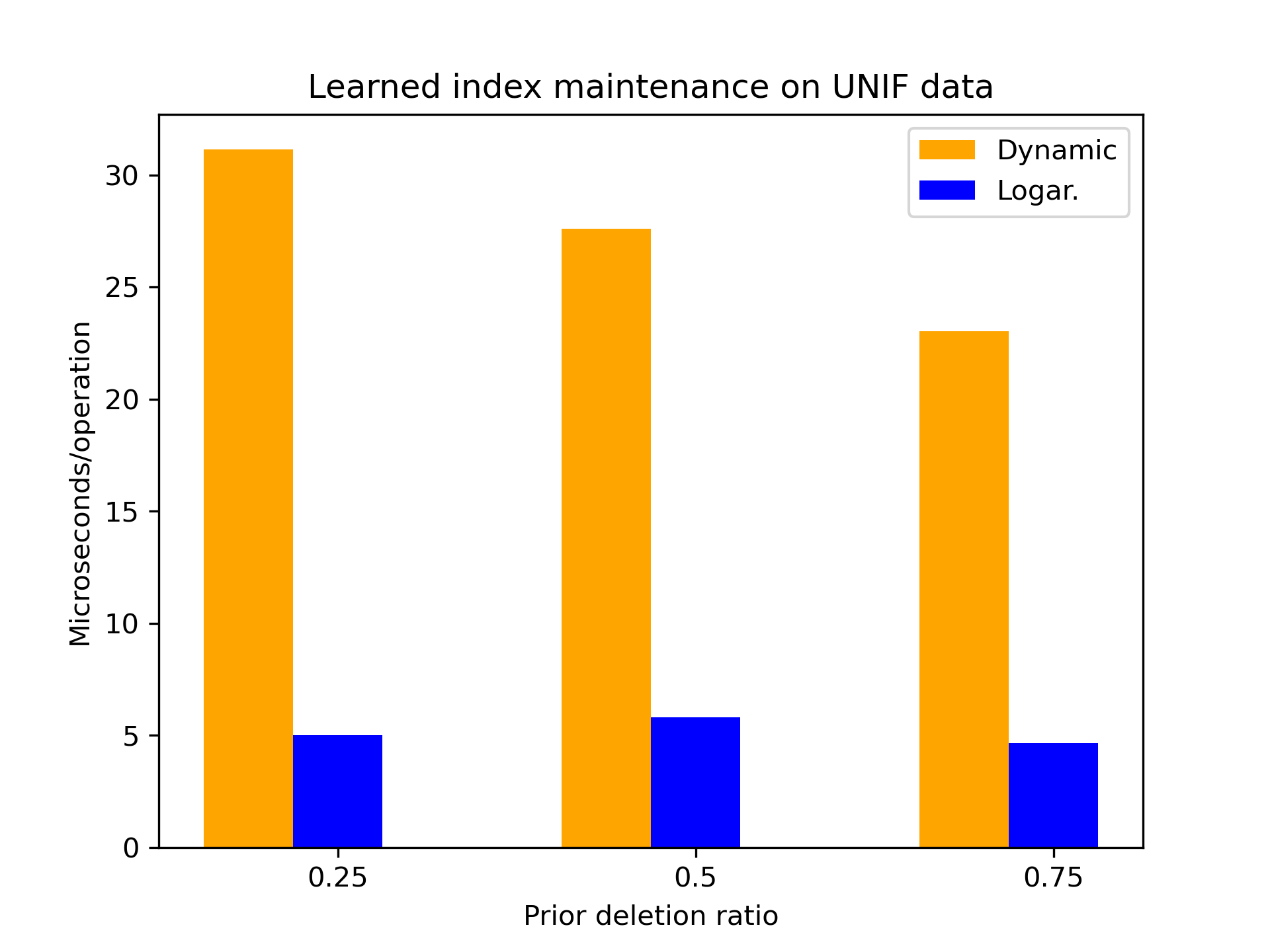}
    \includegraphics[width=0.32\linewidth]{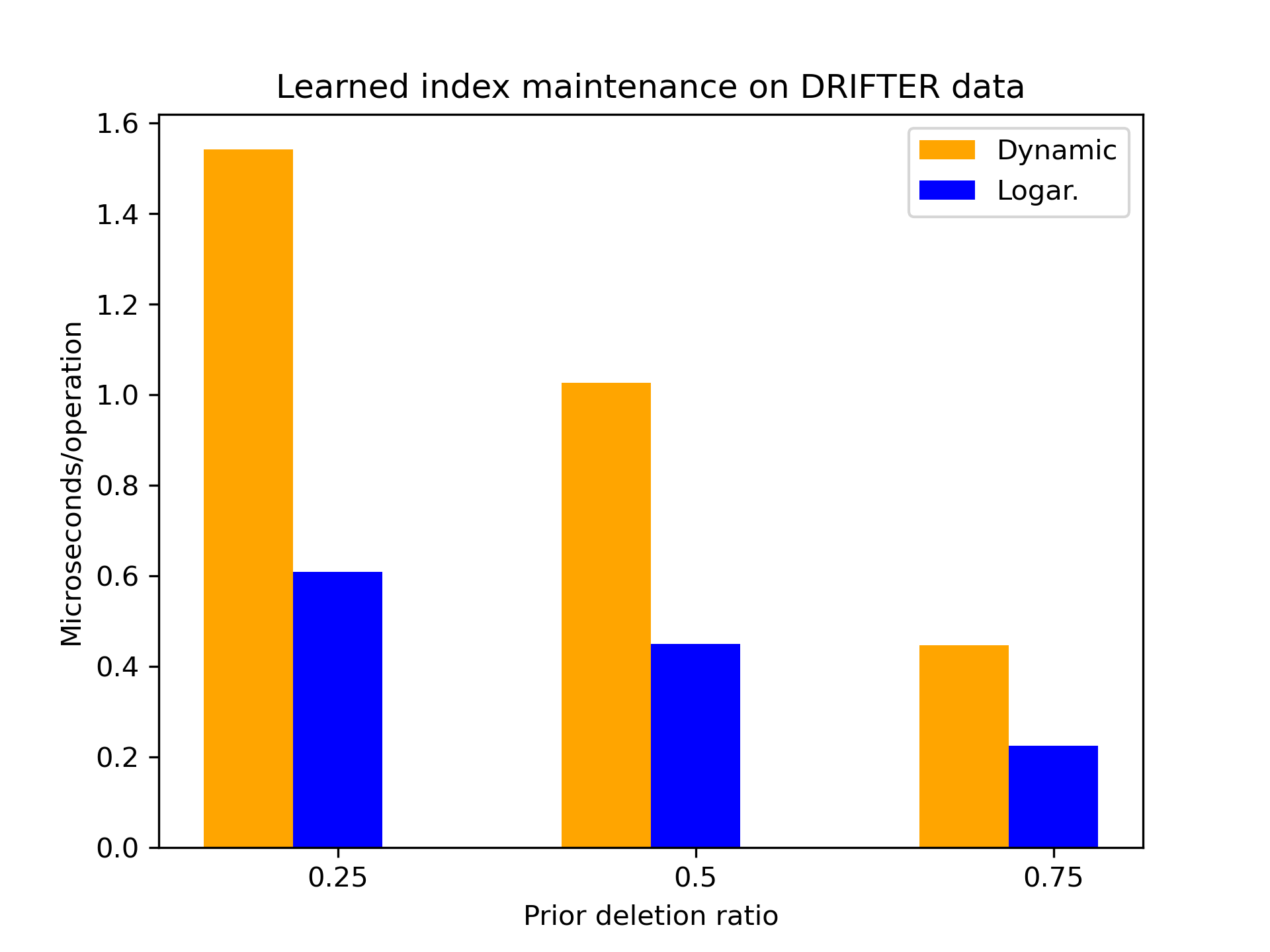}
    \caption{Update times for maintaining the learned index dynamically. }
    \label{fig:maintenance}
\end{figure}

\subsubsection{Indexing data structures}
We next assess the complete indexing structures under mixed workloads. Batches contain a tunable ratio of queries and updates, where updates are evenly split between insertions and deletions, and operations are randomly ordered. More experiments can be found in Appendix~\ref{app:experiments}.
For the \textsc{Lines} dataset, updates are constrained to a single line segment to preserve its idealised structure. For smaller datasets, batch size is limited to 25 percent.


\subparagraph{Adversarial workload (ADV).}
In the above scenario by~\cite{ferragina2020pgm,kipf2019sosdbenchmarklearnedindexes}, an update batch of $10M$ operations affects less than ten percent of the data. Therefore, we do not encounter the worst case scenario where range queries take $O(N + \eps + \sum_i^{\lceil \log n \rceil} \log |f_i|)$ time. So, the worst case difference in performance does not come to light. Therefore, we additionally construct an \emph{adversarial scenario} consisting of 10M range queries after deleting all but $1.000$ values. 

\begin{figure}[h]
    \centering
    \includegraphics[width=0.49\linewidth]{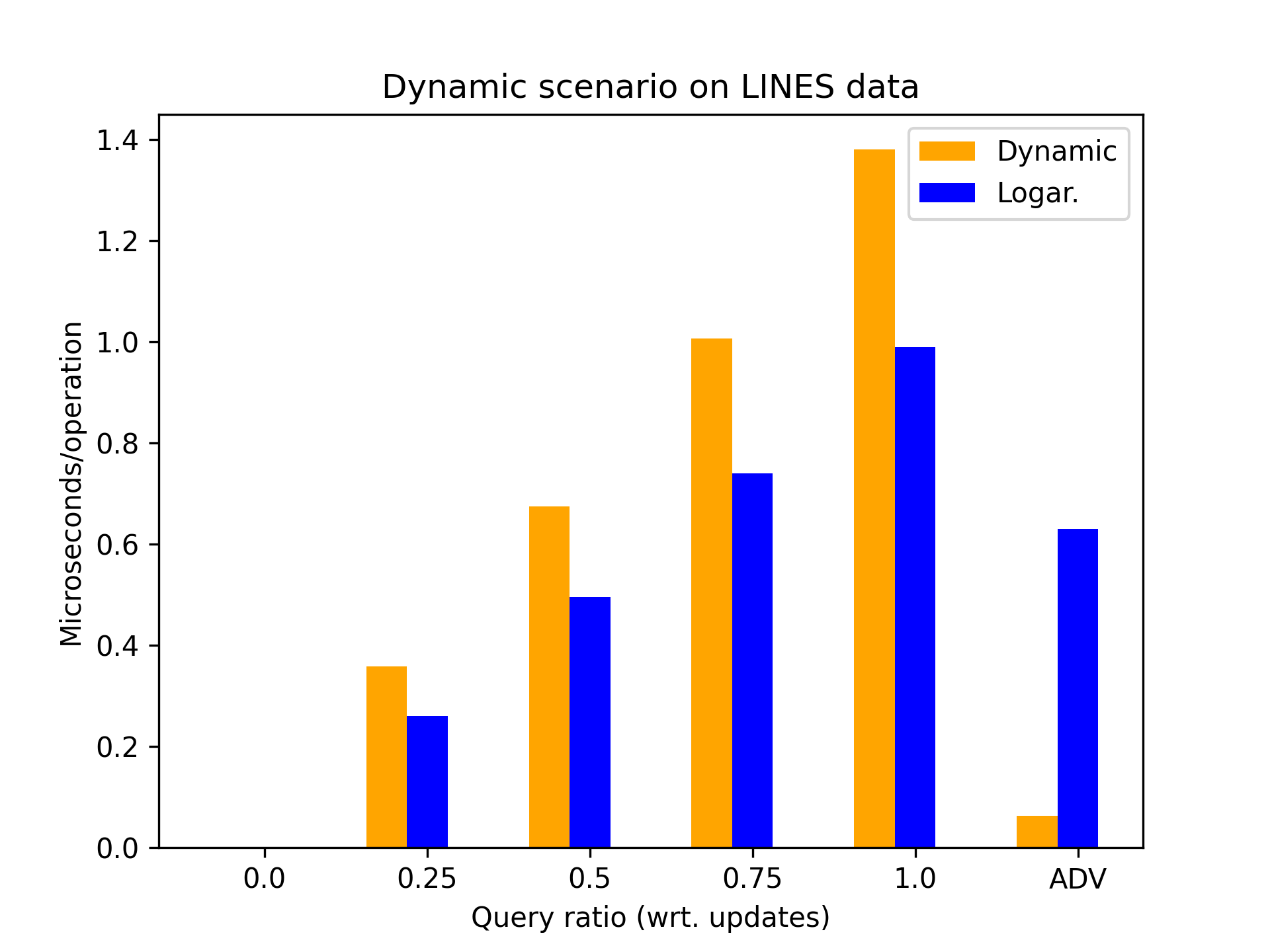}
    \includegraphics[width=0.49\linewidth]{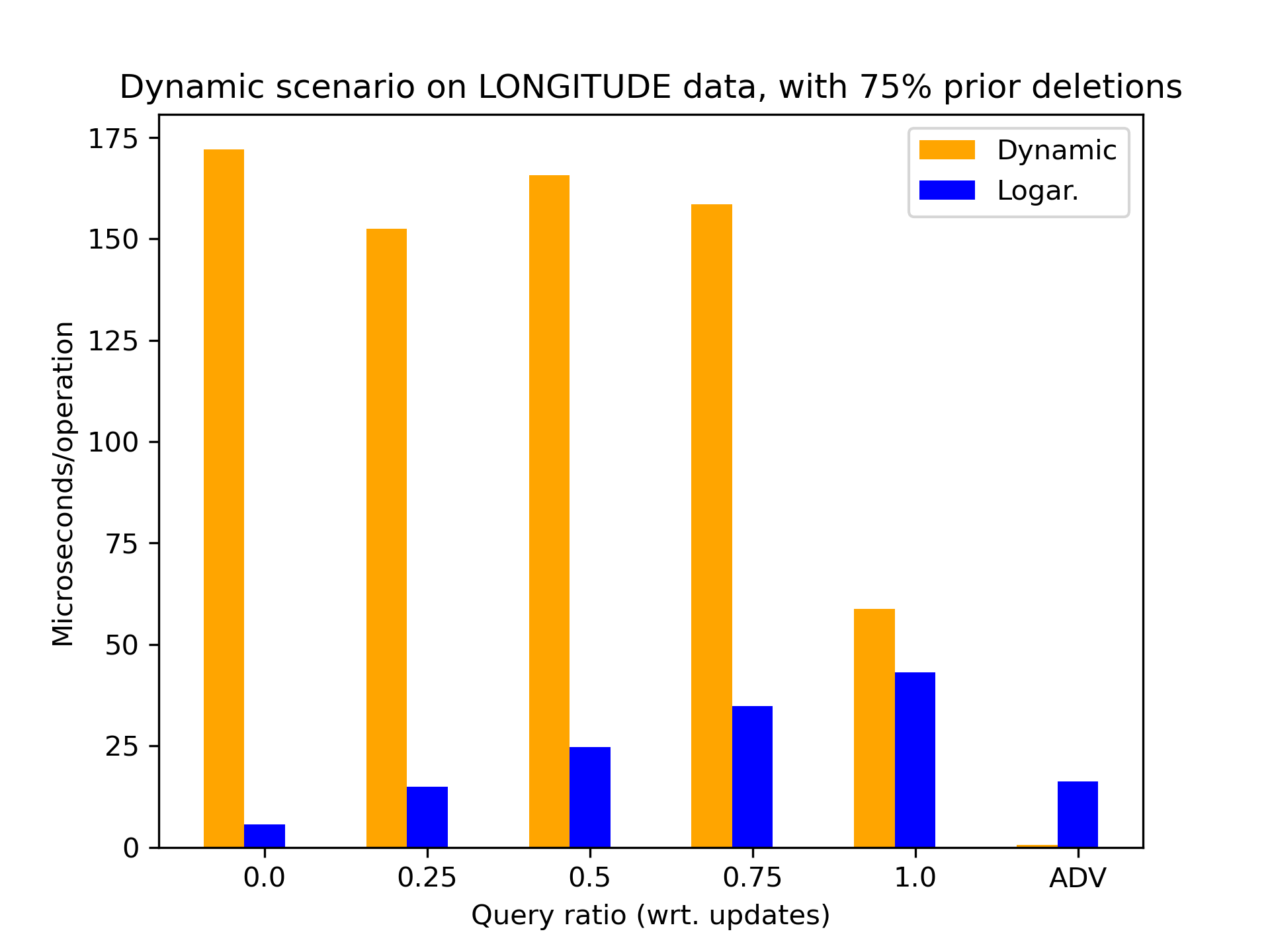}
    \includegraphics[width=0.49\linewidth]{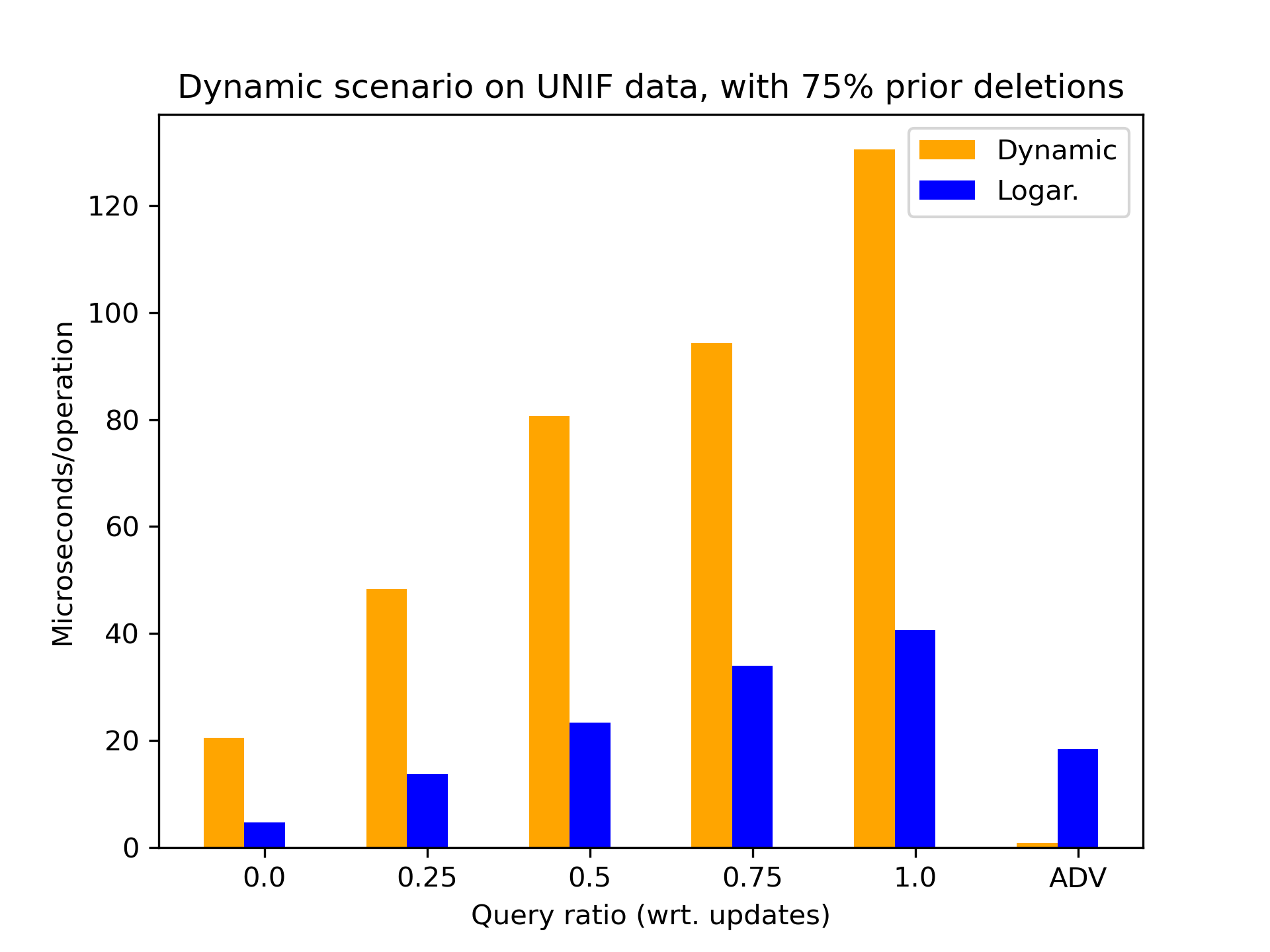}
    \includegraphics[width=0.49\linewidth]{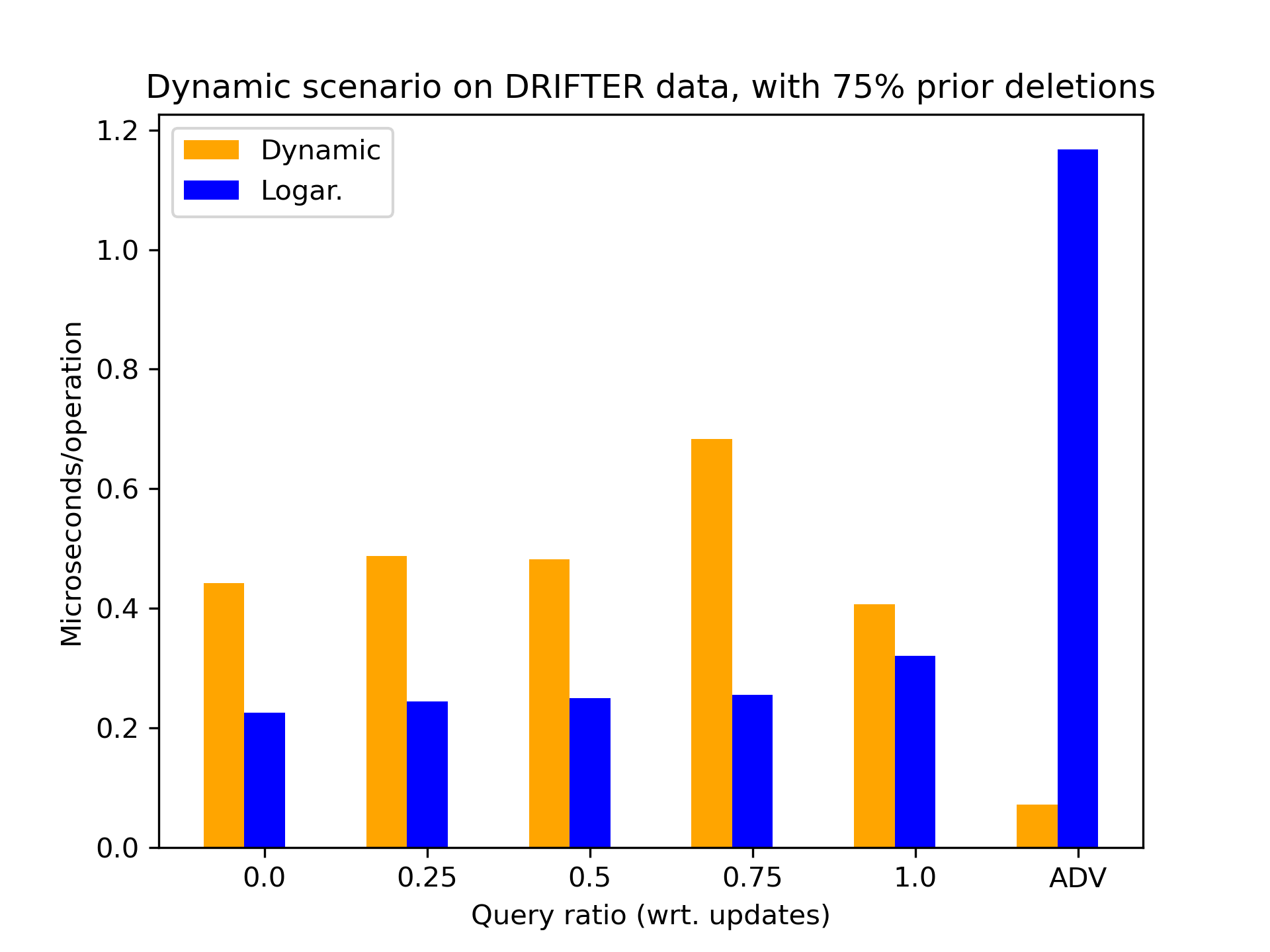}
    \caption{Time per operation in dynamic scenario with varying query ratios.  For the LINES data set updates are restricted to points on a single line. The ADV ratio denotes the adversarial case where all but a constant number of elements are deleted before the structure is queried.}
    \label{fig:res_updates}
\end{figure}

\subparagraph{Operational performance.}
Figure~\ref{fig:res_updates} summarises the total processing time under varying query-to-update ratios. Structured datasets, particularly \textsc{Longitude}, reveal a trade-off between the fact that we have a lower-complexity learned index and  our update time. 
On one hand, our updates are costly due to our line-merging tests which suffer from the poor cache behaviour from pointer-based trees, and complex rebalancing that is present in the state-of-the-art practical dynamic convex hull data structures.
This behaviour is also reflected on our previous analysis of maintaining the learned index itself. On the other hand we see on
 geometrically structured data that, as the query ratio increases, performance improves.  
This trend is absent in the \textsc{Lines} data, where updates are constrained and the structure is small, mitigating cache penalties. Here, running time is largely driven by query processing, which increases as expected with query ratio.



In contrast, the random data from the \textsc{Unif} dataset shows improved performance with higher query ratios. In our data structure, having smaller segments reduce restructuring costs. Our dynamic implementation suffers greatly from cache inefficiency due to random access during range reporting. The logarithmic PGM, with its lower memory overhead and sequential layout, delivers significantly better performance -- particularly as sizes increase.


For the \textsc{Drifter} data, one would expect a trend similar to that of the \textsc{Unif} data, based on the complexity of the $\varepsilon$-cover shown in Figure~\ref{fig:res_linecount}. However, there is little difference in performance for either PGM as the query ratio is varied. This is likely due to the small size of the data set, with both structures performing updates in fractions of microseconds on average. The logarithmic PGM does come out slightly on top, likely due to the machine friendly memory access pattern.

\subparagraph{Adversarial impact.} Across all datasets, the logarithmic PGM's tombstoning strategy becomes a bottleneck in the adversarial scenario. Its range queries must scan and subtract deleted values before outputting results. Our structure, by contrast, is output-sensitive and avoids such overhead.
This illustrates that our data structure, whilst being generally less efficient that the logarithmic PGM, offers worst-case guarantees.

\section{Conclusion}
We studied dynamic learned indices through a geometric perspective. Following the work of Ferragina and Vinciguerra~\cite{ferragina2020pgm}, we maintained a learned index $h_\varepsilon$ of a dynamic set $S$ as a piecewise-linear approximation---an $\varepsilon$-cover---of the rank-space point set $F_S$. We used techniques from computational geometry to answer the following question: 

\begin{quote}
    ``Can a learned index be dynamically maintained with worst-case guarantees?''
\end{quote}

\noindent
We proposed a new approach to maintain a learned index based on dynamic convex hull data structures. We presented an $O(\log^2 n)$ time algorithm to dynamically maintain a learned index $h_\eps$. To obtain this algorithm, we showed an algorithm to compute a separating line between two non-intersecting convex hulls -- an operation previously missing from the literature.
The existing logarithmic PGM index has an  amortised $O(\log n)$ time update algorithm which is more efficient in practice. Indeed, although close in theory, the memory access pattern associated with dynamically maintaining convex hulls incurs heavy penalties for large datasets. At the same time, the resulting learned index of the logarithmic PGM, $h_\eps'$, can be considerably more complex on geometrically structured data.


\subparagraph{From learned indices to indexing.}
Finally, we considered the following question:

\begin{quote}
    ``Can a dynamic learned index be converted into a dynamic indexing data structure?''
\end{quote}

\noindent
To this end, we designed a hybrid technique combining hashing-based fast-access data structures with a doubly linked list to support indexing queries. Our method offers output-sensitive worst-case guarantees, even in the presence of deletions. As it is known that traditional indexing structures currently outperform learned indices in the general dynamic setting~\cite{Sun23survey}, we focused our comparisons on improving the theoretical and practical performance within the class of learned approaches.

In practice, 
the contiguous memory access of the logarithmic PGM index offsets its overhead from querying multiple structures, making it faster in all scenarios, except for an adversarial one.
This means that the lower complexity of our learned index does not immediately translate to improved efficiency in the indexing data structure.
This raises an interesting open question of whether memory-access efficient fully dynamic approaches to convert a learned index into a dynamic indexing structure can exist. 

While we acknowledge that our update-times are slow in comparison with state-of-the art, 
our approach does illustrate that it brings worst-case guarantees: as it has an advantage when the query-to-update ratio is large and the index has undergone sufficiently many deletions.  
In adversarial workloads with frequent deletions followed by range queries, we have seen our structure outperform the logarithmic approach -- highlighting the value of worst-case guarantees even in specialised settings.

\subparagraph{Closing thoughts.}
We showed what we believe is an interesting connection between the geometric learned index by Ferragina and Vinciguerra, and dynamic convex hulls from computational geometry. 
We subsequently provided an implementation of a dynamic learned index that relies on the state-of-the-art dynamic convex hull maintenance algorithms. 
Our empirical analysis shows that the complex tree rebalancing that is used to dynamically maintain a convex hull currently brings considerable operational overhead compared to low-memory techniques under the logarithmic method.
Our experiments, though not uniformly favourable, offer interesting insights into the current barriers and adversarial tradeoffs between worst-case dynamic algorithms and memory-efficient amortised rebuilding schemes.

\begin{algorithm}[H]
    \caption{\texttt{intersection\_test}(\texttt{edge} $\alpha \in CH(A)$, \texttt{edge} $\beta \in CH(B)$ )}
    \label{alg:intersection_test}
    \begin{algorithmic}[1]
    \If{ $\alpha = $ \emph{null} OR $\beta =$ \emph{null}}
     \State \Return No
     \EndIf
    \State $s(\alpha, \beta) = line(\alpha) \cap line(\beta)$
    \If{If $s \in \alpha$ and $s \in \beta$}
    \State \Return Yes
    \EndIf
    \If{$\alpha$.slope  $<$  $\beta$.slope}
    \If{ $\alpha$.first.$x$ > $s(\alpha, \beta)$.x}
    \State \Return \texttt{intersection\_test}($\alpha$.left, $\beta$)
    \ElsIf{ $\beta$.first.$x$ > $s(\alpha, \beta)$.x }
    \State \Return \texttt{intersection\_test}($\alpha$, $\beta$.left)
    \ElsIf {$\alpha$.first.$x$ $>$ $\beta$.second.$x$ AND $\alpha$.first.$y$ > $\beta$.second.$y$}
    \State \Return  \texttt{intersection\_test}($\alpha$.left, $\beta$)
    \ElsIf {$\alpha$.second.$x$ $<$ $\beta$.first.$x$  AND $\alpha$.second.$y$ < $\beta$.first.$y$}
    \State \Return  \texttt{intersection\_test}($\alpha$, $\beta$.left)
    \Else
    \State \Return yes
    \EndIf
    \EndIf
    \If {$\alpha$.slope > $\beta$.slope}
    \If{$\alpha$.second.$x$ $<$ $s(\alpha, \beta).x$}
    \State \Return \texttt{intersection\_test}($\alpha$.right, $\beta$)
    \ElsIf {$\beta$.second.$x$ $<$ $s(\alpha, \beta).x$}
    \State \Return \texttt{intersection\_test}($\alpha$, $\beta$.right)
    \ElsIf {$\alpha$.first.$x$ $>$ $\beta$.second.$x$ AND $\alpha$.first.$y$ > $\beta$.second.$y$}
    \State \Return  \texttt{intersection\_test}($\alpha$, $\beta$.right)
    \ElsIf {$\alpha$.second.$x$ $<$ $\beta$.first.$x$  AND $\alpha$.second.$y$ < $\beta$.first.$y$}
    \State \Return  \texttt{intersection\_test}($\alpha$.right, $\beta$)
    \Else
    \State \Return yes
    \EndIf
    \EndIf
    \end{algorithmic}
  \end{algorithm}

\bibliography{refs}

\begin{thebibliography}{10}

\bibitem{ahmed2015}
Mahmuda Ahmed, Sophia Karagiorgou, Dieter Pfoser, and Carola Wenk.
\newblock A comparison and evaluation of map construction algorithms using vehicle tracking data.
\newblock {\em GeoInformatica}, 19(3):601--632, 2015.

\bibitem{athanassoulis2014bf}
Manos Athanassoulis and Anastasia Ailamaki.
\newblock Bf-tree: approximate tree indexing.
\newblock In {\em International Conference on Very Large Databases (VLDB)}, 2014.

\bibitem{Barba2015optimal}
Luis Barba and Stefan Langerman.
\newblock Optimal detection of intersections between convex polyhedra.
\newblock {\em ACM-SIAM Symposium on Discrete Algorithms (SODA)}, 2015.

\bibitem{bender2000cache}
Michael~A Bender, Erik~D Demaine, and Martin Farach-Colton.
\newblock Cache-oblivious b-trees.
\newblock In {\em Symposium on Foundations of Computer Science (FOCS)}. IEEE, 2000.

\bibitem{chan1998bitmap}
Chee-Yong Chan and Yannis~E Ioannidis.
\newblock Bitmap index design and evaluation.
\newblock In {\em ACM International Conference on Management of Data (SIGMOD)}, 1998.

\bibitem{Chazelle1987intersection}
B.~Chazelle and D.~P. Dobkin.
\newblock Intersection of convex objects in two and three dimensions.
\newblock {\em Journal of the ACM}, 1987.
\newblock \href {https://doi.org/10.1145/7531.24036} {\path{doi:10.1145/7531.24036}}.

\bibitem{chazelle1980detection}
Bernard Chazelle and David~P Dobkin.
\newblock Detection is easier than computation.
\newblock In {\em ACM Symposium on Theory Of Computing (STOC)}, 1980.

\bibitem{conradi2023}
Jacobus Conradi and Anne Driemel.
\newblock Finding complex patterns in trajectory data via geometric set cover.
\newblock {\em CoRR}, abs/2308.14865, 2023.
\newblock \href {https://doi.org/10.48550/ARXIV.2308.14865} {\path{doi:10.48550/ARXIV.2308.14865}}.

\bibitem{ding2020alex}
Jialin Ding, Umar~Farooq Minhas, Jia Yu, Chi Wang, Jaeyoung Do, Yinan Li, Hantian Zhang, Badrish Chandramouli, Johannes Gehrke, Donald Kossmann, et~al.
\newblock Alex: an updatable adaptive learned index.
\newblock In {\em ACM International Conference on Management of Data (SIGMOD)}, 2020.

\bibitem{DOBKIN1991}
David Dobkin and Diane Souvaine.
\newblock Detecting the intersection of convex objects in the plane.
\newblock {\em Computer Aided Geometric Design}, 1991.
\newblock \href {https://doi.org/10.1016/0167-8396(91)90001-R} {\path{doi:10.1016/0167-8396(91)90001-R}}.

\bibitem{DOBKIN1983}
David~P. Dobkin and David~G. Kirkpatrick.
\newblock Fast detection of polyhedral intersection.
\newblock {\em Theoretical Computer Science (TSC)}, 1983.
\newblock \href {https://doi.org/10.1016/0304-3975(82)90120-7} {\path{doi:10.1016/0304-3975(82)90120-7}}.

\bibitem{Dobkin1990}
David~P. Dobkin and David~G. Kirkpatrick.
\newblock Determining the separation of preprocessed polyhedra - a unified approach.
\newblock In {\em International Colloquium on Automata, Languages and Programming (ICALP)}, 1990.

\bibitem{fabri2000design}
Andreas Fabri, Geert-Jan Giezeman, Lutz Kettner, Stefan Schirra, and Sven Sch{\"o}nherr.
\newblock On the design of cgal a computational geometry algorithms library.
\newblock {\em Software: Practice and Experience}, 30(11):1167--1202, 2000.

\bibitem{ferragina2020effective}
Paolo Ferragina, Fabrizio Lillo, and Giorgio Vinciguerra.
\newblock Why are learned indexes so effective?
\newblock In {\em Proceedings of the 37th International Conference on Machine Learning, {ICML} 2020, 13-18 July 2020, Virtual Event}, volume 119 of {\em Proceedings of Machine Learning Research}, pages 3123--3132. {PMLR}, 2020.

\bibitem{ferragina2021performance}
Paolo Ferragina, Fabrizio Lillo, and Giorgio Vinciguerra.
\newblock On the performance of learned data structures.
\newblock {\em Theor. Comput. Sci.}, 871:107--120, 2021.

\bibitem{ferragina2020learned}
Paolo Ferragina and Giorgio Vinciguerra.
\newblock Learned data structures.
\newblock In {\em Recent Trends in Learning From Data: Tutorials from the INNS Big Data and Deep Learning (INNSBDDL)}. Springer, 2020.

\bibitem{ferragina2020pgm}
Paolo Ferragina and Giorgio Vinciguerra.
\newblock The pgm-index: a fully-dynamic compressed learned index with provable worst-case bounds.
\newblock {\em International Conference on Very Large Databases (VLDB)}, 2020.

\bibitem{galakatos2019fiting}
Alex Galakatos, Michael Markovitch, Carsten Binnig, Rodrigo Fonseca, and Tim Kraska.
\newblock Fiting-tree: A data-aware index structure.
\newblock In {\em ACM International Conference on Management of Data (SIGMOD)}, pages 1189--1206, 2019.

\bibitem{Gaede2024simple}
Emil Gæde, Inge Li~Gørtz, Ivor Van Der~Hoog, Christoffer Krogh, and Eva Rotenberg.
\newblock Simple and robust dynamic two-dimensional convex hull.
\newblock {\em ACM Symposium on Algorithm Engineering and Experiments (ALENEX)}, 2024.

\bibitem{impl_index}
Emil~Toftegaard Gæde, Ivor van~der Hoog, Eva Rotenberg, and Tord Stordalen.
\newblock {Implementation of a Dynamic Learned Index}.
\newblock \url{https://github.com/Sgelet/DynamicLearnedIndex}, 2025.

\bibitem{impl_bench}
Emil~Toftegaard Gæde, Ivor van~der Hoog, Eva Rotenberg, and Tord Stordalen.
\newblock {Testbed for Dynamic Learned Indices}.
\newblock \url{https://github.com/Sgelet/LearnedIndexBench}, 2025.

\bibitem{kipf2019sosdbenchmarklearnedindexes}
Andreas Kipf, Ryan Marcus, Alexander van Renen, Mihail Stoian, Alfons Kemper, Tim Kraska, and Thomas Neumann.
\newblock Sosd: A benchmark for learned indexes.
\newblock {\em Conference on Neural Information Processing Systems (NEURIPS)}, 2019.

\bibitem{kipf2020radixspline}
Andreas Kipf, Ryan Marcus, Alexander van Renen, Mihail Stoian, Alfons Kemper, Tim Kraska, and Thomas Neumann.
\newblock Radixspline: a single-pass learned index.
\newblock In {\em International workshop on exploiting artificial intelligence techniques for data management}, 2020.

\bibitem{koudas2000space}
Nick Koudas.
\newblock Space efficient bitmap indexing.
\newblock In {\em ACM international conference on Information and knowledge management (SIGMOD)}, 2000.

\bibitem{kraska2018case}
Tim Kraska, Alex Beutel, Ed~H Chi, Jeffrey Dean, and Neoklis Polyzotis.
\newblock The case for learned index structures.
\newblock In {\em ACM International Conference on Management of Data (SIGMOD)}, 2018.

\bibitem{lin2023learning}
Yuming Lin, Zhengguo Huang, and You Li.
\newblock Learning hash index based on a shallow autoencoder.
\newblock {\em Applied Intelligence}, 2023.

\bibitem{m42020}
Spyros Makridakis, Evangelos Spiliotis, and Vassilios Assimakopoulos.
\newblock The m4 competition: 100,000 time series and 61 forecasting methods.
\newblock {\em International Journal of Forecasting}, 36(1):54--74, 2020.

\bibitem{books200m}
Ryan Marcus, Andreas Kipf, and Alex van Renen.
\newblock {Searching on Sorted Data}, 2019.
\newblock \href {https://doi.org/10.7910/DVN/JGVF9A} {\path{doi:10.7910/DVN/JGVF9A}}.

\bibitem{marcus2020benchmark}
Ryan Marcus, Andreas Kipf, Alexander van Renen, Mihail Stoian, Sanchit Misra, Alfons Kemper, Thomas Neumann, and Tim Kraska.
\newblock Benchmarking learned indexes.
\newblock {\em Proc. {VLDB} Endow.}, 14(1):1--13, 2020.

\bibitem{o1981line}
Joseph O'Rourke.
\newblock An on-line algorithm for fitting straight lines between data ranges.
\newblock {\em Communications of the ACM}, 1981.

\bibitem{ORourke1998Computational}
Joseph O'Rourke.
\newblock {\em Computational geometry in C (second edition)}.
\newblock Cambridge University Press, USA, 1998.

\bibitem{overmars1983design}
Mark~H Overmars.
\newblock {\em The design of dynamic data structures}, volume 156.
\newblock Springer Science \& Business Media, 1983.

\bibitem{overmars1981maintenance}
Mark~H Overmars and Jan Van~Leeuwen.
\newblock Maintenance of configurations in the plane.
\newblock {\em Journal of computer and System Sciences}, 1981.

\bibitem{pagh2004cuckoo}
Rasmus Pagh and Flemming~Friche Rodler.
\newblock Cuckoo hashing.
\newblock {\em Journal of Algorithms}, 2004.

\bibitem{puatracscu2006time}
Mihai P{\u{a}}tra{\c{s}}cu and Mikkel Thorup.
\newblock Time-space trade-offs for predecessor search.
\newblock In {\em ACM Symposium on Theory of Computing (STOC)}, 2006.

\bibitem{Sun23survey}
Zhaoyan Sun, Xuanhe Zhou, and Guoliang Li.
\newblock Learned index: {A} comprehensive experimental evaluation.
\newblock {\em Proc. {VLDB} Endow.}, 16(8):1992--2004, 2023.

\bibitem{torralba200880}
Antonio Torralba, Rob Fergus, and William~T Freeman.
\newblock 80 million tiny images: A large data set for nonparametric object and scene recognition.
\newblock {\em IEEE transactions on pattern analysis and machine intelligence}, 2008.

\bibitem{WaltherThesis}
Lukas Walther, Gerth Brodal, and Peyman Afshani.
\newblock Intersection of convex objects in the plane.
\newblock Master's thesis, Aarhus University, 2015.
\newblock Available at \url{https://cs.au.dk/~gerth/advising/thesis/lukas-walther.pdf}.

\bibitem{wang2015learning}
Jun Wang, Wei Liu, Sanjiv Kumar, and Shih-Fu Chang.
\newblock Learning to hash for indexing big data—a survey.
\newblock {\em Proceedings of the IEEE}, 2015.

\bibitem{wang2018building}
Ziqi Wang, Andrew Pavlo, Hyeontaek Lim, Viktor Leis, Huanchen Zhang, Michael Kaminsky, and David~G Andersen.
\newblock Building a bw-tree takes more than just buzz words.
\newblock In {\em ACM International Conference on Management of Data (SIGMOD)}, 2018.

\bibitem{wongkham2022updatable}
Chaichon Wongkham, Baotong Lu, Chris Liu, Zhicong Zhong, Eric Lo, and Tianzheng Wang.
\newblock Are updatable learned indexes ready?
\newblock {\em International Conference on Very Large Databases (VLDB)}, 2022.

\end{thebibliography}
\appendix
\newpage 

\newpage
\section{Algorithms engineering for separating lines of convex hulls}
\label{app:intersection_testing}

Let $A = (\alpha_1, \ldots, \alpha_n)$ and $B = (\beta_1, \ldots, \beta_n)$ be convex chains of $n$ edges with positive slope. 
Let $CH(A)$ be an upper-quarter convex hull (the boundary of the minimum convex area containing $A$ and $(\infty, -\infty)$).
Let $CH(B)$ be a lower quarter convex hull (the boundary of the minimum convex area containing $B$ and $(-\infty, \infty)$).

We do intersection testing between $A$ and $B$. In addition, we show an algorithm to compute a separating line between $A$ and $B$ in the negative case. 
We assume that we receive $A$ and $B$ in a tree structure.
Formally, we define a struct: \\

\noindent
\textbf{\texttt{edge}} $\alpha$

\texttt{vertex} first \hspace{2.7 cm} \emph{the first endpoint of $\alpha$}

\texttt{vertex} second \hspace{2.3 cm} \emph{the second endpoint of $\alpha$}

\texttt{edge} left \hspace{3.15 cm} \emph{the median edge of the remaining edges that precede $\alpha$}

\texttt{edge} right \hspace{2.9cm} \emph{the median edge of the remaining edges that succeed $\alpha$} 

\texttt{real} slope \hspace{2.85cm} \emph{the slope of the supporting line of the segment}\\

\noindent
For the first edge of $\alpha_1$ of $A$, we define $\alpha_1$.left as a vertical downward halfline.
For the last edge $\alpha_2$ of $A$, we define $\alpha_1$.right as a horizontal rightward halfline. The first and last edges of $B$ are also incident to either a horizontal and vertical halfline respectively.

\subparagraph{Goal and organisation.}
We first show an $O(\log n)$-time algorithm to test whether $A$ and $B$ intersect and we prove its correctness.
Then we extend this algorithm so that it gives a separating line of $CH(A)$ and $CH(B)$ in the negative case and prove its correctness. 
These algorithms compare the slopes of edges and compute intersection points between supporting lines of edges, which are na\"{i}vely not robust operations.
By Section~\ref{sec:robust}, we can implement these operations in a robust manner.

  \begin{lemma}
    \label{lem:intersection_test}
      Algorithm~\ref{alg:intersection_test} correctly determines whether $\CH(A)$ and $\CH(B)$ intersect.
  \end{lemma}

  \begin{proof}
      The proof follows the approach in~\cite{Chazelle1987intersection} and is a case distinction.

      \textbf{Suppose first} that $\alpha$.slope $<$ $\beta$.slope.
      We define $W$ as the cone with supporting lines $line(\alpha)$ and $line(\beta)$, that lies  left of $s(\alpha, \beta)$ and has $s(\alpha, \beta)$ as its apex  (Figure~\ref{fig:intersection_finding}).
      The area $\CH(A)$ is contained in the halfplane bounded from above by $line(\alpha)$. 
      The area $\CH(B)$ is contained in the halfplane bounded from below by $line(\beta)$.
      Any intersection between $\CH(A)$ and $\CH(B)$ must be contained in $W$ and therefore lie left of $s(\alpha, \beta)$.

        If $\alpha$.first lies strictly right of $s(\alpha, \beta)$ (line 11) then $\alpha$ cannot intersect any edge of $\CH(B)$. 
        Moreover, any edge in the right subtree of $\alpha$ cannot intersect $\CH(B)$ and so we may safely recurse on $\alpha$.left. 
        If $\beta$.first lies strictly right of $s(\alpha, \beta)$ (line 13) then $\beta$ and its right subtree cannot intersect any edge of $\CH(A)$ and we may safely recurse on $\beta$.left.

\begin{figure}[h]
    \centering
    \includegraphics[]{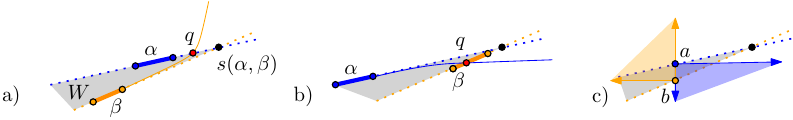}
    \caption{   (a) The edge $\beta$ strictly precedes $\alpha$ and the first vertex of $\alpha$ precedes $q$.
    (b) The edge $\alpha$ strictly precedes $\beta$.
    (c) There exists points $(a, b)$ where $a$ is in the top left quadrant if $b$. 
    }
    \label{fig:intersection_finding}
\end{figure}

 \begin{algorithm}[H]
    \caption{\texttt{intersection\_test}(\texttt{edge} $\alpha \in CH(A)$, \texttt{edge} $\beta \in CH(B)$ )}
    \label{alg:intersection_test}
    \begin{algorithmic}[1]
    \If{ $\alpha = $ \emph{null} OR $\beta =$ \emph{null}}
     \State \Return No
     \EndIf
    \State $s(\alpha, \beta) = line(\alpha) \cap line(\beta)$
    \If{If $s \in \alpha$ and $s \in \beta$}
    \State \Return Yes
    \EndIf
    \If{$\alpha$.slope  $<$  $\beta$.slope}
    \If{ $\alpha$.first.$x$ > $s(\alpha, \beta)$.x}
    \State \Return \texttt{intersection\_test}($\alpha$.left, $\beta$)
    \ElsIf{ $\beta$.first.$x$ > $s(\alpha, \beta)$.x }
    \State \Return \texttt{intersection\_test}($\alpha$, $\beta$.left)
    \ElsIf {$\alpha$.first.$x$ $>$ $\beta$.second.$x$ AND $\alpha$.first.$y$ > $\beta$.second.$y$}
    \State \Return  \texttt{intersection\_test}($\alpha$.left, $\beta$)
    \ElsIf {$\alpha$.second.$x$ $<$ $\beta$.first.$x$  AND $\alpha$.second.$y$ < $\beta$.first.$y$}
    \State \Return  \texttt{intersection\_test}($\alpha$, $\beta$.left)
    \Else
    \State \Return yes
    \EndIf
    \EndIf
    \If {$\alpha$.slope > $\beta$.slope}
    \If{$\alpha$.second.$x$ $<$ $s(\alpha, \beta).x$}
    \State \Return \texttt{intersection\_test}($\alpha$.right, $\beta$)
    \ElsIf {$\beta$.second.$x$ $<$ $s(\alpha, \beta).x$}
    \State \Return \texttt{intersection\_test}($\alpha$, $\beta$.right)
    \ElsIf {$\alpha$.first.$x$ $>$ $\beta$.second.$x$ AND $\alpha$.first.$y$ > $\beta$.second.$y$}
    \State \Return  \texttt{intersection\_test}($\alpha$, $\beta$.right)
    \ElsIf {$\alpha$.second.$x$ $<$ $\beta$.first.$x$  AND $\alpha$.second.$y$ < $\beta$.first.$y$}
    \State \Return  \texttt{intersection\_test}($\alpha$.right, $\beta$)
    \Else
    \State \Return yes
    \EndIf
    \EndIf
            \If{ $\alpha$.slope $=$  $\beta$.slope}
            \If{$line(\beta)$ is above $line(\alpha)$}
            \State \Return No
                \ElsIf {$\alpha$.first.$x$ $>$ $\beta$.second.$x$ AND $\alpha$.first.$y$ > $\beta$.second.$y$}
    \State \Return  \texttt{intersection\_test}($\alpha$.left, $\beta$)
    \ElsIf {$\alpha$.second.$x$ $<$ $\beta$.first.$x$  AND $\alpha$.second.$y$ < $\beta$.first.$y$}
    \State \Return  \texttt{intersection\_test}($\alpha$, $\beta$.left)
    \Else
        \State \Return No
            \EndIf        
    \EndIf
    \end{algorithmic}
  \end{algorithm}

        \newpage
              Let $\alpha$.first dominate $\beta$.second. 
        Suppose that $\CH(A)$ intersects $\CH(B)$ in a point $q$ right of $\alpha$.first (Figure~\ref{fig:intersection_finding} (a)). 
        Since $CH(A)$ has positive slope, $q$ dominates $\alpha$.first. 
        Consider the convex area $G$ enclosed by a curve $\gamma$ that traverses $CH(B)$ backwards until $\beta$.second, at which point it ends with a horizontal leftwards halfline. 
        Since $q \in \gamma$ dominates $\alpha$.first, and $\alpha$.first dominates $\beta$.second $\in \gamma$, $\alpha$.first lies in $G$. 
        However, $\CH(A)$ starts with a vertical downwards halfline left of $\alpha$.first.
        It follows that $\CH(A)$ and $\CH(B)$ also intersect in a point $q'$ left of $\alpha$.first and so we may safely recurse on $\alpha$.left.

        Let $\beta$.first dominate $\alpha$.second. Suppose that $\CH(B)$ intersects $\CH(A)$ in a point $q$ right of $\beta$.first (Figure~\ref{fig:intersection_finding} (b)).  
        Since $CH(B)$ has positive slope, $q$ dominates $\beta$.first.
        Consider the convex area $G$ enclosed by a curve $\gamma$ that traverses $CH(A)$ backwards until $\alpha$.first, at which point it ends with a vertical downwards halfline. 
        Since $q \in \gamma$ dominates $\beta$.first and $\beta$.first dominates $\alpha$.second $\in \gamma$, it follows that $\beta$.first is in $G$. 
        However, $\CH(B)$ starts with a horizontal leftwards halfline below $\beta$.first.
        It follows that $\CH(A)$ and $\CH(B)$ also intersect in a point $q'$ left of $\beta$.first and so we may safely recurse on $\beta$.left.

        Otherwise, $\alpha$ and $\beta$ share a vertical line with $\alpha$ above $\beta$ (or, a horizontal line with $\alpha$ left of $\beta$).  This is illustrated by Figure~\ref{fig:intersection_finding} (c).
        Pick a point $a \in \alpha$ and $b \in \beta$ such that they share a vertical line with $a$ above $b$ (or, such that they share a horizontal line with $a$ left of $b$). 
        The curve $CH(A)$ bounds an area containing the quarter plane that has $a$ as its top left corner.
        The curve $CH(B)$ bounds an area containing the quarter plane that has $b$ as its bottom right corner. It follows that $\CH(A)$ and $\CH(B)$ intersect and so we may output yes. 

        \textbf{Suppose next that $\slope(\alpha) > \slope(\beta)$.}
        Consider the cone $W$ with supporting lines $line(\alpha)$ and $line(\beta)$,  right of $s(\alpha, \beta)$, that has $s(\alpha, \beta)$ as its apex (Figure~\ref{fig:intersection_finding_two}).
            Any intersection between $\CH(A)$ and $\CH(B)$ must be contained in $W$ and therefore lie right of $s(\alpha, \beta)$. 

              Let $\alpha$.first dominate $\beta$.second. 
        Suppose that $\CH(A)$ intersects $\CH(B)$ in a point $q$ left of $\beta$.second. 
        Since $CH(B)$ has positive slope, $\beta$.second dominates $q$. 
        Consider the convex area $G$ enclosed by a curve $\gamma$ that traverses $CH(A)$ forwards until $\alpha$.first, at which point it ends with a horizontal rightwards halfline. 
        Since $q$ is dominated by $\beta$.second, and $\beta$.second is dominated by $\alpha$.first it follows that $\beta$.second  lies in $G$. 
        However, $\CH(B)$ ends with a vertical upwards halfline right of $\beta$.second. 
        It follows that $\CH(A)$ and $\CH(B)$ also intersect in a point $q'$ right of $\beta$.second and so we may safely recurse on $\beta$.right.

        Let $\beta$.first dominate $\alpha$.second. Suppose that $\CH(B)$ intersects $\CH(A)$ in a point $q$ left of $\alpha$.second.    
        Since $CH(A)$ has positive slope, $q$ is dominated by $\alpha$.second. 
        Consider the convex area $G$ enclosed by a curve $\gamma$ that traverses $CH(B)$ forwards until $\beta$.first, at which point it ends with a vertical upwards halfline. 
        Since $q$ is dominated by $\alpha$.second and $\alpha$.second is dominated by $\beta$.first, it follows that $\alpha$.second is in $G$. 
        However, $\CH(A)$ ends with a horizontal rightwards halfline above $\alpha$.second. 
        It follows that $\CH(A)$ and $\CH(B)$ also intersect in a point $q'$ right of $\alpha$.second and so we may safely recurse on $\alpha$.right.

        Otherwise, we may find a point $a \in \alpha$ that lies top left to a point $b \in \beta$ and so $CH(A)$ and $CH(B)$ are guaranteed to intersect. 
        
        \begin{figure}[h]
    \centering
    \includegraphics[page = 2]{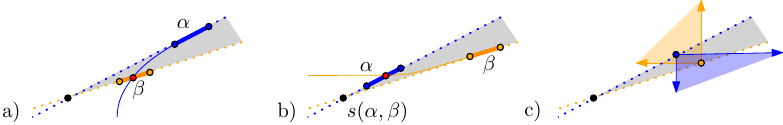}
    \caption{   (a) The edge $\beta$ strictly precedes $\alpha$ and the first vertex of $\alpha$ precedes $q$.
    (b) The edge $\alpha$ strictly precedes $\beta$.
    (c) There exists points $(a, b)$ where $a$ is in the top left quadrant if $b$. 
    }
    \label{fig:intersection_finding_two}
\end{figure}

\textbf{Otherwise, let $\slope(\alpha) = \slope(\beta)$.}
    If the supporting line of $\beta$ lies above $\alpha$ then $CH(A)$ and $CH(B)$ can never intersect (Figure~\ref{fig:intersection_finding_three}).
    Otherwise, the argument is identical to the previous two cases.    
  \end{proof}

\begin{figure}[h]
    \centering
    \includegraphics[page = 3]{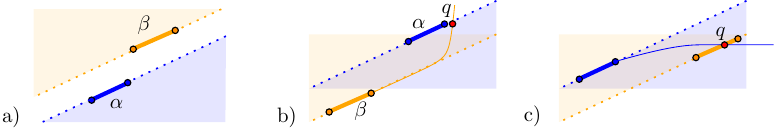}
    \caption{  If $line(\alpha)$ and $line(\beta)$ are parallel with the first below the latter, $CH(A)$ and $CH(B)$ cannot intersect. Otherwise, we may apply any of the previous arguments.  }
    \label{fig:intersection_finding_three}
\end{figure}

\subsection{Finding a separating line}
If Algorithm~\ref{alg:intersection_test} terminates  and outputs that $CH(A)$ and $CH(B)$ do not intersect then we can find a line that separates $CH(A)$ and $CH(B)$. 
Algorithm~\ref{alg:intersection_test} outputs \emph{no} in two cases. 
The first case is the special case where there exist two parallel edges $\alpha \in CH(A)$ and $\beta \in CH(B)$ where $line(\beta)$ lies above $line(\alpha)$.
In this case both $line(\alpha)$ and $line(\beta)$ are a separating line.

\begin{definition}
    For any edge $\alpha$, we denote by $\overleftarrow{\alpha}$ and $\overrightarrow{\alpha}$ its two supporting leftward and rightward halflines.  For any pair of edges $(\beta, b)$ that share a vertex with $\beta$ left of $b$, we denote by $w(\beta, b) = \overleftarrow{\beta} \cup \overrightarrow{b}$ their \emph{wedge}.
\end{definition}

The second case is that either argument of the function was \emph{null}.
Without loss of generality, we assume that $\beta$ was null. 
Then there exist two pairs of edges $(\alpha, \beta), (a, b) \in CH(A) \times CH(B)$ where
$\texttt{intersection\_test}(\alpha, \beta)$ recurses on $\beta$.right and $\texttt{intersection\_test}(a, b)$ recurses on $b$.left.
Moreover, the edges $\beta$ and $b$ must share a vertex. 
By keeping track of the traversal of Algorithm~\ref{alg:intersection_test}, we obtain $w(\beta, b)$  at no overhead.

\begin{lemma}
    \label{lem:nointersect}
    Let Algorithm~\ref{alg:intersection_test} terminate without finding an intersection between $CH(A)$ and $CH(B)$ and denote by $w(\beta, b)$ the corresponding wedge. Then:
    \begin{itemize}
        \item the halfline $\overleftarrow{\beta}$ cannot intersect $CH(A)$, and
        \item the halfline $\overrightarrow{b}$ cannot intersect $CH(A)$.
    \end{itemize}
\end{lemma}

\begin{proof}
              We first prove that the halfline $\overleftarrow{\beta}$ cannot intersect $CH(A)$. There exists some $\alpha \in CH(A)$ where \texttt{intersection\_test}($\alpha$, $\beta$) recurses on $\beta$.right.           
          Thus, $slope(\alpha) > slope(\beta)$.
          Define $s(\alpha, \beta) = line(\alpha) \cap line(\beta)$.  Observe that     \texttt{intersection\_test}($\alpha$, $\beta$) recurses on $\beta$.right in two cases. The first case is whenever $\beta$.second.$x$ $<$ $s(\alpha, \beta)$.$x$.
        Since $CH(A)$ lies in the plane upper bounded by $line(\alpha)$ this implies that $\overleftarrow{\beta}$ cannot intersect $CH(A)$. 

        In the second case, the vertex $\alpha$.first dominates  $\beta$.second (Figure~\ref{fig:intersection_finding_four} (a)).
        Suppose for the sake of contradiction that $\overleftarrow{\beta}$ intersects $CH(A)$ in some point $q$ left of $\beta$.second. 
        Since $line(\beta)$ has positive slope, $\beta$.second must dominate $q$. Consider the convex area $G$ bounded by a curve $\gamma$ that traverses $CH(A)$ backwards until $q$, after which it becomes a vertical downward halfline.  It follows that $\beta$.second is contained in $G \subseteq CH(A)$.
        This implies that $CH(A)$ and $CH(B)$ intersect which is a contradiction. 

        We argue that $\overrightarrow{b}$ cannot intersect $CH(A)$ in the same way ( Figure~\ref{fig:intersection_finding_four} (b)). 
          There must exist some $a \in CH(A)$ where  \texttt{intersection\_test}($a$, $b$) recurses on $b$.left.  
          If $slope(a) < slope(b)$ then we first consider the special case where $\beta$.first.$x$ $>$ $s(a, b)$.$x$. Since $CH(A)$ is contained in a halfplane upper bounded by $line(a)$ this implies that $\overrightarrow{b}$ cannot intersect $CH(A)$.
                
          If the special case does not apply, or whenever $slope(\alpha) = slope(\beta)$ then  it must be that $a$ is dominated by $b$. 
          If $\overrightarrow{b}$ intersects $CH(A)$ in a point $q$ then we may use $q$ to argue that $CH(A)$ and $CH(B)$ intersect in an identical manner as above.
\end{proof}

\begin{figure}[h]
    \centering
    \includegraphics[page = 4]{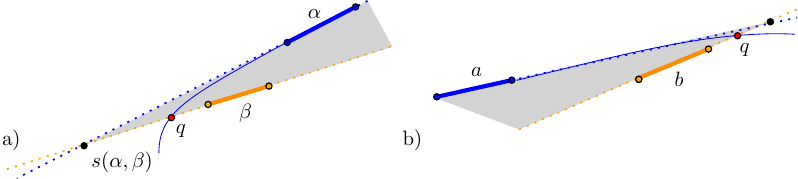}
    \caption{  (a) If there exists an edge $\alpha$ of $CH(A)$ that dominates $\beta$ and $\overrightarrow{\beta}$ intersects $CH(A)$ in a point $q$ then we may argue that $\beta$ is contained in $CH(A)$.    
    (b)  If there exists an edge $a$ of $CH(A)$ that is dominated by an edge $b$ then we make the symmetrical argument. }
    \label{fig:intersection_finding_four}
\end{figure}

\noindent
Given the edge $w(\beta, b)$ we run Algorithm~\ref{alg:separationfind}, starting with the root of $\alpha$.

\begin{algorithm}[h]
      \caption{\texttt{separation\_find}(\texttt{wedge} $w(\beta$, $b)$, 
      \texttt{edge} $a \in CH(A)$}
      \label{alg:separationfind}
      \begin{algorithmic}[1]
      \If{$\alpha = $ \emph{null}}
      \State \Return $line(\beta)$ or $line(b)$ 
      \EndIf
      \If{$line(\alpha) \cap w(\beta, b) = \emptyset$}
      \State \Return $line(\alpha)$
      \ElsIf{ $\overleftarrow{\alpha} \cap w(\beta, b) \neq \emptyset$}
      \State \Return \texttt{separation\_find}( $w(\beta, b)$, $\alpha$.left)
      \Else
      \State  \Return \texttt{separation\_find}($\beta$, $b$, $\alpha$.right)
      \EndIf
        \end{algorithmic}
\end{algorithm}

  \begin{lemma}
      Algorithm~\ref{alg:separationfind} outputs a edge in $\CH(A) \cup CH(B)$ whose supporting line separates $\CH(A)$ and $\CH(B)$.
  \end{lemma}

  \begin{proof}
      By Chazelle and Dobkin~\cite{Chazelle1987intersection} there always exists an edge on $CH(A)$ or $CH(B)$ whose supporting line separates the two convex hulls.
      
      First, we show that our algorithm always finds either an edge of $CH(A)$, or guarantees that for all edges in $\CH(A)$ their supporting line intersects $w(\beta, b)$. 
     Indeed, since $CH(B)$ is contained in $w(\beta, b)$, a line $line(\alpha)$ separates the two hulls if it does not intersect $w(\beta, b)$.

      Whenever $line(\alpha)$ \emph{does} intersect $w(\beta, b)$, either $\overleftarrow{\alpha}$ or $\overrightarrow{\alpha}$ must intersect $w(\beta, b)$. 
      Suppose that $\overleftarrow{\alpha}$ intersects $w(\beta, b)$. Any edge $a \in CH(A)$ succeeding $\alpha$ must have lower slope and so $\overleftarrow{a}$ must intersect $w(\beta, b)$. 
      Similarly if $\overrightarrow{\alpha}$ intersects $w(\beta, b)$ for any edge $a \in CH(A)$ preceding $\alpha$, $\overrightarrow{a}$ intersects $w(\beta, b)$.
      Finally, since $A$ starts with a vertical downwards halfline and ends with a horizontal rightwards halfline, it cannot be that for all edges $a \in CH(A)$ the halfline $\overleftarrow{a}$ intersects $w(\beta, b)$ (the same is true for $\overrightarrow{a}$). 

      Thus, if Algorithm~\ref{alg:separationfind} does not output an edge $\alpha \in CH(A)$ then there must exist two consecutive edges $(\gamma, g)$ on $CH(A)$ with the following property:
      for all edges $\gamma'$ of $CH(A)$ preceding and including $\gamma$,   $\overrightarrow{\gamma'}$ intersects $w(\beta, b)$, and, 
      for all edges $g'$ of $CH(A)$ succeeding an including $g$, $\overleftarrow{g}$ intersects $w(\beta, b)$. We now make a case distinction.

      \begin{itemize}
          \item If $\overrightarrow{\beta}$ does not intersect $CH(A)$ then by Lemma~\ref{lem:nointersect}, $line(\beta)$ separates $CH(A)$ and $CH(B)$.
          \item If $\overleftarrow{b}$ does not intersect $CH(A)$ then by Lemma~\ref{lem:nointersect}, $line(b)$ separates $CH(A)$ and $CH(B)$.
            \item If the edge $\overrightarrow{\beta}$ intersects $CH(A)$ in an edge $\gamma'$ that equals or precedes $\gamma$ then, per definition of $\gamma$, the halfline $\overrightarrow{\gamma'}$ intersects $w(\beta, b)$. 
            It cannot be that $\overrightarrow{\gamma'}$ intersects $\overleftarrow{\beta}$ since $\gamma'$ is already intersected by $\overrightarrow{\beta}$. And so, $\overrightarrow{b}$ intersects $\overrightarrow{\gamma'}$ (Figure~\ref{fig:intersection_finding_five}).
            In particular, this implies that $\overleftarrow{b}$ does not intersect $line(\gamma')$. 

            We now note that all edges of $CH(A)$ are contained in the halfplane bounded from above by $\gamma'$. And so, $\overleftarrow{b}$ cannot intersect any edge of $CH(A)$.  Lemma~\ref{lem:nointersect} guarantees that $\overrightarrow{b}$ cannot intersect $CH(A)$ and so $line(b)$ separates $CH(A)$ and $CH(B)$. 
            \item If the edge $\overleftarrow{b}$ intersects $CH(A)$ in an edge $g'$ that equals or succeeds $g$ then it follows by symmetry that $line(\beta)$ separates $CH(A)$ and $CH(B)$.
            \item It cannot be that $\overrightarrow{\beta}$ intersects $CH(A)$ on an edge strictly succeeding $\gamma$ \emph{and} that $\overleftarrow{b}$ intersects $CH(A)$ in an edge strictly preceding $g$. 
      \end{itemize}

\noindent
      We showed that Algorithm~\ref{alg:separationfind} either outputs a edge $\alpha$ where $line(\alpha)$ separates $CH(A)$ and $CH(B)$, or, that either $line(\beta)$ or $line(b)$ separates $CH(A)$ or $CH(B)$.
\end{proof}

\begin{figure}[h]
    \centering
    \includegraphics[page = 5]{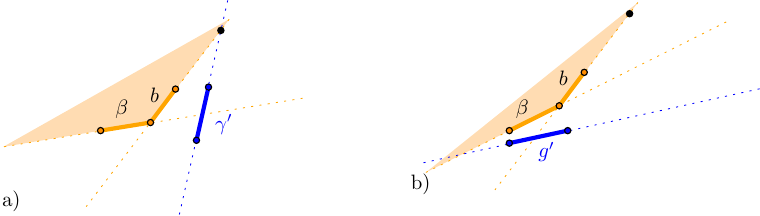}
    \caption{  (a) Let $\gamma' \in CH(A)$ be an edge where $\overrightarrow{\gamma'}$ intersects $\overrightarrow{b}$. Then $\overleftarrow{b}$ does not intersect the supporting line of $\gamma'$. However, then $\overleftarrow{b}$ cannot intersect any edge of $CH(A)$. 
    (b)  Let $g' \in CH(B)$ be an edge where $\overleftarrow{g'}$ intersects $\overleftarrow{\beta}$. 
    Then $\overrightarrow{\beta}$ does not intersect the supporting line of $g'$. 
    However, then $\overrightarrow{\beta}$ cannot intersect $CH(B)$. }
    \label{fig:intersection_finding_five}
\end{figure}

\subparagraph{A logarithmic-time robust algorithm}

Algorithms~\ref{alg:intersection_test} and \ref{alg:separationfind} have a recursive depth of $O(\log n)$. 
Thus, if each function call takes constant time then these algorithms take $O(\log n)$ time. 
For ease of exposition, we showed these algorithms using three geometric predicates:
\begin{itemize}
    \item \texttt{slope}. Given positive segments $(\alpha, \beta)$, output whether $\slope(\alpha) < \slope(\beta)$.
    \item \texttt{lies\_right}. Given two positive segments $\alpha$ and $\beta$ with different slopes, output whether the first vertex of $\beta$ lies right of $line(\alpha) \cap line(\beta)$.
    \item \texttt{wedge}. Consider a pair of positive segments $(\beta, b)$ that share a vertex and define $w(\beta, b)$ as the cone formed by their supporting halflines containing $(-\infty, \infty)$. 
    Given a positive segment $\alpha$ outside of $w(\beta, b)$, output whether $line(\alpha)$ intersects $w(\beta, b)$. 
\end{itemize}

In Section~\ref{sec:robust} we showed, under slightly different notation, that these can be implemented as robust, constant-time functions.
Thus, we showed: 

\separation*

\section{From a learned index to an indexing structure}
\label{app:indexing}

Let $S$ be a dynamic set of distinct integers in sorted order. 
Let $F_S$ denote the two-dimensional point set, obtained by mapping each $s \in S$ to $(\textsc{RANK}(s), s)$. 

\begin{definition}
         Let $\eps$ be a positive integer. We define a vertical $\eps$-cover $F$ of $S$ as set of vertically separated segments with slope at least $1$ with the following property:
         \begin{itemize}
             \item for all $p \in F_S$, the vertical line segment with width $2\eps$, centred at $p$, intersects a segment in $F$. 
         \end{itemize}
\end{definition}

Denote by $L$ the point set obtained by taking every point in $F_S$ and shifting it downwards by $\eps$, together with the point $(\infty, -\infty)$.
Denote by $U$ the point set obtained by taking every point in $F_S$ and shifting it upwards by $\eps$, together with the point $(-\infty, \infty)$.
Then a line $\ell$ is a vertical $\eps$-cover of $S$ if and only if it separates $CH(L)$ and $CH(U)$.
If follows that we may immediately adapt Theorem~\ref{eps:eps_cover} to the following:

\begin{theorem}
    \label{eps:vertical_cover}
   We can dynamically maintain a vertical $\eps$-cover $F$ of $S$ in $O(\log^2 n)$ worst-case time.
    We guarantee that there exists no vertical $\eps$-cover $F'$ of $S$ with $|F| > \frac{3}{2}|F'|$.
\end{theorem}

\noindent
$F$ stores segments from low to high, and we store segments using our segment type: \\

\noindent
\textbf{\texttt{segment} $f$}

\texttt{function} $f$ \hspace{1.8cm}\textit{the function $f(x)$ that forms the line through start and end}

\texttt{value} $start$

\texttt{value} $end$

\texttt{segment} $succ$  \hspace{1.5cm} \textit{the segment $f'$  in $F$ with the minimal start value greater than $end$}  \\

We assume that all segments in $F$ are maintained in a balanced binary tree $B(F)$, where segments are ordered by $f$.start.

\subsection{Defining our data structure}

We combine the dynamic vertical $\eps$-cover $F$ with a dynamic data structure to perform indexing queries.
Interestingly, our data structure is independent of $F$. 
The core of our data structure is our value type: \\

\noindent
\textbf{\texttt{value} $v$}

\texttt{int} $v$   \hspace{3cm}\textit{corresponding to a unique $v \in S$}

\texttt{int} id $= \lfloor \frac{v}{\eps} \rfloor$.  \\

\noindent
All values are stored in a \texttt{page}:\\

\noindent
\textbf{\texttt{vector}} $A$ is a dynamic vector that stores all non-empty pages in arbitrary order.  \\

\noindent
\textbf{\texttt{page} $p$}

\texttt{int} $p$ \hspace{3.5cm}\textit{ where $\exists v \in V$ with } $p = v.\textnormal{id}$

\texttt{vector}<\texttt{value}> $values$ \hspace{0.8cm}\textit{storing all $v \in V$ with } $p = v.\texttt{id}$ \textit{in sorted order} 

\texttt{int} $prev$ \hspace{3cm}\textit{the index in $A$ of the maximum $p' \in A$ with $p' < p$} 

\texttt{int} $succ$ \hspace{3cm}\textit{the index in $A$ of the minimum $p'' \in A$ with $p < p''$}  \\

\noindent
$H : 
\texttt{int} \rightarrow \texttt{int}$ is a hash map that for any ID $p$ returns the index $i$ such that $p = A[i]$.

\subsection{Answering queries using our data structure}

Given this data structure, we first show how to support $\texttt{member}$  and queries (Algorithm~\ref{alg:member}).
Then, we show how to answer \texttt{predecessor} and $\texttt{rank}$ queries~\ref{alg:predecessor}.
Finally, we show how to answer \texttt{range} queries.

\begin{lemma}
    \label{lem:member}
    Our data structure supports \texttt{member} queries in $O(\eps)$ time. 
\end{lemma}

\begin{proof}
    Using $H$, Algorithm~\ref{alg:member} may access the page $p$ containing all values $v'$ with $v'$.id $=$ $v$.id in constant time. 
    Since all values in $S$ are distinct, this page contains at most $\eps$ values and so the lemma follows.     
\end{proof}

 \begin{algorithm}[h]
    \caption{\texttt{member}(\texttt{value} $v$)}
    \label{alg:member}
    \begin{algorithmic}
    \State $i \gets$ $H(\lfloor  \frac{v}{\eps} \rfloor)$
    \State \Return  $A[i]$.values.contains($v$)
    \end{algorithmic}
  \end{algorithm}

\begin{lemma}
    \label{lem:predecessor}
    Our data structure supports \texttt{rank} and \texttt{predecessor} queries in $O(\eps + \log |F|)$ time. 
\end{lemma}

\begin{proof}
    Consider first the special case where $v$ is smaller than all values in $S$. Then its rank is equal to $0$  and its predecessor is the first element of $S$. 
    
    Otherwise, let $f \in F$ such that $v$ is greater than $f$.start, but less that the start of the successor of $f$ in $F$ (in the special case where $f$ is the last segment in $F$, we always say that $v$ is less than the start of its successor). 

    We again consider a special case, where $v$ exceeds the value at the end of $f$. 
    In this case, the horizontal line through $v$ lies in between $f$ and its successor in $F$. The predecessor of $v$ must therefore be the last value in $f$.    
    This value is stored in the page $p$ at $A[i]$ for $i = H(f.\textnormal{end.id})$. 
    The rank of $v$ is then the rank of its predecessor. 

    If no special case applies then the predecessor $s$ of $v$ has a rank $r =\lfloor f^{-1}(v) \rfloor$. 
    We want to find the index $i$ such that $A[i]$ stores the page $p$ containing $v$. 
    Consider the point $(r, b) := (r, f(r))$. 
    Per definition of a vertical $\eps$-cover, the point $(r, s) \in F_S$ corresponding to the predecessor $s$ of $v$ lies within vertical distance $\eps$ of $(r, b)$. 
    It follows that:
    \[
    s.\textnormal{id} =  \left\lfloor \frac{b - \eps}{\eps} \right\rfloor \textnormal{ or }  \left\lfloor \frac{b }{\eps} \right\rfloor \textnormal{ or }  \left\lfloor \frac{b + \eps}{\eps} \right\rfloor \quad \Rightarrow \quad i =  H\left(
     \left\lfloor \frac{v}{\eps} \right\rfloor - 1 \right) \textnormal{ or  }  H\left( \left\lfloor \frac{v}{\eps} \right\rfloor \right) \textnormal{ or } H\left( \left\lfloor \frac{v}{\eps} \right\rfloor + 1 \right)
    \]
    We check in constant time which $i$ we need to choose. We index $A[i]$ to find the page containing $s$. We then iterate over all values in that page to find $s$.    
\end{proof}

\begin{lemma}
    \label{lem:range}
    Our data structure supports \texttt{range} queries in $O(k + \eps + \log |F|)$ time where $k$ denotes the output size.
\end{lemma}

\begin{proof}
Given two values $u$ and $v$ with $u < v$, we find the predecessors $A[i]$.values[$j$] and $A[i']$.values[$j'$] of $u$ and $v$ in $O(k + \log |F|)$ time.

Each page $p$ contains all values in the range $[\eps p, \eps p + \eps - 1]$ in sorted order so within a page we can output all values in the range in output-sensitive time. 
Recall that a page $p$ stores the integers $s, r$ where $A[s]$ is the page preceding $p$ and $A[r]$ is the page succeeding $p$ in their sorted order. 
So we simply iterate over this doubly linked list implementation, iterate over all values in these pages to output the range $S \cap [u, v]$ in output-sensitive time. 
\end{proof}

    \begin{algorithm}[h]
    \caption{\texttt{index}(\texttt{int} $v$ )}
    \label{alg:index}
    \begin{algorithmic}
    \State $(a, b, c) \gets (H( \lfloor \frac{v}{\eps} \rfloor - 1), H( \lfloor \frac{v}{\eps} \rfloor) , H( \lfloor \frac{v}{\eps} \rfloor) + 1)$
    \State \Return $arg \, \min \{ |A[a].values.first - v|, |A[b].values.first - v|, |A[c].values.first - v|\}$
    \end{algorithmic}
  \end{algorithm}

 \begin{algorithm}[h]
    \caption{\texttt{predecessor}(\texttt{value} $v$, $B(F)$ )}
    \label{alg:predecessor}
    \begin{algorithmic}
        \If {$v < $ $F$.first.start}
    \State \Return $(0, 0)$
    \EndIf
    \State $f \gets$ segment in $F$ where $v \geq f$.start and $v < f$.succ.start
    \If{$v \geq f$.end}
    \State $i = H(f.\textnormal{end.id})$
    \Else 
    \State $i =$ \texttt{index}$(f(\lfloor f^{-1}(v) \rfloor ))     $ \EndIf    
    \State $j = A[i]$.values.size() $- 1$
    \While{$A[i]$.values$[j] > v$}
    \State decrement $j$
    \EndWhile
    \State \Return $(i, j)$
    \end{algorithmic}
  \end{algorithm}

 \begin{algorithm}[H]
    \caption{\texttt{range}(\texttt{value} $u$, \texttt{value} $v$, $B(F)$ )}
    \label{alg:range_report}
    \begin{algorithmic}
    \State $(i, j) \gets$ \texttt{predecessor}$(u, B(F))$
    \State $(i', j') \gets$ \texttt{predecessor}$(v, B(F))$
    \If{$i = i'$}
    \State \texttt{report}$(A[i]\textnormal{.values}[j + 1, j' + 1] )$
    \Else
    \State \texttt{report}$(A[i]\textnormal{.values}[j + 1, A[i]\textnormal{.values.size()}] )$
    \State $k \gets A[i]$.succ
    \While{$A[k] < A[i']$}
    \State \texttt{report}$(A[k]\textnormal{.values})$
    \State $k \gets A[k]$.succ
    \EndWhile
    \State \texttt{report}$(A[i']\textnormal{.values}[1, j' + 1] )$
    \EndIf
    \end{algorithmic}
  \end{algorithm}

  \subsection{Updating our data structure}

  \begin{lemma}
      Let $v$ be a value in $\mathcal{U}$. Given \texttt{predecessor}($v$), we may update our data structure in $O(\eps)$ time.
  \end{lemma}

  \begin{proof}
      Algorithms~\ref{alg:insertion} and~\ref{alg:delete} show how to insert values into and delete values from $S$ respectively.
        Let $(i, j)$ be such that $A[i]$.values[$j$] is the predecessor of $v$. 
        For insertions, we consider three cases. 
        In the first case, $A[i]= v$.id. Then $v$ needs to be inserted into the page at $A[i]$, succeeding $v' =$ $A[i]$.values[$j$]. 
        We achieve this in $O(\eps)$ time by first appending $v$ to $A[i]$.values, and then swapping $v$ with its predecessor in the vector until it is incident to $v'$.
        In the second case, $A[i] < v$.id but the successor page of $A[i]$ matches $v$.id. We obtain this successor by indexing $A[A[i]\textnormal{.succ}]$ and then insert $v$ into the values vector in the same way. 
        Finally, it may be that there exists no page $p \in A$ where $p = v$.id. 
        We create a new page, append it to $A$, add it to the hash map and insert it into our doubly linked list implementation. 
        Deletions are handled analogously.      
  \end{proof}

And so, we may conclude:

\begin{theorem}\label{thm:main}
    For any $\eps$, there exists a data structure to dynamically maintain a vertical $\eps$-cover $F$ of a dynamic set of distinct integers in $O(\eps + \log^2 n)$ time. 
    We guarantee that there exists no vertical $\eps$-cover $F'$ with $|F| > \frac{3}{2}|F'|$. The  
    data structure supports indexing queries in $O(\eps + \log |F|)$ expected time and range queries in additional $O(k)$ time where $k$ is the output size. 
\end{theorem}
  
 \begin{algorithm}[h]
    \caption{insert(\texttt{value} $v$, $B(F)$)}
    \label{alg:insertion}
    \begin{algorithmic}
    \State $(i, j)  \gets$ \texttt{predecessor}($v$, $B(F)$)
    \If{$A[i] = v.id$}
    \State \texttt{insert\_into\_page}($i$, $v$)
    \ElsIf{$A[ A[i]\textnormal{.succ}] = v\textnormal{.id}$}
    \State \texttt{insert\_into\_page}($A[i]\textnormal{.succ}$, $v$)
    \Else
    \State $i' \gets A$.size()
    \State \texttt{page} $p \gets$ \texttt{new page}$(v)$
    \State $p$.prev $\gets P[i]$
    \State $p$.succ $\gets P[i]$.succ
    \State $A[p\textnormal{.prev}]$.succ $\gets i'$
    \State $A[p\textnormal{.succ}]$.prev$ \gets i'$
    \State $A[i'] \gets p$
    \State $H$.insert($v.id$, $i'$)
    \EndIf
    \end{algorithmic}
  \end{algorithm}

 \begin{algorithm}[h]
    \caption{\texttt{insert\_into\_page}(\texttt{int} $i$, \texttt{value} $v$)}
    \label{alg:insert_into_page}
    \begin{algorithmic}
    \State $j = A$.values.size()
    \State $A\textnormal{.values}[j] \gets v$
    \While{$j > 0$ and $A[i]\textnormal{.values}[j-1] > A[i]\textnormal{.values}[j]$}
    \State $A[i]$.values.Swap($j-1$, $j$)
    \State decrement $j$
    \EndWhile    
    \end{algorithmic}
  \end{algorithm}

 \begin{algorithm}[h]
    \caption{\texttt{delete}(\texttt{value} $v$, $B(F)$)}
    \label{alg:delete}
    \begin{algorithmic}
    \State $(i, j)  \gets$ \texttt{predecessor}($v$, $B(F)$)
    \If{$A[i] = v$.id}
    \State {\texttt{delete\_from\_page}}($i$, $j+1$, $v$)
    \Else
    \State {\texttt{delete\_from\_page}}($A[i]$.succ, $0$, $v$)
    \EndIf
    \end{algorithmic}
  \end{algorithm}

 \begin{algorithm}[h]
    \caption{\texttt{delete\_from\_page}(\texttt{int} $i$, \texttt{int} $j$, \texttt{value} $v$)}
    \label{alg:delete_from_page}
    \begin{algorithmic}
    \If{$A[i]\textnormal{.values}[j] == v$}
    \For{$t \in [j, A[i]$.values.size() $- 1]$}
    \State $A[i]$.values.Swap($t$, $t+1$)
    \EndFor
    \State $A[i]$.reduce()
    \EndIf
     \If{$A[i]$.values.size() $= 0$}
         \State $A[A[i]\textnormal{.prev}]\textnormal{.succ} \gets A[i]$.succ
    \State $A[A[i]\textnormal{.succ}]\textnormal{.prev} \gets A[i]$.prev
    \State $A$.Swap($i$, $A$.size() $- 1$)
    \State $H$.insert($A[i], i$)
    \State $H$.delete$(v\textnormal{.id}, i)$
     \EndIf 
    \end{algorithmic}
  \end{algorithm}

\newpage

\section{Remarks on computing a PGM-index or $\eps$-cover}
\label{app:orourke}

We briefly note that Ferragina and Vinciguerra~\cite{ferragina2020pgm} wrongfully claim that the algorithm by O'Rourke~\cite{o1981line} computes a PGM-index of minimum complexity. 
In~\cite{o1981line}, O'Rourke considers the following problem (we flip the $x$ and $y$-axis to align with~\cite{ferragina2020pgm})

The input is a set of \emph{data ranges}, i.e., horizontal segments in the place, sorted by $y$-coordinate. The goal is to compute a $y$-monotone polyline that intersects all data ranges. O'Rourke assumes a streaming setting, which adds horizontal segments in sorted order from low to high. 
This algorithm maintains a line segment from $p$ that stabs \emph{all} received horizontal segments. The algorithm tests after an insertion whether such a segment still exists. If not, it outputs the previous line segment and deletes all horizontal segments (apart from the current insertion). 
This algorithm exists in a streaming setting. However,~\cite{ferragina2020pgm} is a normal static setting where this algorithm is not optimal!

The algorithm by O'Rourke can maintain a horizontal $\eps$-cover in the following way: for each insertion into $G$, test whether there exists a segments intersecting all segments in $G$. If there does not exist such a segment, output the previous segment and delete all segments from $G$ (apart from the segment that was just inserted). 
This algorithm is also not optimal in a static setting, as a similar adversarial input can be constructed. However, it can easily be shown that this algorithm is a $\frac{3}{2}$ approximation:

Indeed, this algorithm maintains the invariant that for all consecutive segments $\lambda, \lambda'$ in the output, the input $I$ corresponding to these segments is \emph{blocked}. I.e., there exists no segment $\lambda^*$ that intersects all segments in $I$. This implies that for any three consecutive segments in the output, any other algorithm must include at least two segments and so a $\frac{3}{2}$ approximation follows.

\newpage
\section{Additional Experiments}
\label{app:experiments}

We here describe our full suite of experiments, including additional real world datasets and variations of the experiments shown in Section~\ref{sec:experiments}.
We've used the following datasets, each consisting of unique 8 byte integers in randomly shuffled order:
\begin{itemize}
    \item \textbf{LINES} is a synthetic data set of 5M integers that, in rank space, produces 5 lines of exponentially increasing slope. This set models the ideal scenario for a PGM index. 
    \item \textbf{LONGITUDE} is a real world data set that contains the longitudes of roughly 246M points of interest from OpenStreetMap, over the region of Italy. 
    This data is thereby inherently of geometric nature.
    This data set was used in both~\cite{ferragina2020pgm}  and \cite{kipf2019sosdbenchmarklearnedindexes}. We follow~\cite{ferragina2020pgm} and convert the data to integers by removing the decimal point from the raw longitudes.
    \item \textbf{UNIF} originates from~\cite{ferragina2020pgm}. It is a synthetic data set, containing a uniform random sample of 50M integers from $(0,10^{11})$. We adapt this data set to our dynamic setting.
    \item \textbf{DRIFTER} is a real world data set, containing roughly 1.7M steps of accumulated distance travelled by ocean drifters tracked through GPS. Used previously as a subtrajectory clustering benchmark~\cite{conradi2023}.
    \item \textbf{BOOK}, from~\cite{books200m}, contains Amazon book sale popularity data. It was used in a benchmarking paper on learned indices~\cite{kipf2019sosdbenchmarklearnedindexes}. We use a truncated sample of 100M integers.
    \item \textbf{BERLIN, CHICAGO}, datasets each containing roughly $150.000$ GPS samples from car traffic in these cities. Used often as benchmarks for road map construction~\cite{ahmed2015}. We use the accumulated distance travelled as input.
    \item \textbf{TIMESERIES}, from~\cite{m42020}, contains a concatenation of $15.000$ steps of time series from various real world domains. Originally used as part of the M4 forecasting competition.
\end{itemize}

\subsection{Learned index complexity}

As a measure of the complexity of the learned indices, we count the number of lines stored in the structures during construction. In Figure~\ref{fig:linecounts} the measurements across all our data can be seen. As the data arrives shuffled, even if the total data set has some underlying structure, it might not come to light until more data arrives. This is easily seen on the synthetic \textsc{Lines} data, where the structure starts decreasing in complexity roughly halfway for the dynamic PGM. Since the logarithmic PGM has the data split across buckets, the complexity grows until close to the end where there is a sudden drop.

Although geometric in nature, obtained as accumulated distances travelled, the \textsc{Drifter}, \textsc{Berlin}, and \textsc{Chicago} data sets do not show trends that indicate favorable conditions for the examined learned indices. 
\begin{figure}
    \centering
    \includegraphics[width=0.49\linewidth]{Figures/Results/linecount_lines.png}
    \includegraphics[width=0.49\linewidth]{Figures/Results/linecount_longitude.png}
    \includegraphics[width=0.49\linewidth]{Figures/Results/linecount_unif.png}
    \includegraphics[width=0.49\linewidth]{Figures/Results/linecount_drifter.png}
    \includegraphics[width=0.49\linewidth]{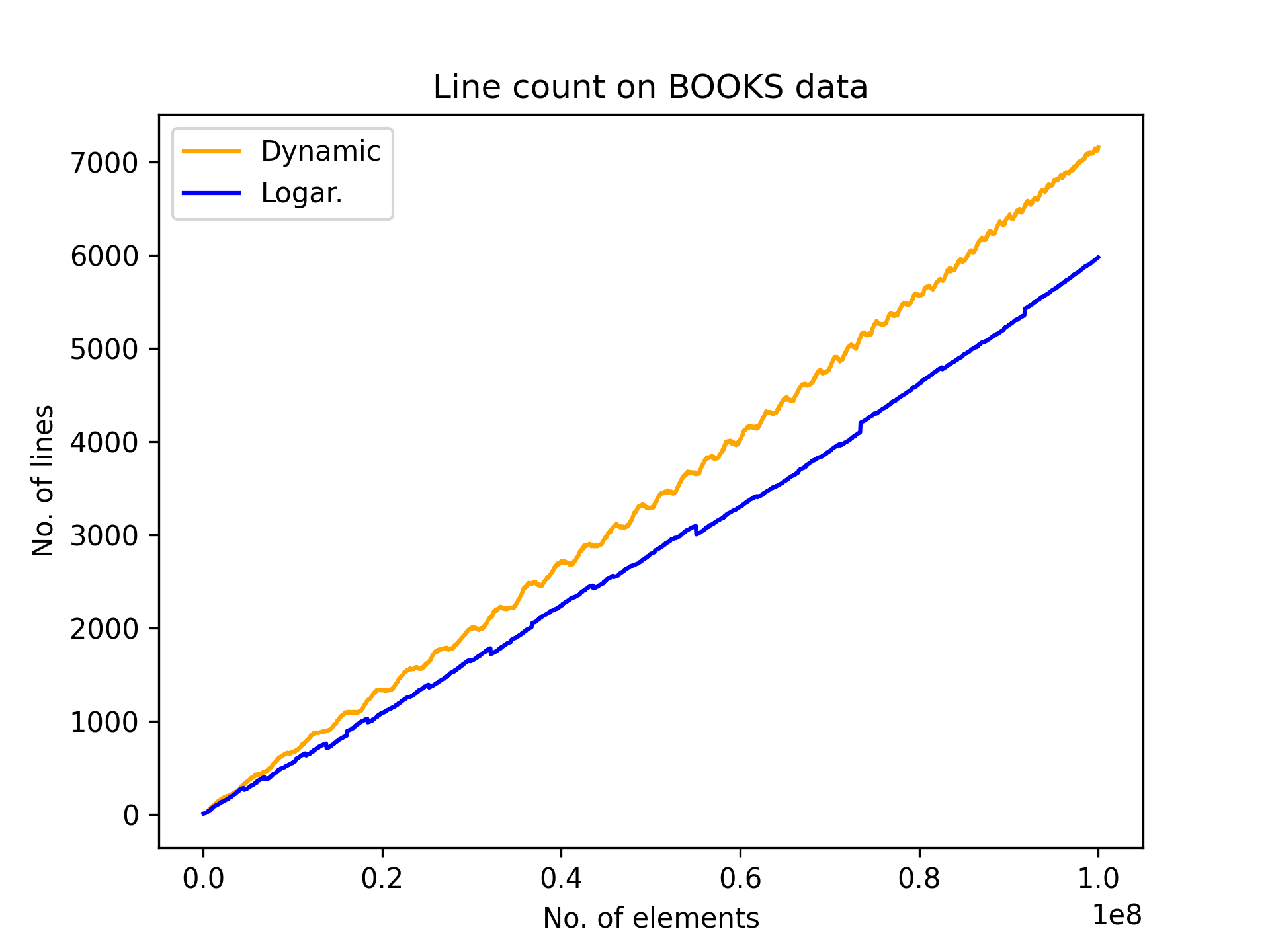}
    \includegraphics[width=0.49\linewidth]{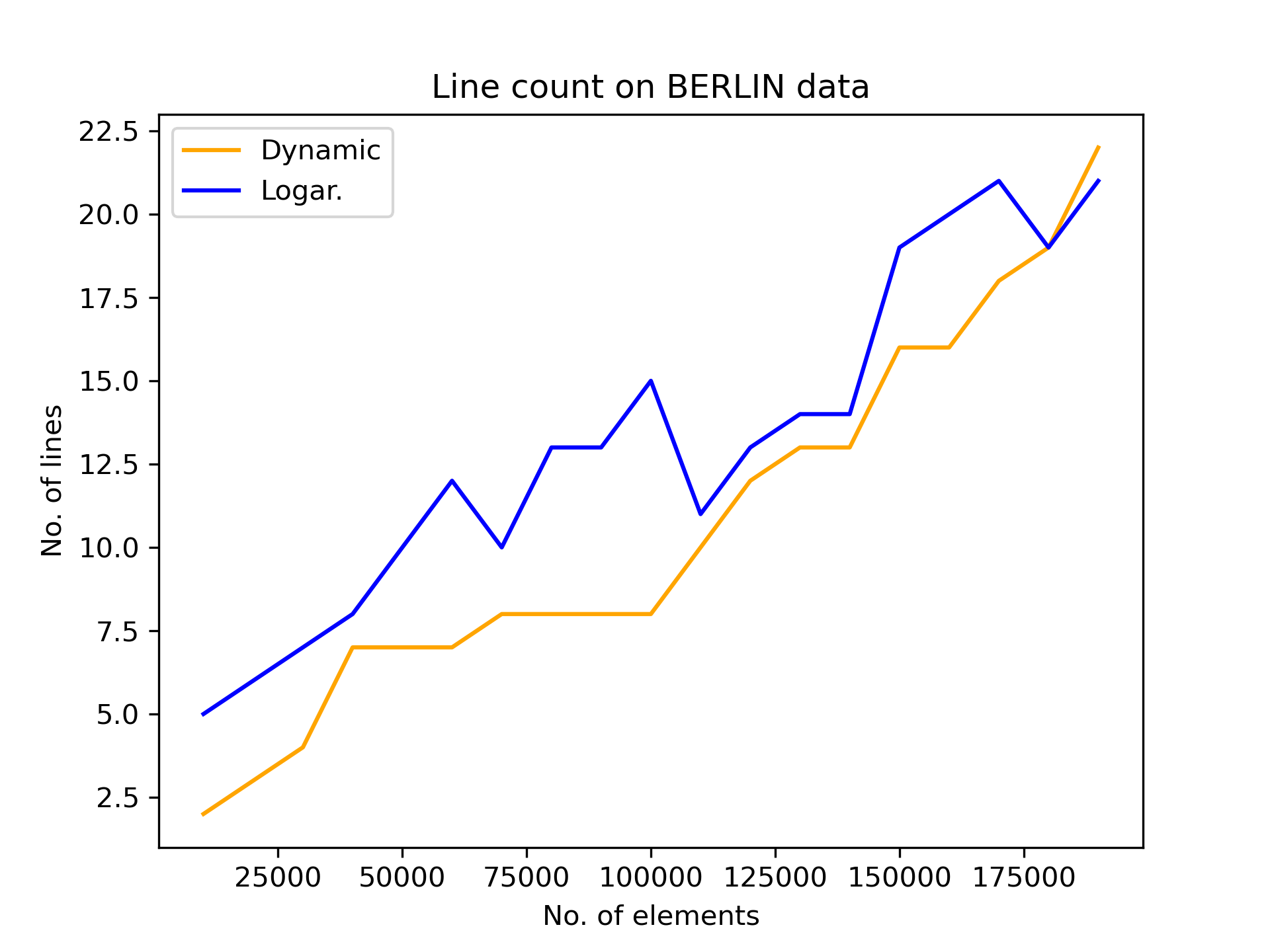}
    \includegraphics[width=0.49\linewidth]{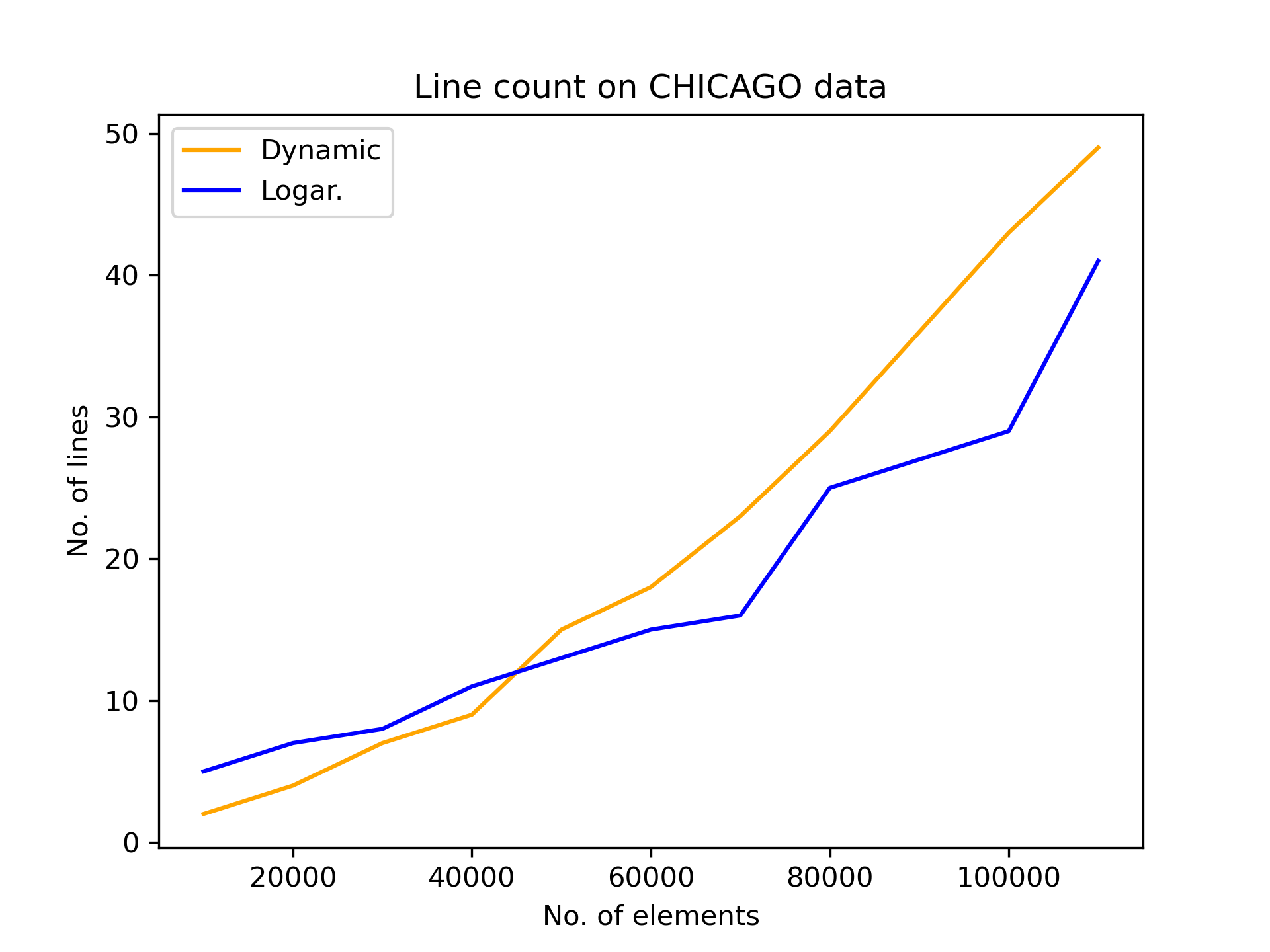}
    \includegraphics[width=0.49\linewidth]{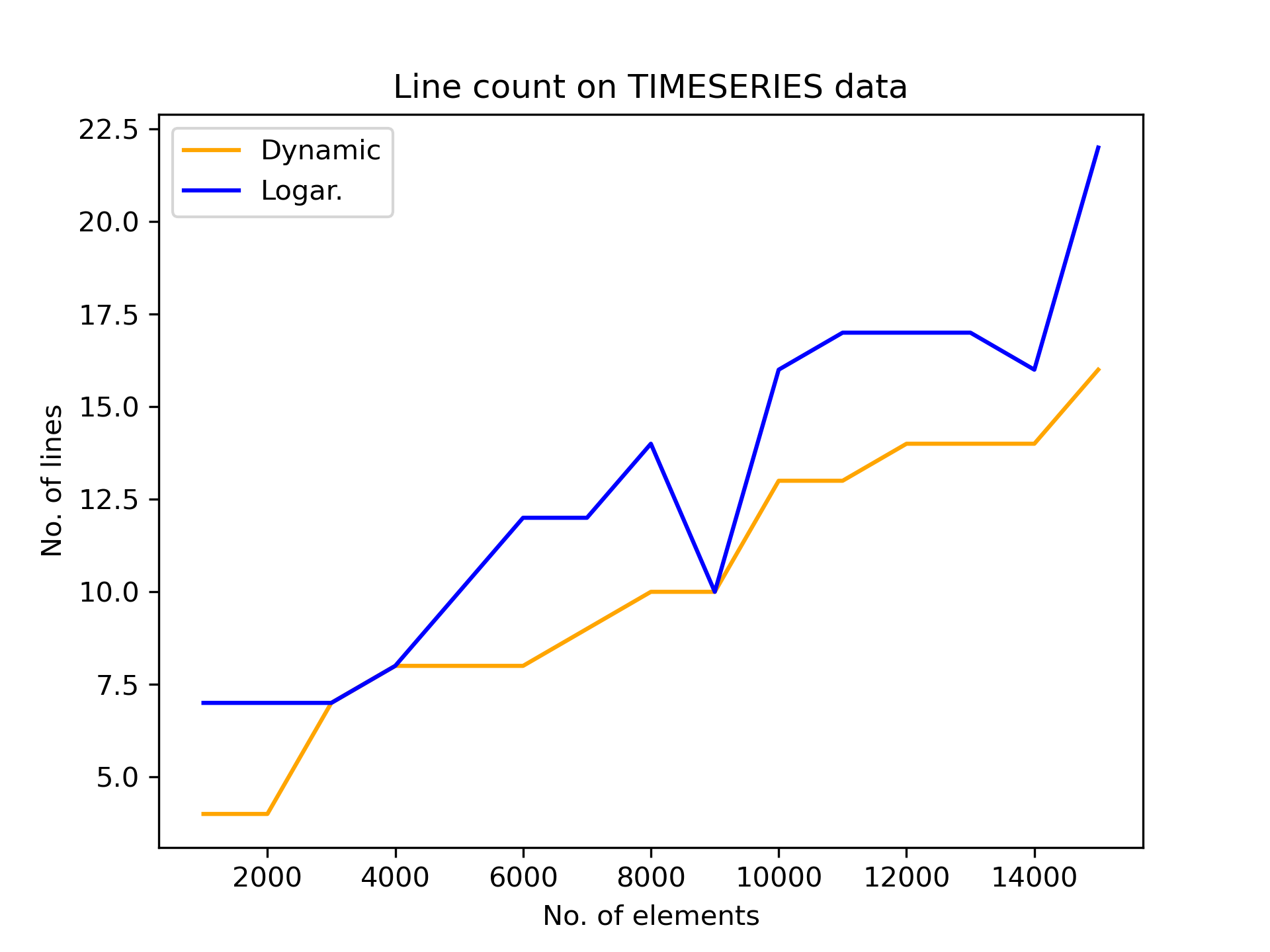}
    \caption{The number of lines stored in the learned indices, insertion only.}
    \label{fig:linecounts}
\end{figure}

\subsection{Maintaining an indexing data structure}

Here we show the full results from varying query ratio and prior deletion ratio on all of the datasets described above.

\subparagraph{The learned index.}
When measuring the running times of maintaining the learned index, the key observation is that the smaller structured datasets \textsc{Drifter}, \textsc{Berlin}, and \textsc{Chicago} have much closer performance than the larger \textsc{Longitude} data. As is the case also on the unstructured data, the performance penalties incurred by the poor memory access patterns of the convex hull data structures becomes increasingly prohibitive once data surpasses a certain size. 

\begin{figure}[h]
    \centering
    \includegraphics[width=0.32\linewidth]{Figures/Results/Maintenance/italy_maintenance.png}
    \includegraphics[width=0.32\linewidth]{Figures/Results/Maintenance/ran_maintenance.png}
    \includegraphics[width=0.32\linewidth]{Figures/Results/Maintenance/drifter_maintenance.png}
    \includegraphics[width=0.24\linewidth]{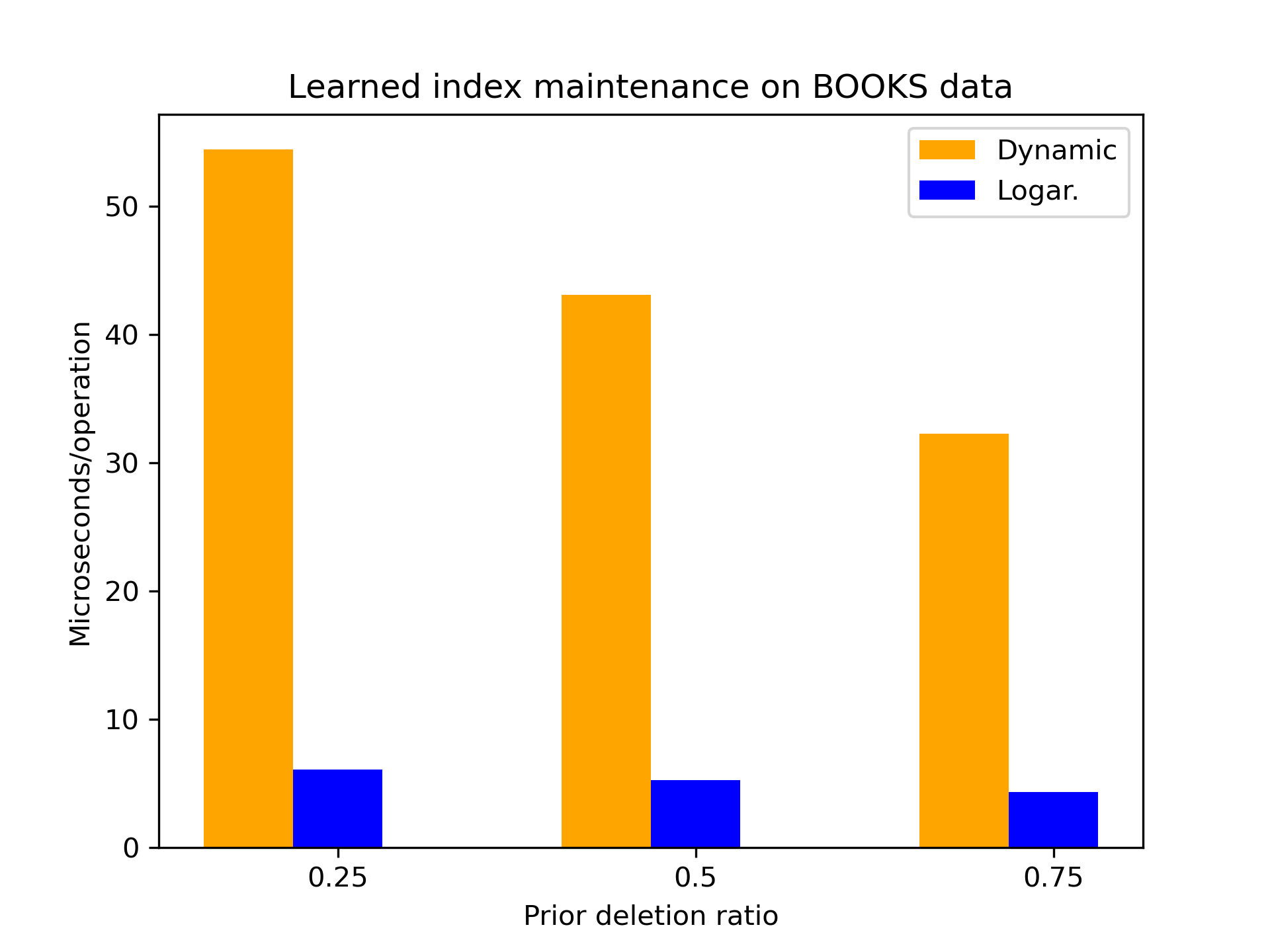}
    \includegraphics[width=0.24\linewidth]{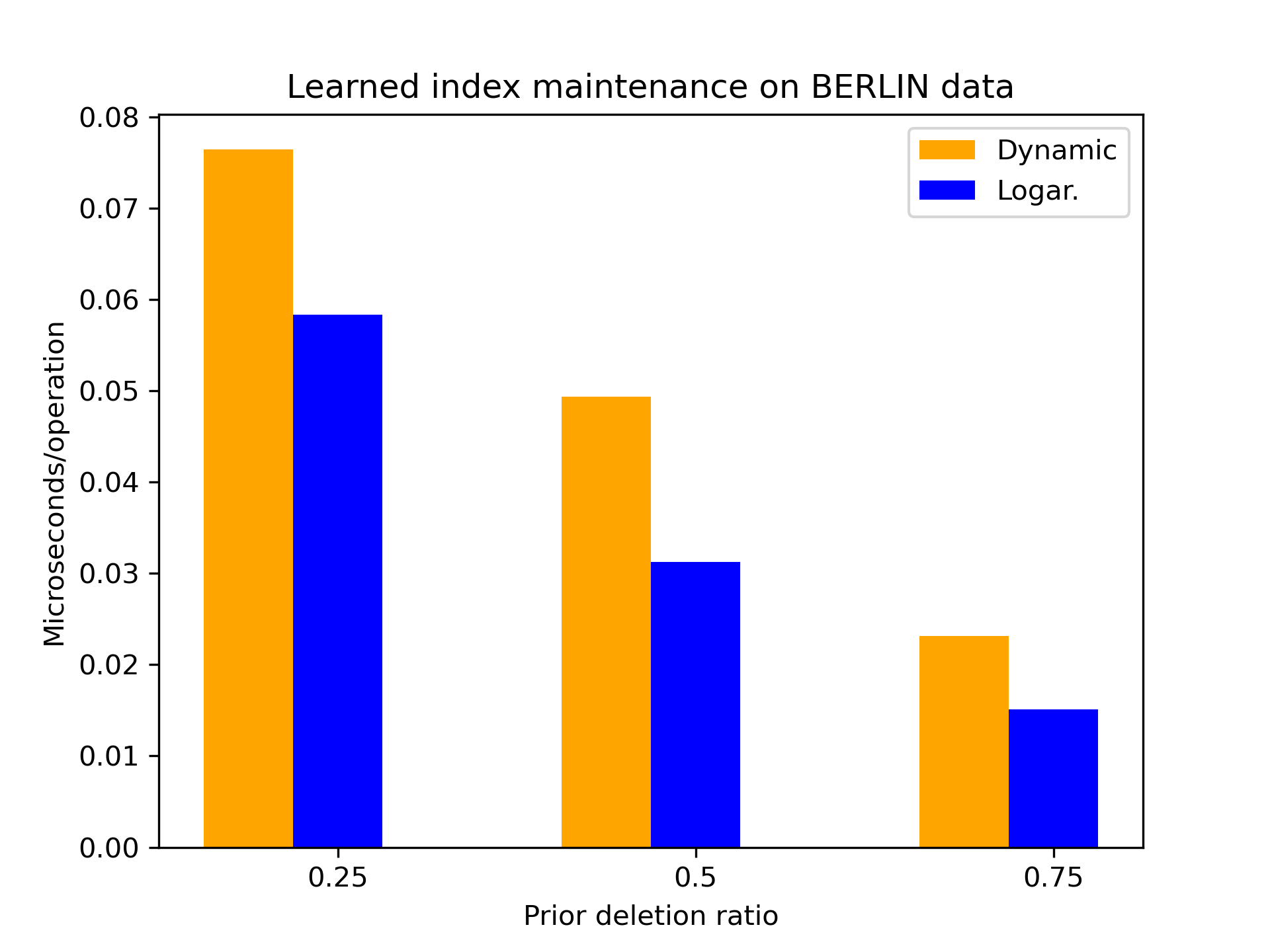}
    \includegraphics[width=0.24\linewidth]{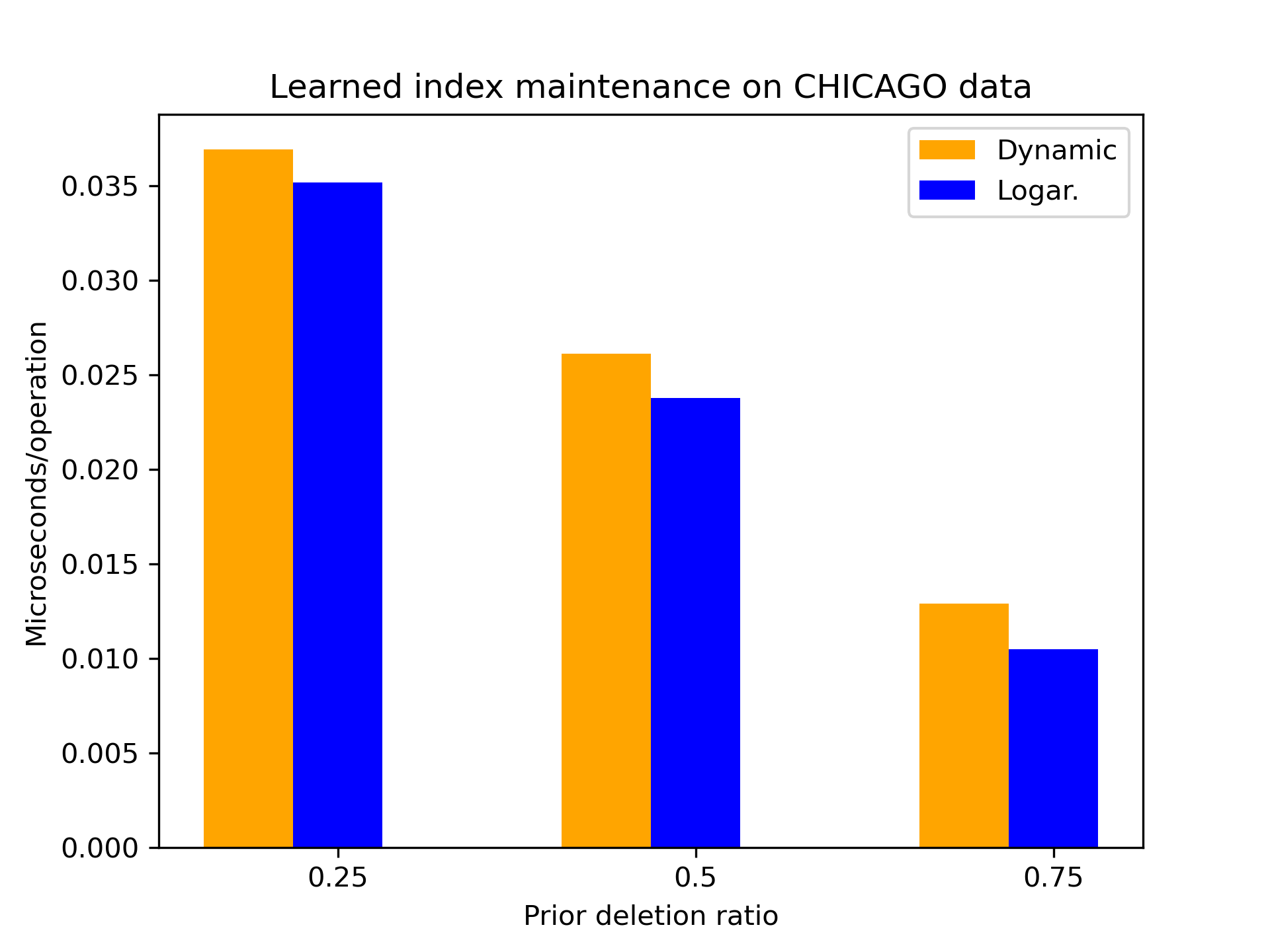}
    \includegraphics[width=0.24\linewidth]{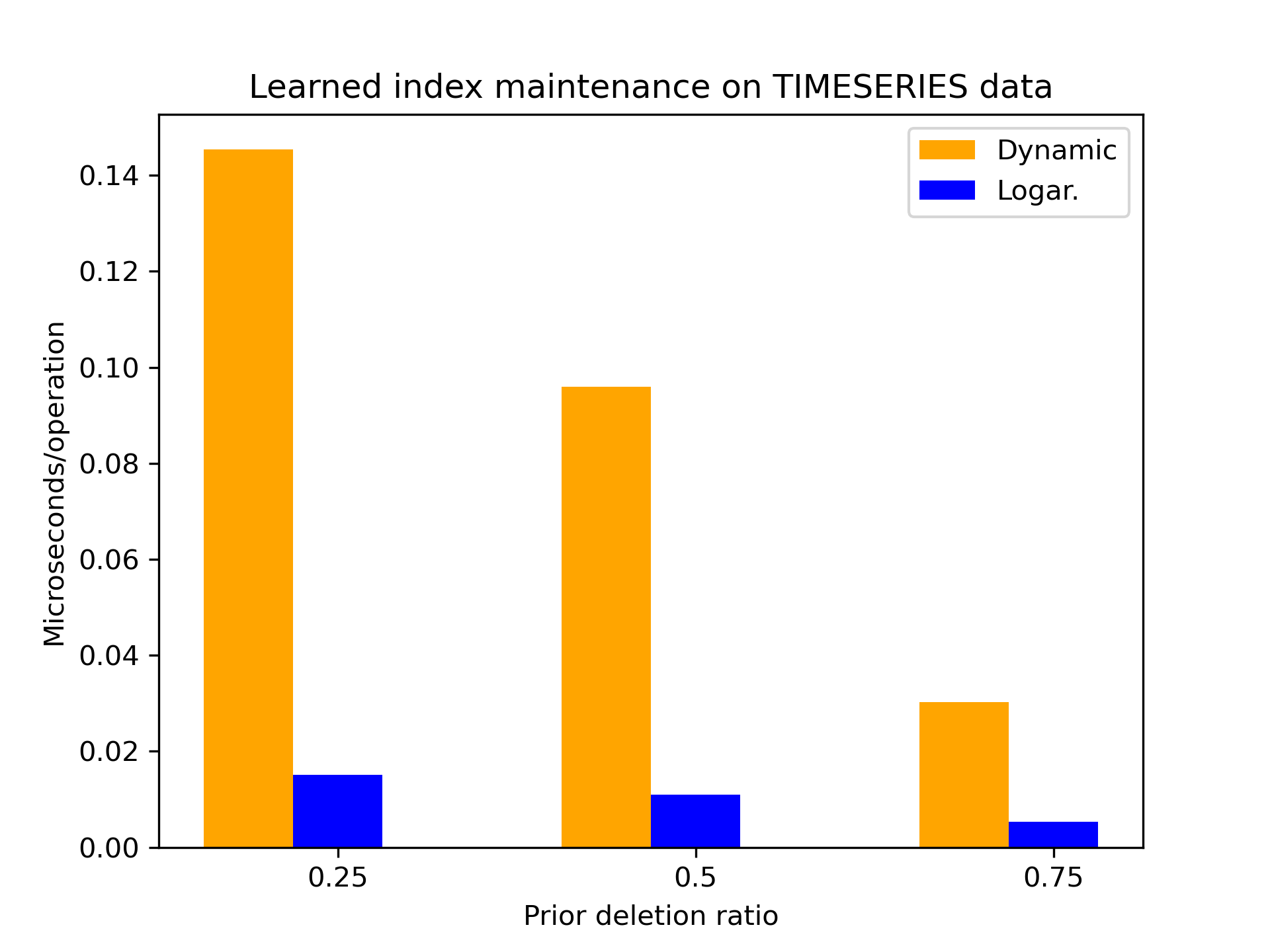}
    \caption{Update times for maintaining the learned index dynamically. }
    \label{fig:app_maintenance}
\end{figure}

\subparagraph{The indexing data structure.}
For the larger datasets, \textsc{Longitude}, \textsc{Unif}, and \textsc{Book}, we see some interesting trends. \textsc{Longitude} is the only structured large set, where we see an increase of performance from the dynamic PGM as query ratio increases, whereas it decreases for the logarithmic PGM. In the unstructured cases, both structures occur more time spent for larger query ratios. For the smaller datasets, performance is almost identical. The overhead associated with maintaining the learned index becomes much smaller, since it can now be manipulated with less bad effects on cache. Still, memory access is the deciding factor, and for many cases the logarithmic PGM comes out on top.

As we increase the number of previous deletions, the gap in performance between indexing structures lessens. This somewhat ties into the previous notes on memory, but we also see that for the larger query ratios performance evens out, with the dynamic PGM only a sliver from the logarithmic PGM on structured data.

The adversarial scenario is taking the above to its extreme, and the dynamic then outperforms the logarithmic on all data.

\begin{figure}
    \centering
    \includegraphics[width=0.32\linewidth]{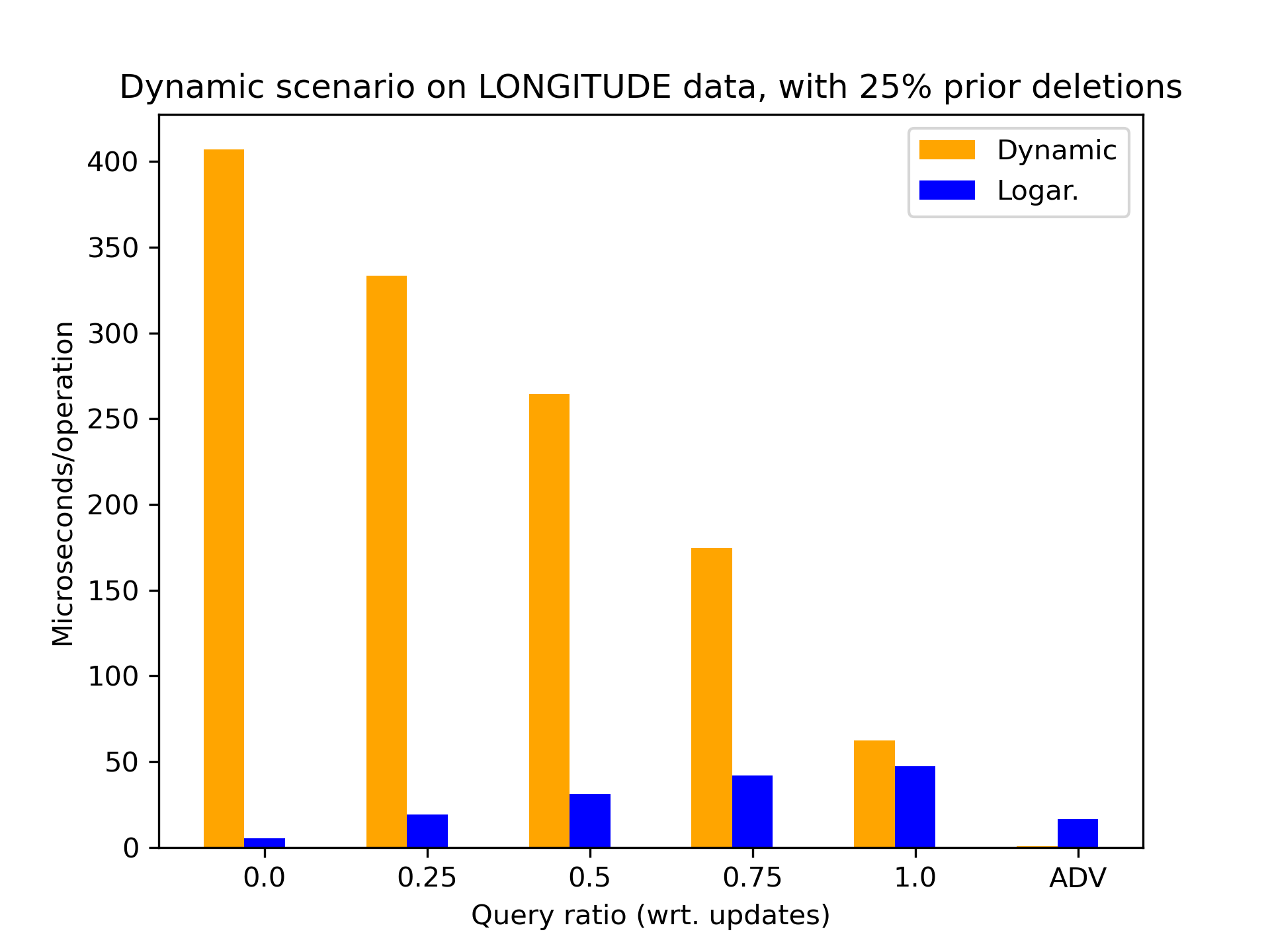}
    \includegraphics[width=0.32\linewidth]{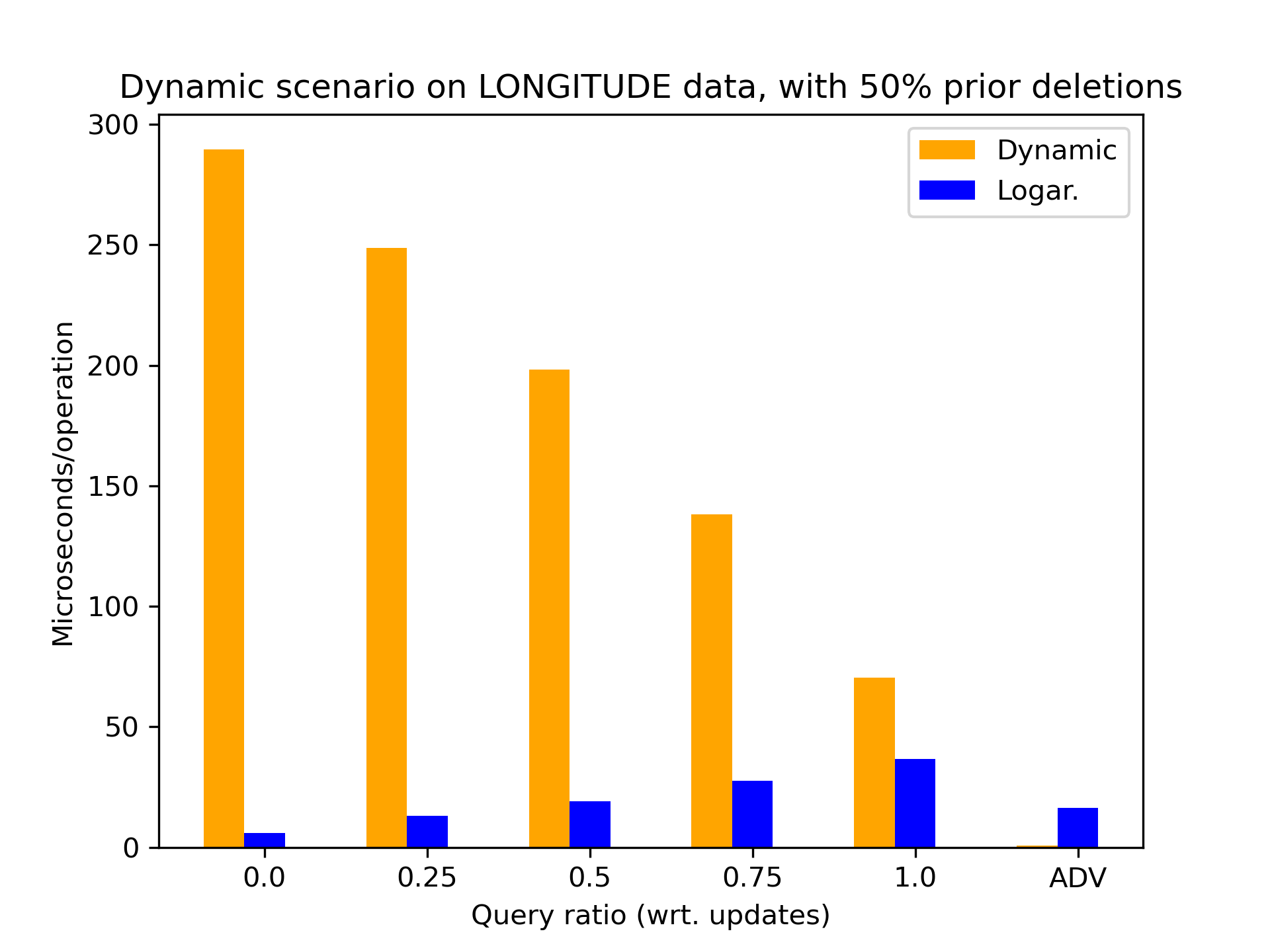}
    \includegraphics[width=0.32\linewidth]{Figures/Results/Updates/italy_updates0.75.png}
    \caption{Time per operation in dynamic scenario with varying query ratio and prior deletions on \textsc{Longitude} data}
    \label{fig:updates_long}
\end{figure}

\begin{figure}
    \centering
    \includegraphics[width=0.32\linewidth]{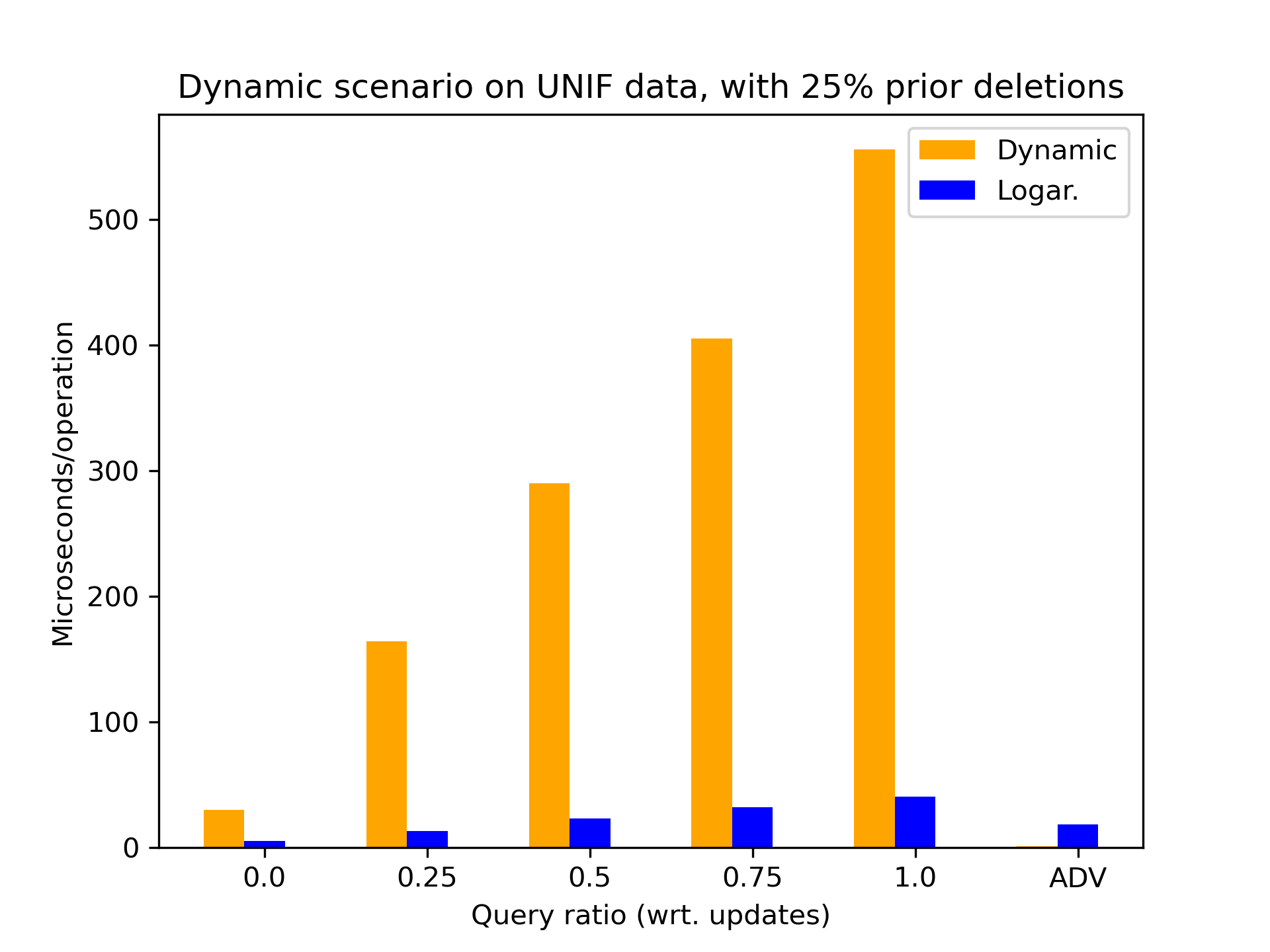}
    \includegraphics[width=0.32\linewidth]{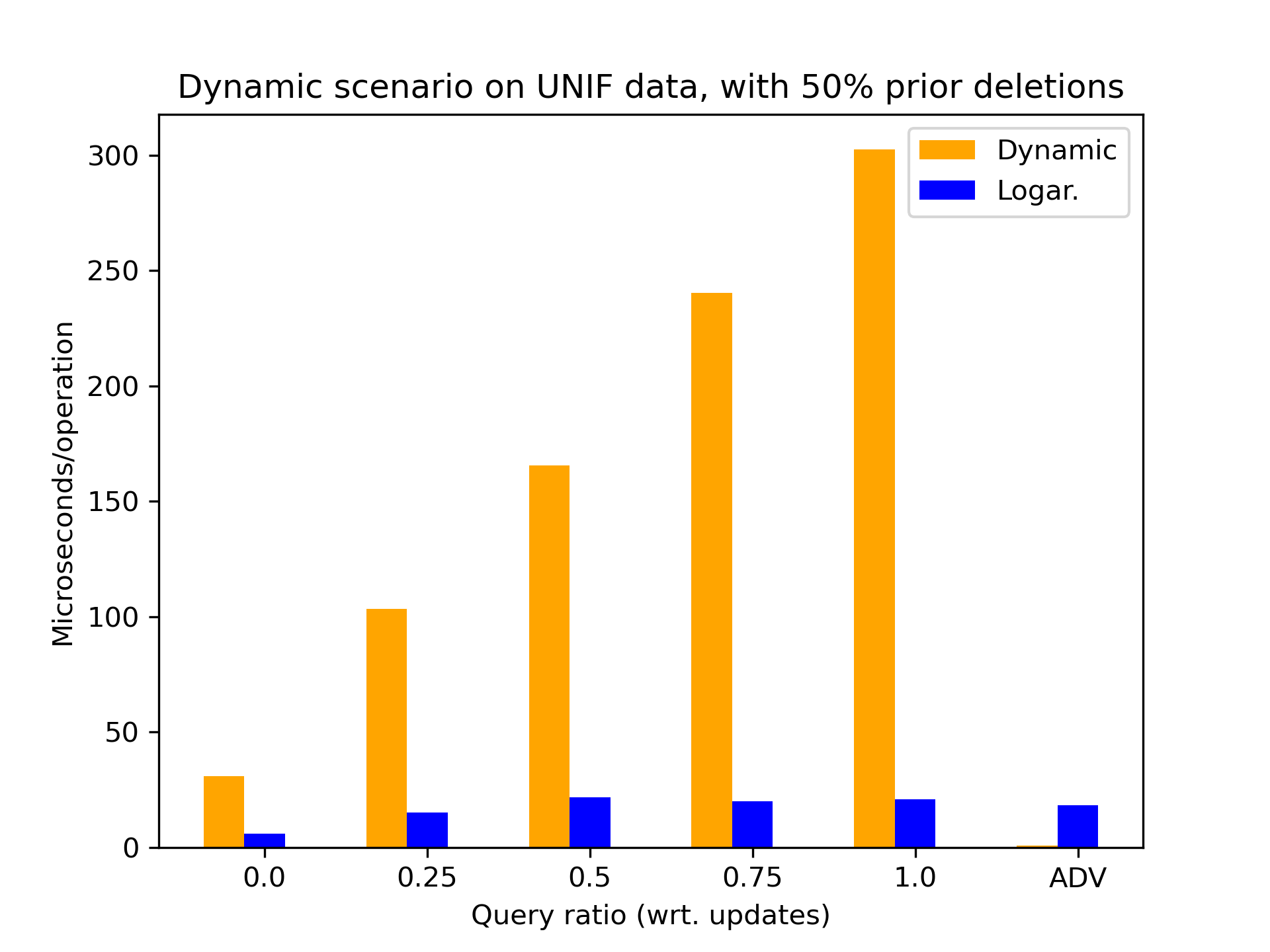}
    \includegraphics[width=0.32\linewidth]{Figures/Results/Updates/ran_updates0.75.png}
    \caption{Time per operation in dynamic scenario with varying query ratio and prior deletions on \textsc{Unif} data}
    \label{fig:updates_unif}
\end{figure}

\begin{figure}
    \centering
    \includegraphics[width=0.32\linewidth]{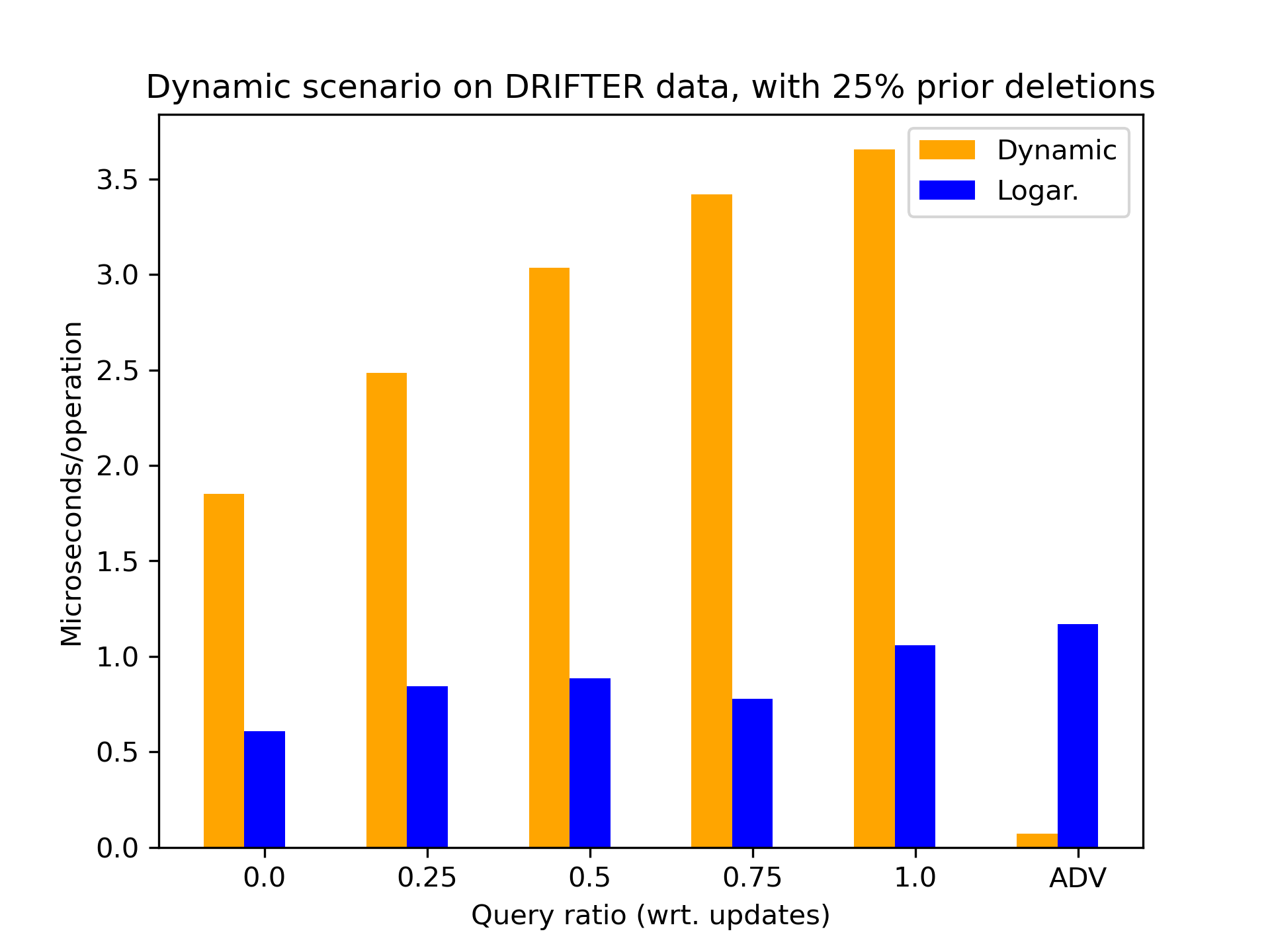}
    \includegraphics[width=0.32\linewidth]{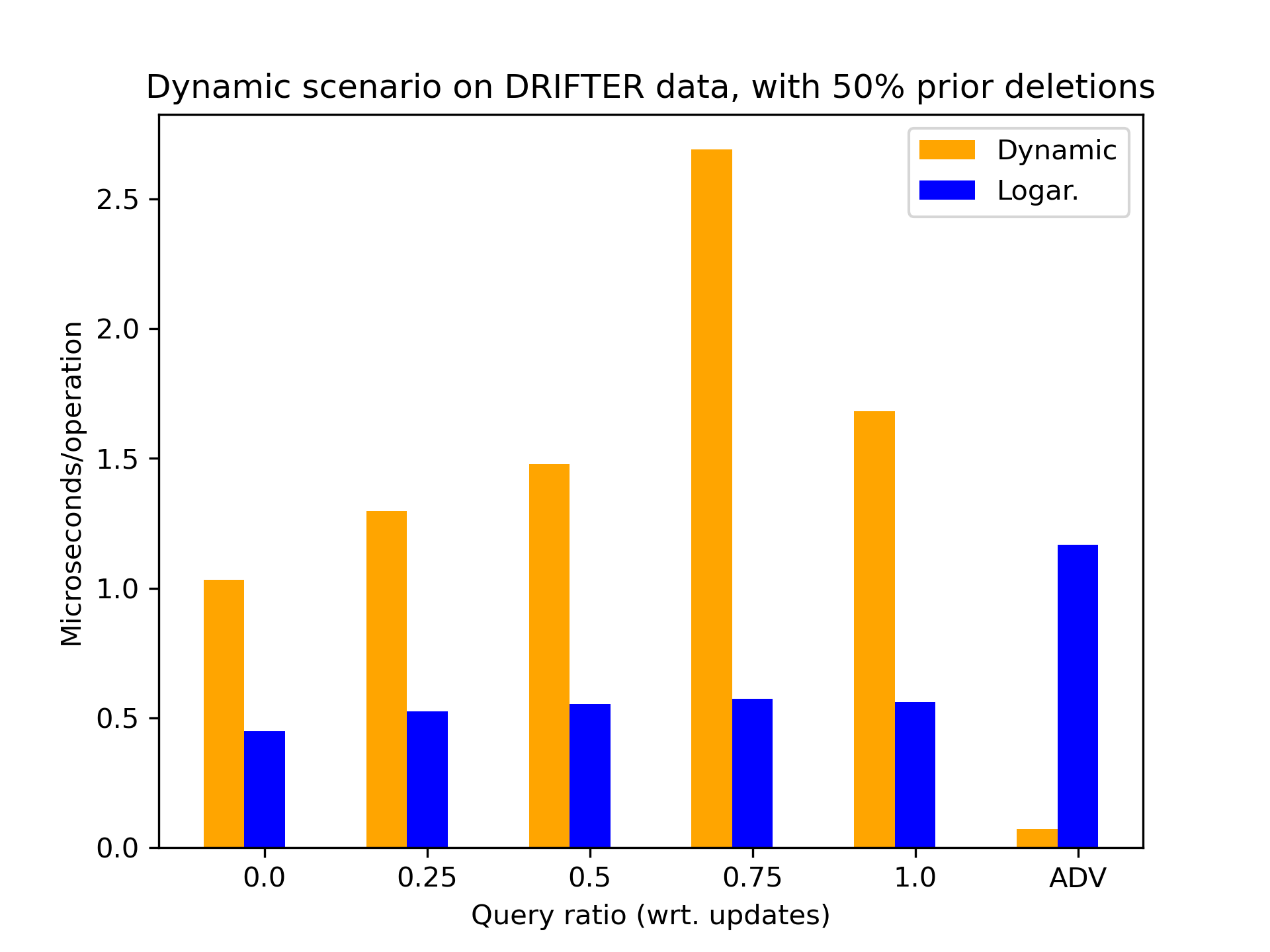}
    \includegraphics[width=0.32\linewidth]{Figures/Results/Updates/drifter_updates0.75.png}
    \caption{Time per operation in dynamic scenario with varying query ratio and prior deletions on \textsc{Drifter} data}
    \label{fig:updates_drifter}
\end{figure}

\begin{figure}
    \centering
    \includegraphics[width=0.32\linewidth]{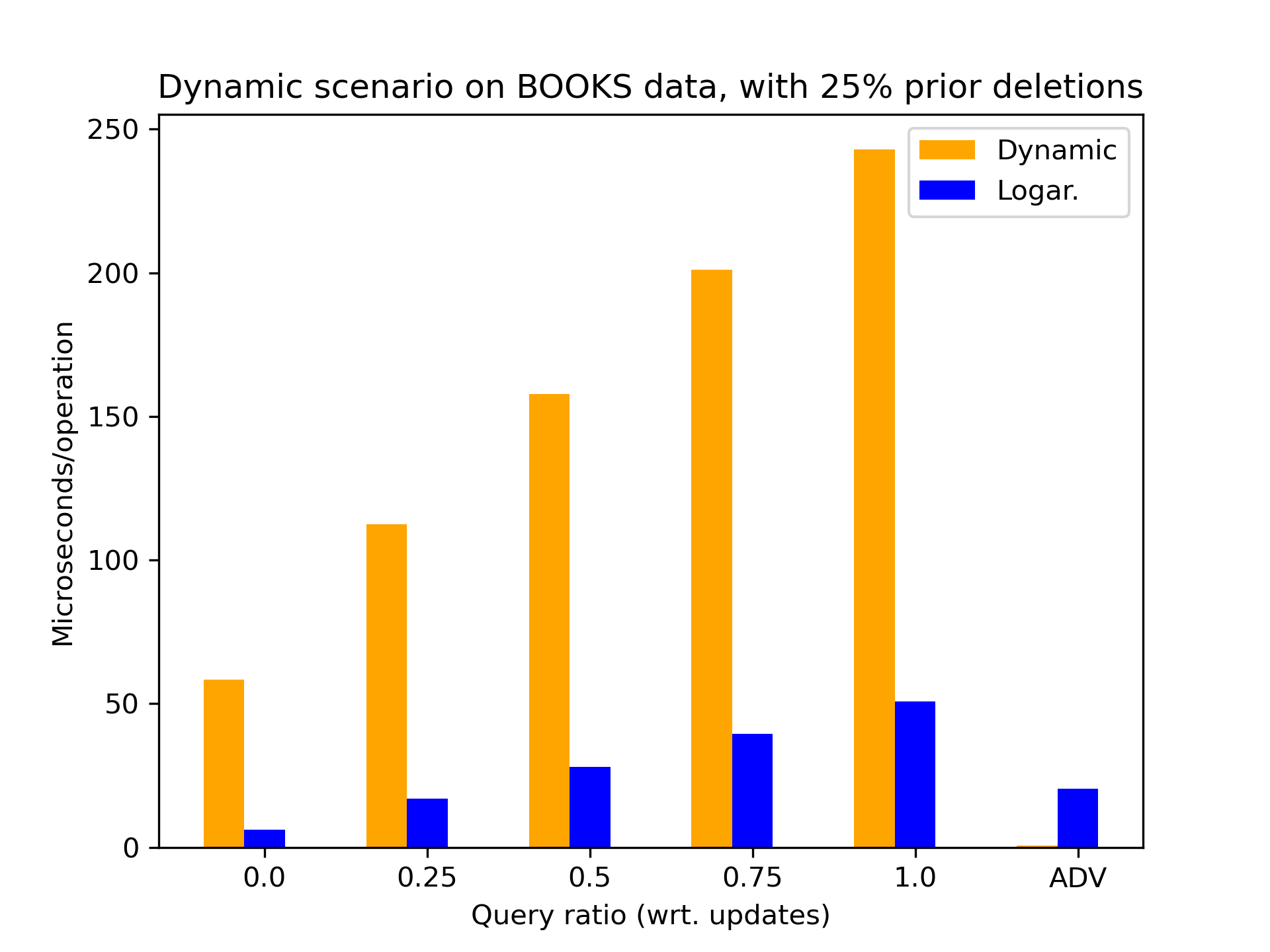}
    \includegraphics[width=0.32\linewidth]{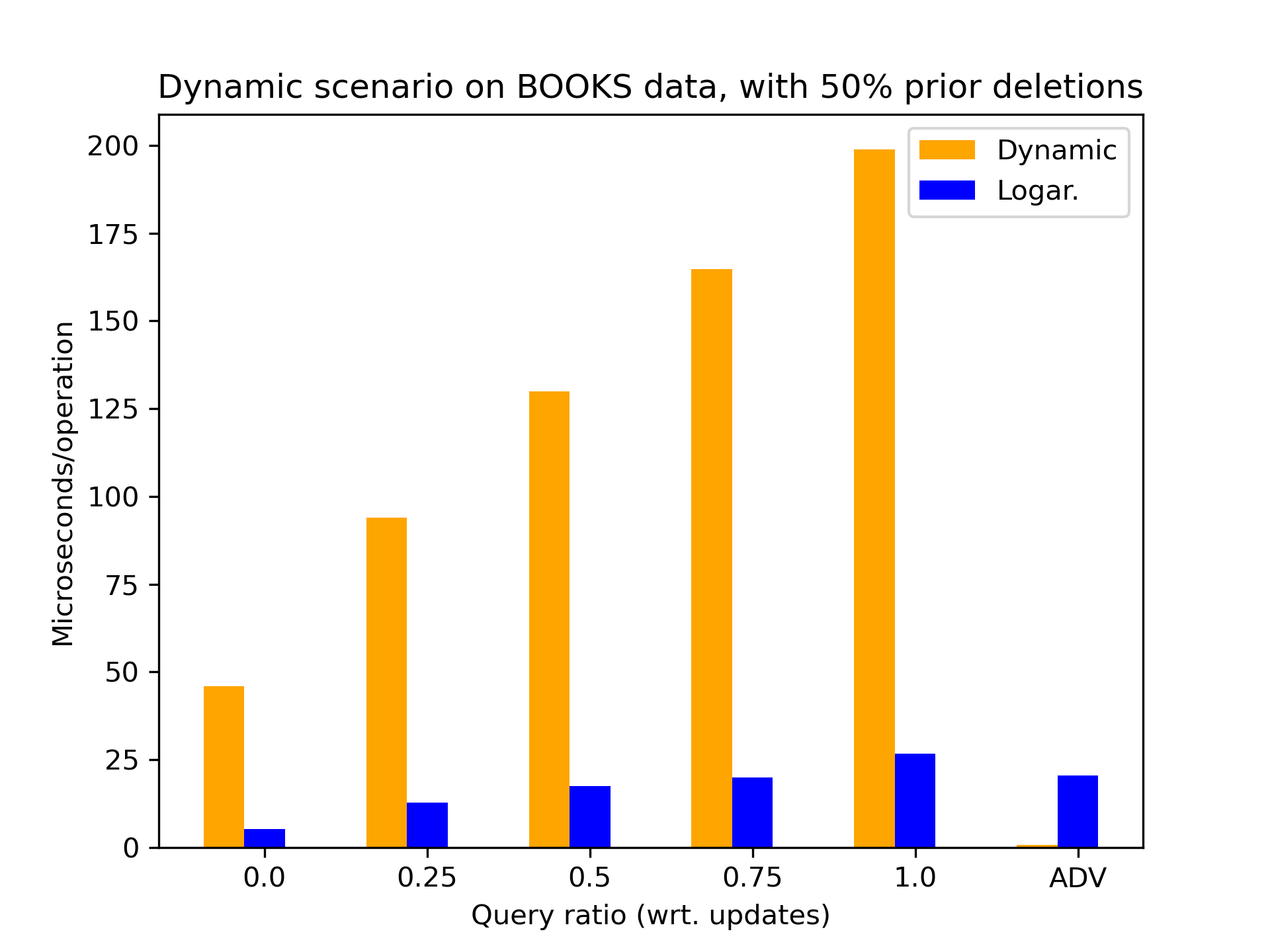}
    \includegraphics[width=0.32\linewidth]{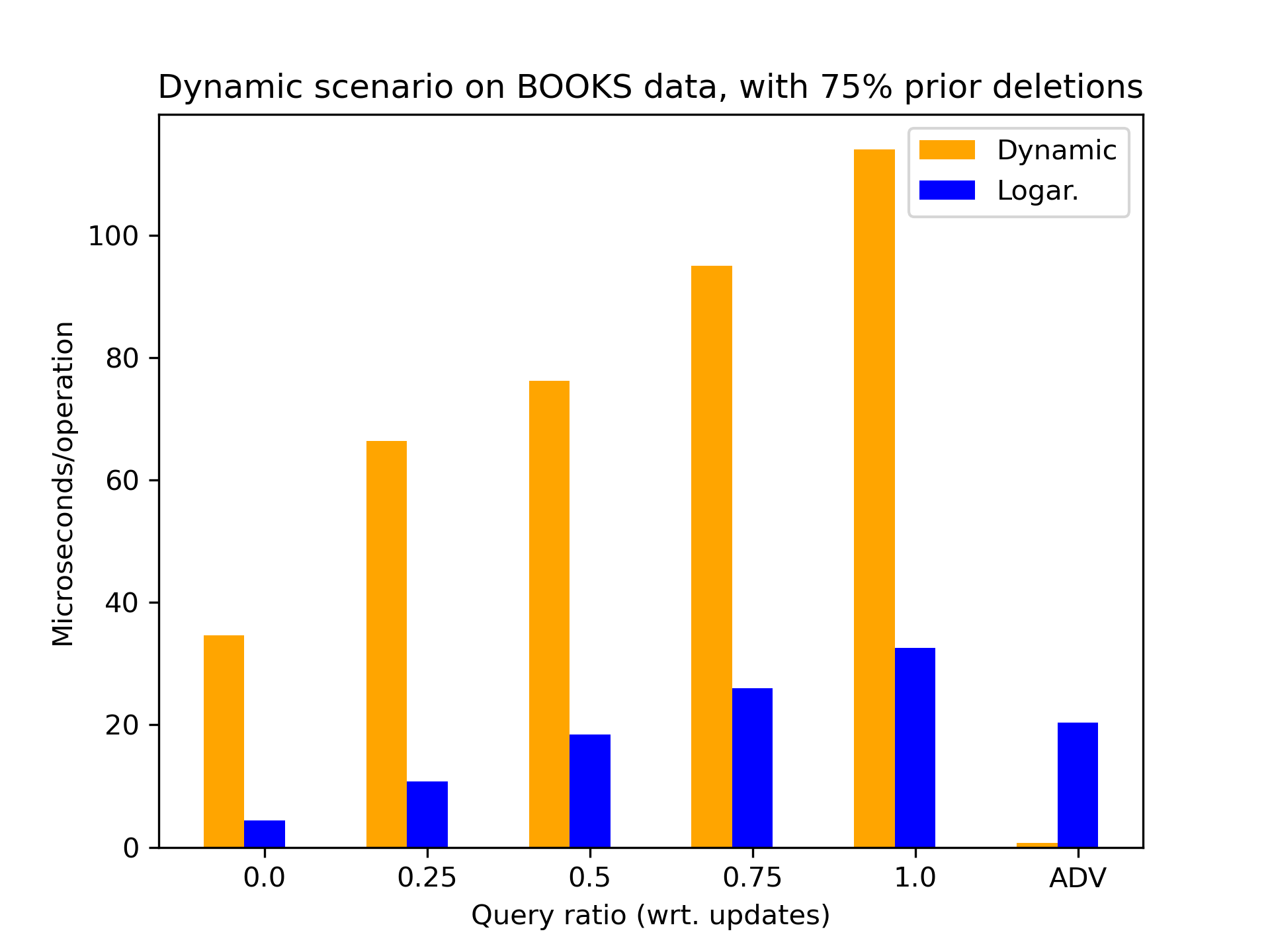}
    \caption{Time per operation in dynamic scenario with varying query ratio and prior deletions on \textsc{Book} data}
    \label{fig:updates_book}
\end{figure}

\begin{figure}
    \centering
    \includegraphics[width=0.32\linewidth]{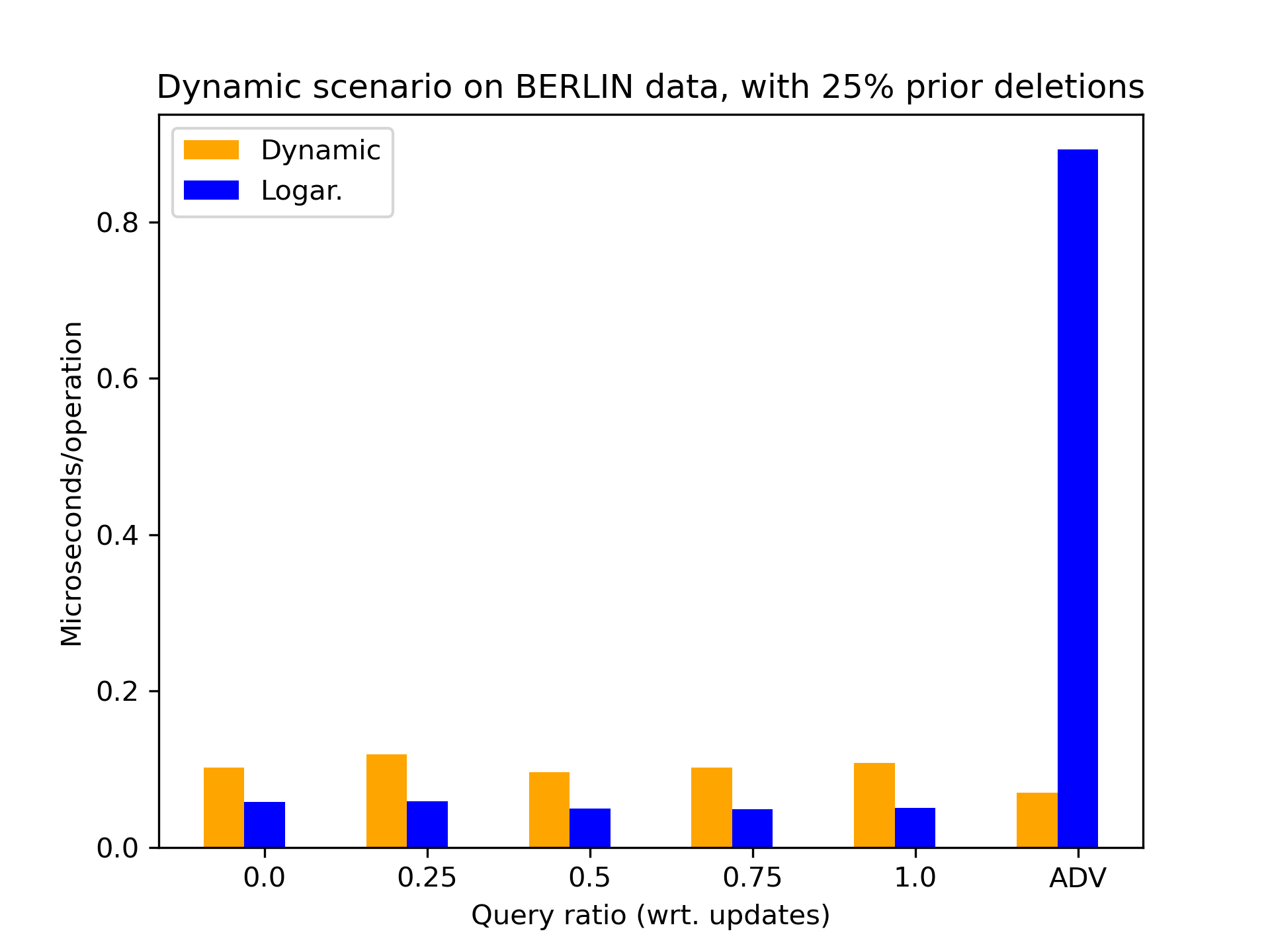}
    \includegraphics[width=0.32\linewidth]{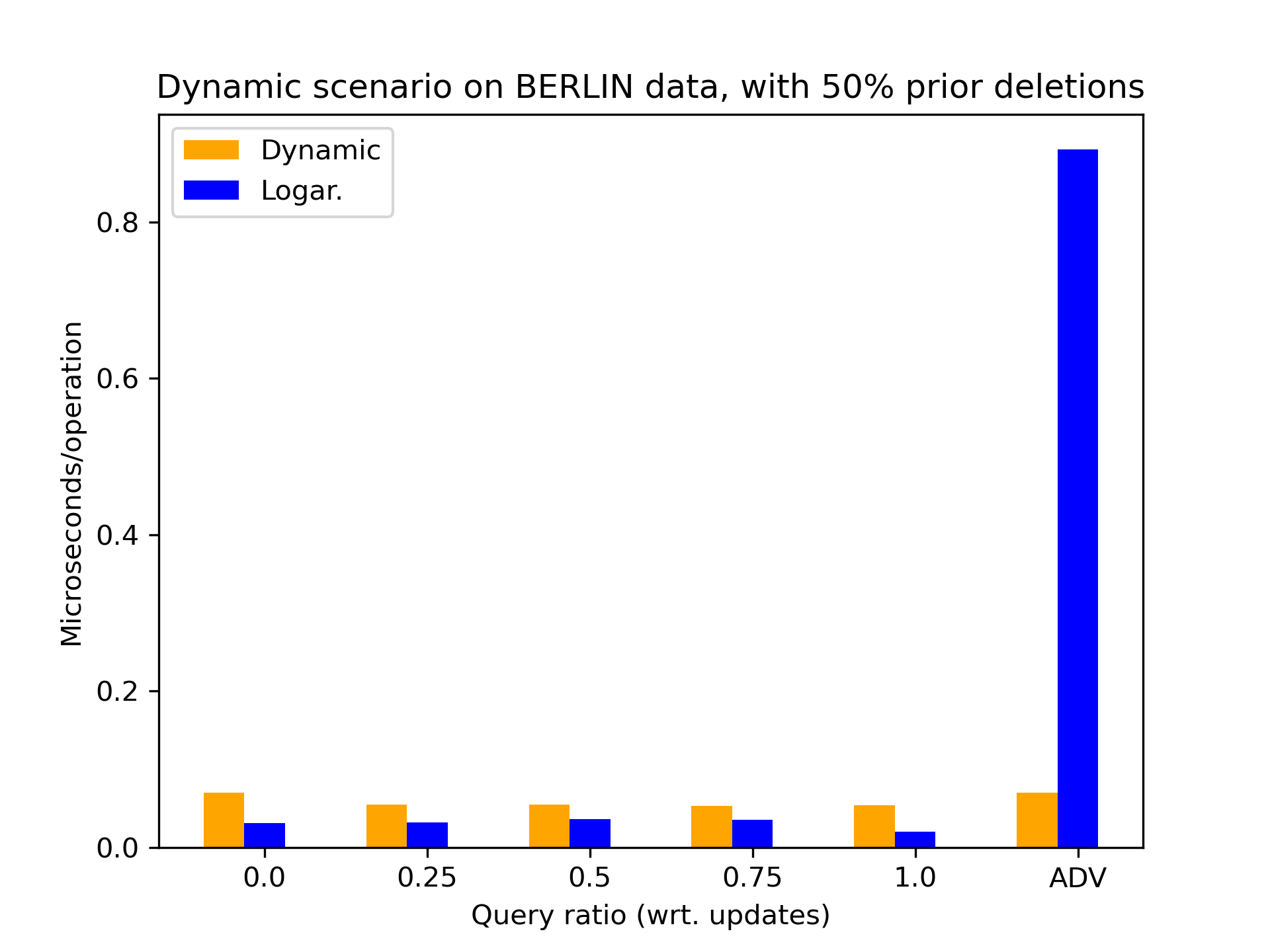}
    \includegraphics[width=0.32\linewidth]{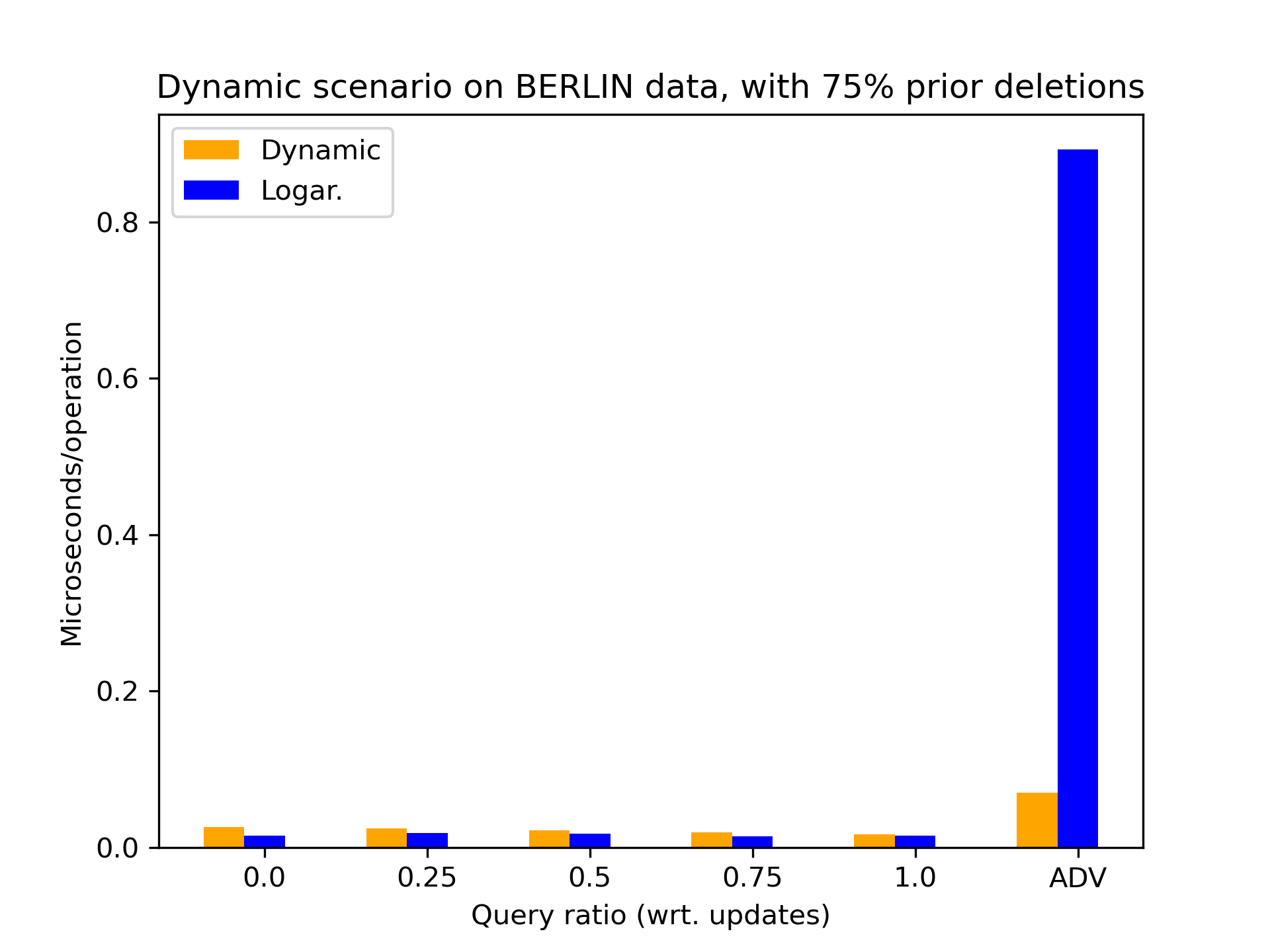}
    \caption{Time per operation in dynamic scenario with varying query ratio and prior deletions on \textsc{Berlin} data}
    \label{fig:updates_berlin}
\end{figure}
\begin{figure}
    \centering
    \includegraphics[width=0.32\linewidth]{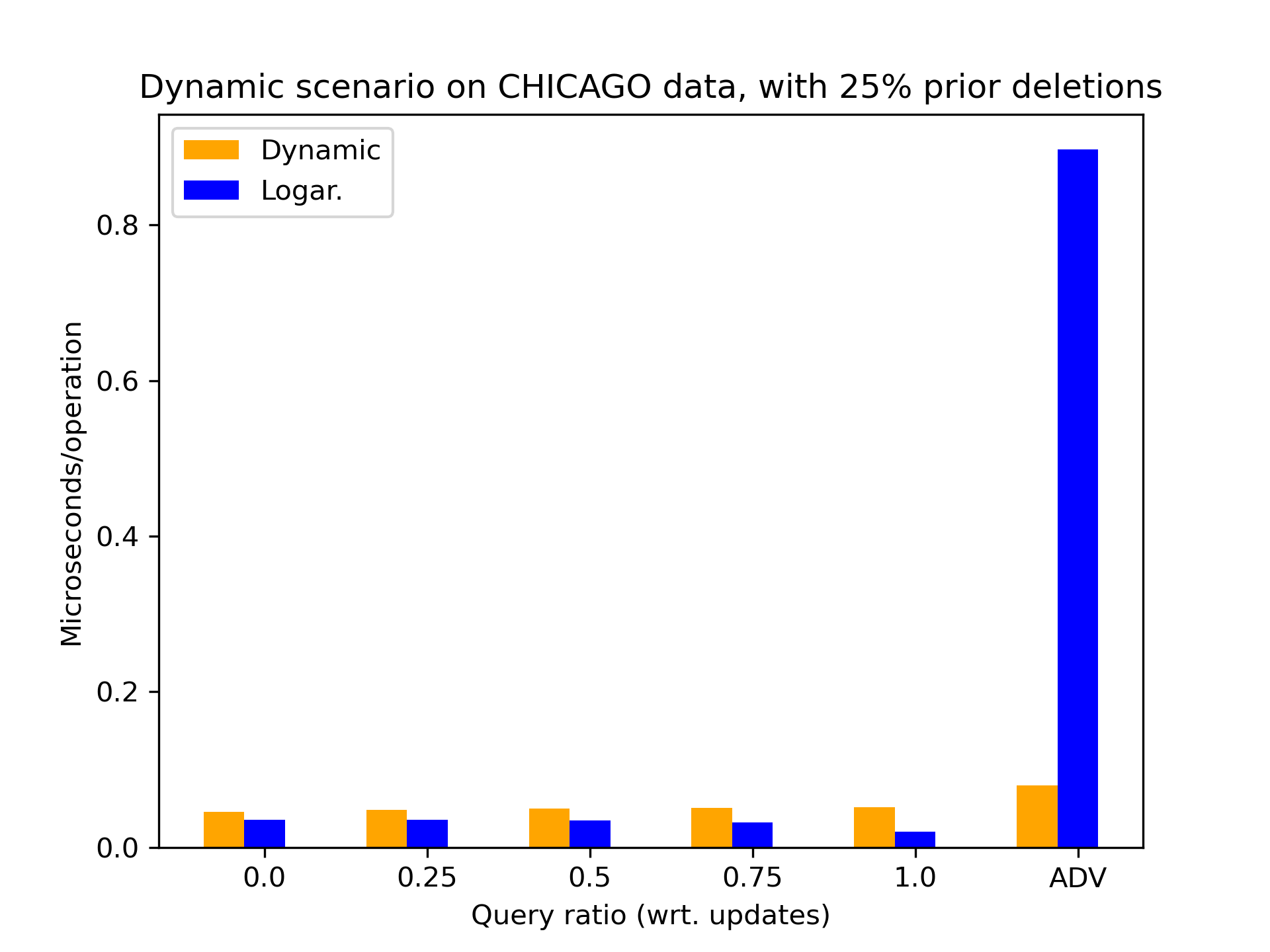}
    \includegraphics[width=0.32\linewidth]{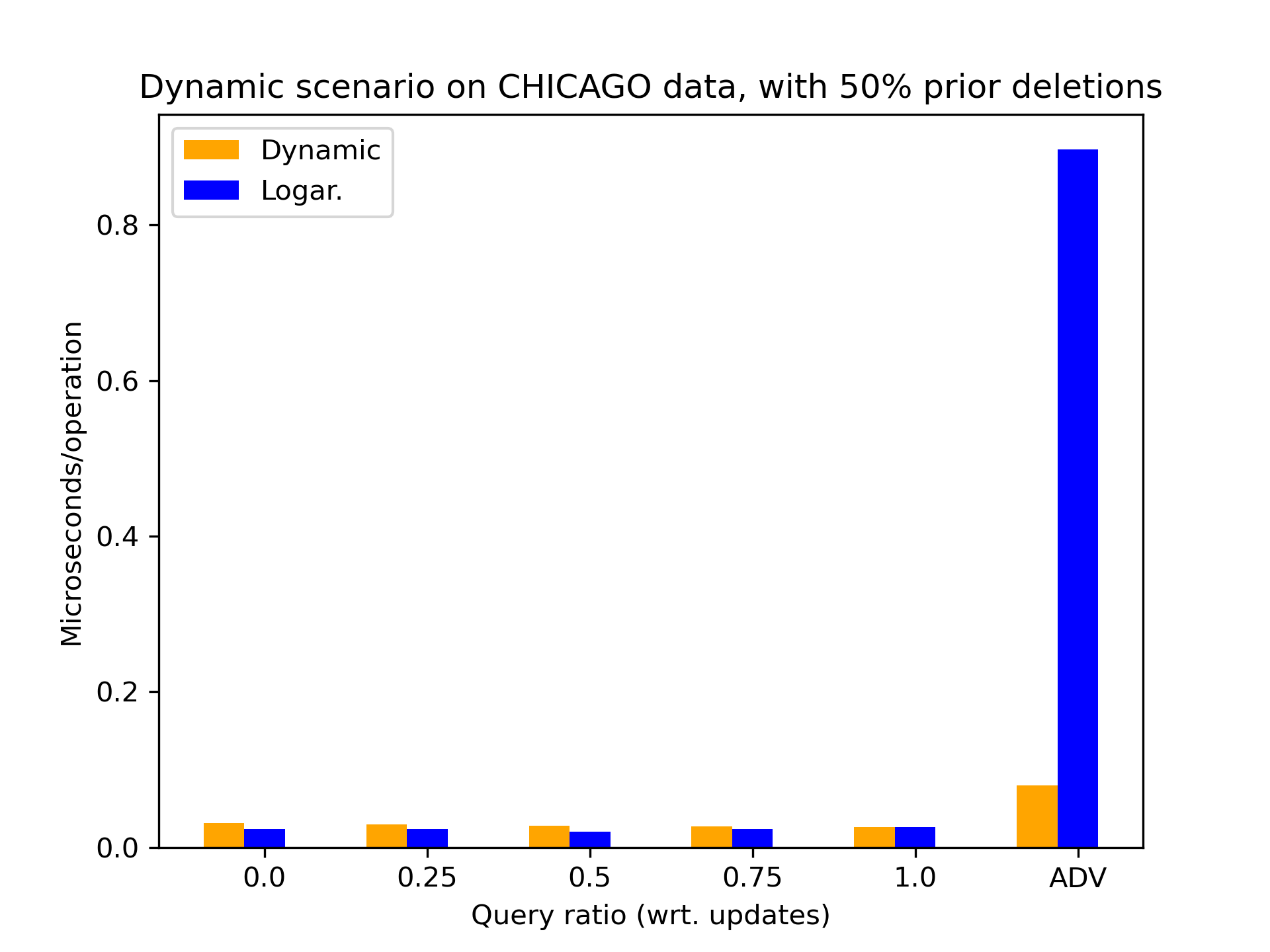}
    \includegraphics[width=0.32\linewidth]{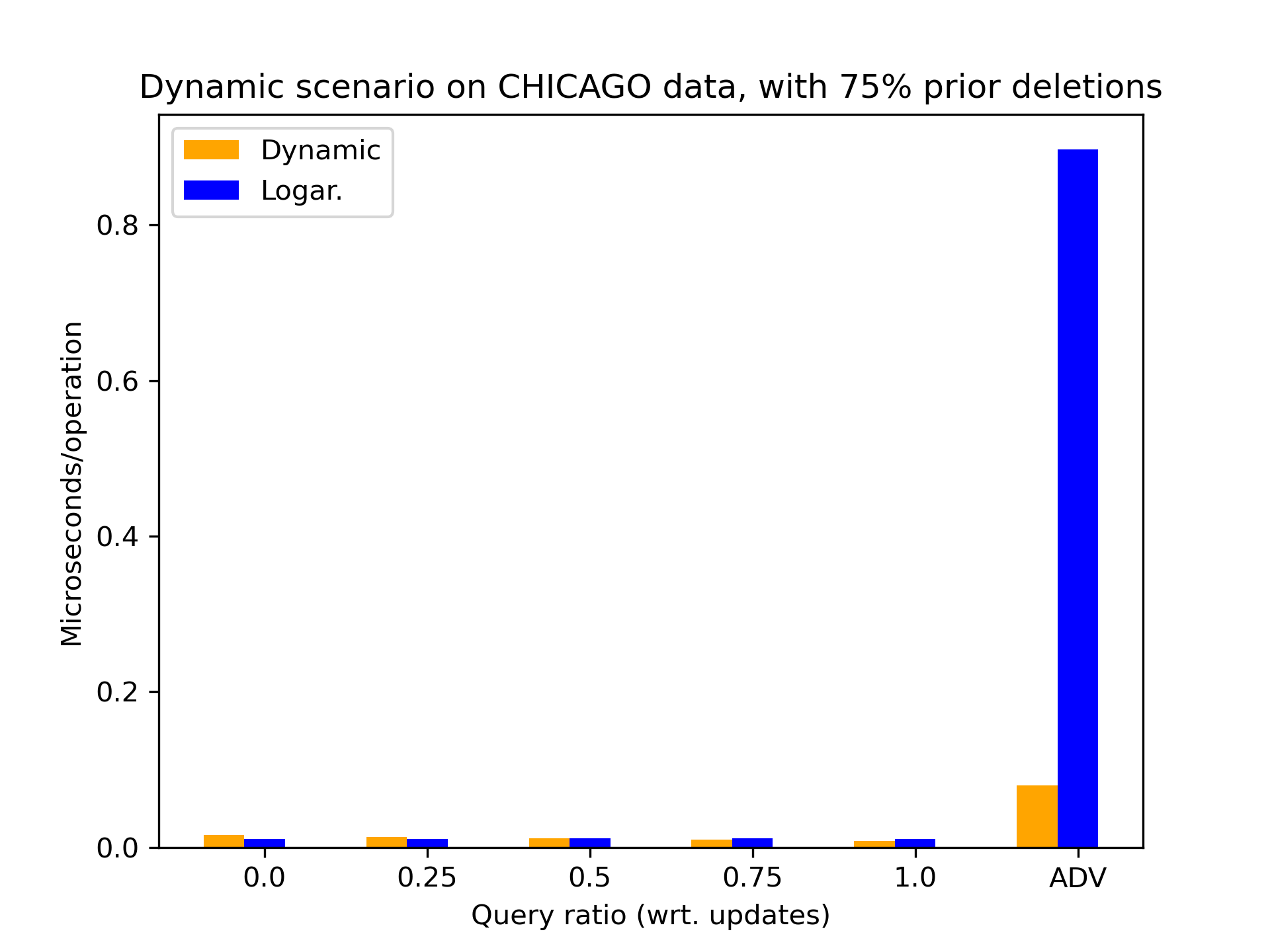}
    \caption{Time per operation in dynamic scenario with varying query ratio and prior deletions on \textsc{Chicago} data}
    \label{fig:updates_chicago}
\end{figure}

\begin{figure}
    \centering
    \includegraphics[width=0.32\linewidth]{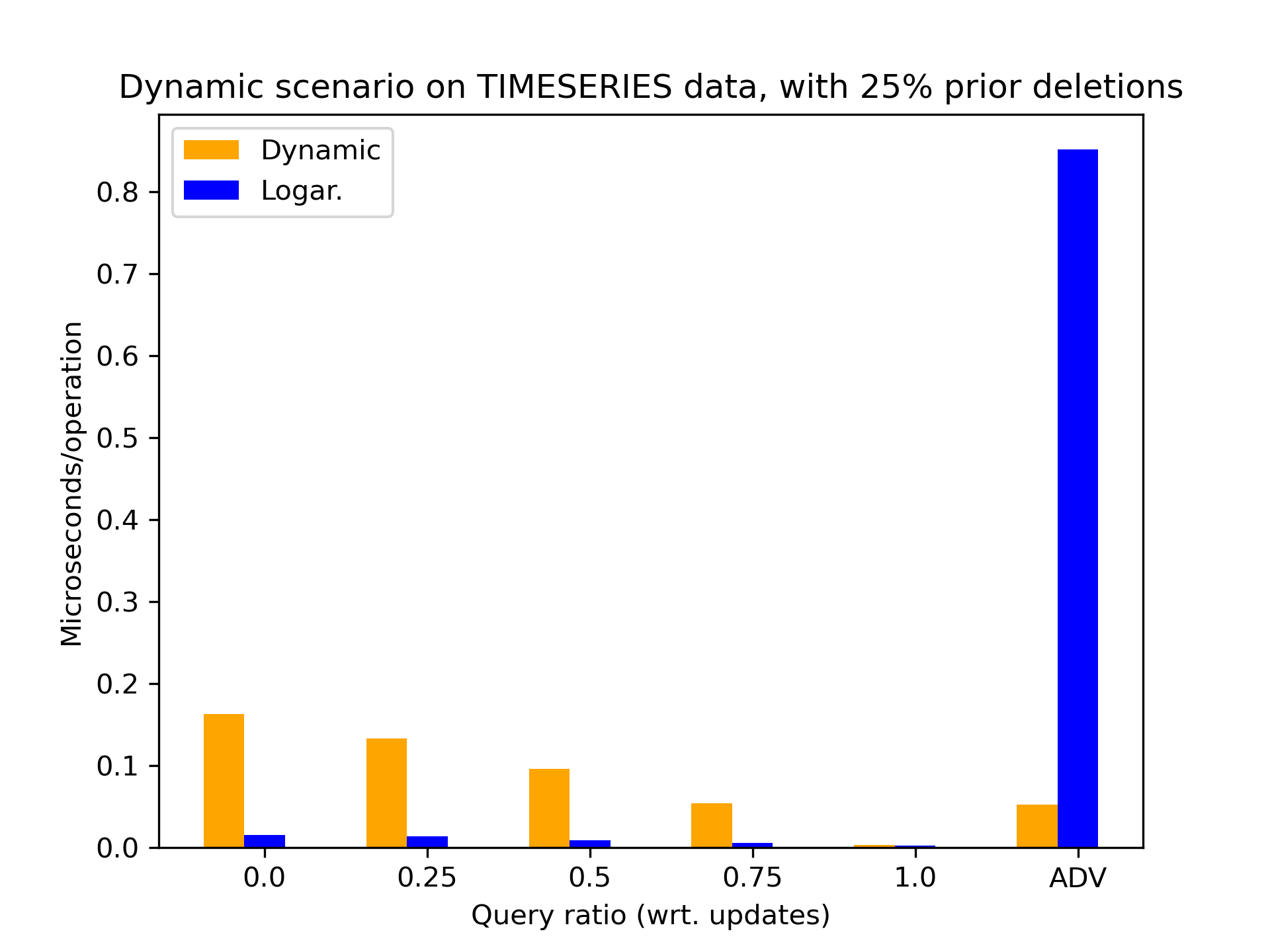}
    \includegraphics[width=0.32\linewidth]{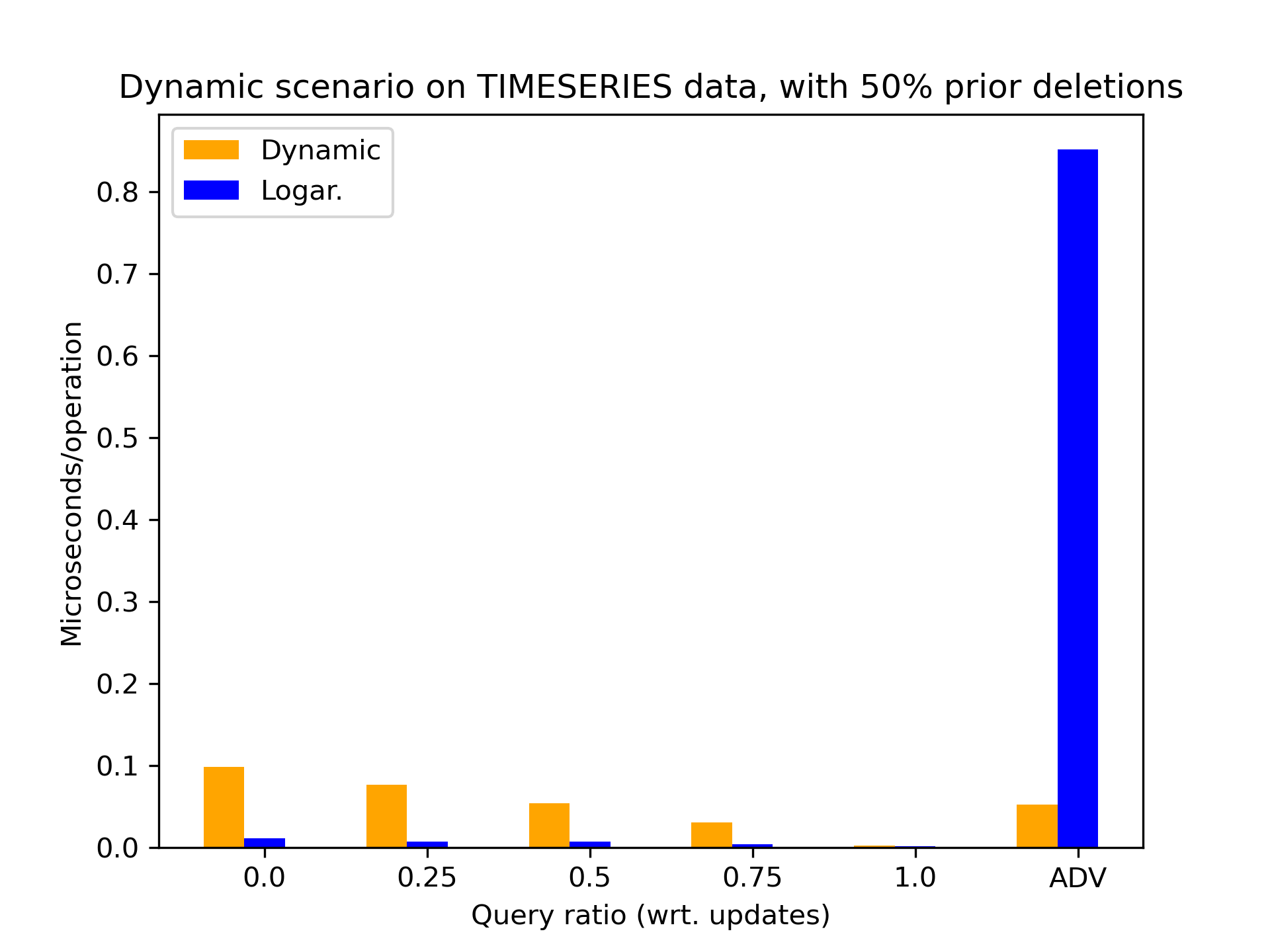}
    \includegraphics[width=0.32\linewidth]{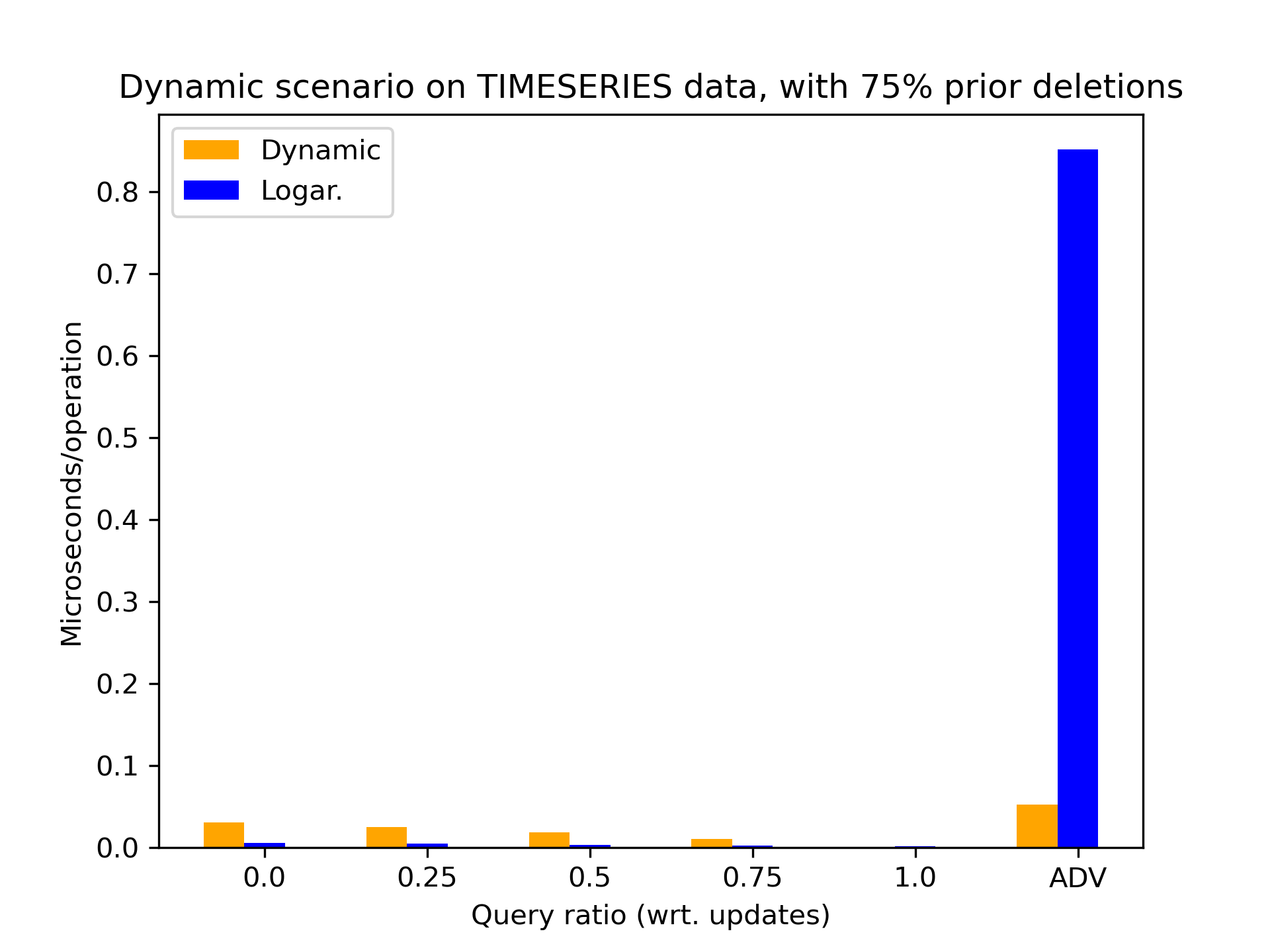}
    \caption{Time per operation in dynamic scenario with varying query ratio and prior deletions on \textsc{Timeseries} data}
    \label{fig:updates_m4}
\end{figure}
\end{document}